\documentclass[nospthms]{svjour3}                     
\def\makeheadbox{\relax}%
\smartqed  

\usepackage{graphicx}
\usepackage[utf8]{inputenc}
\usepackage[pagebackref=true,linktoc=page,hidelinks]{hyperref}
\usepackage[square,numbers,sort]{natbib} 
\usepackage{xcolor}
\usepackage{mathptmx}      
\usepackage{paralist}
\usepackage{amsmath}
\usepackage{amssymb}
\usepackage{booktabs}
\usepackage{cleveref}
\Crefname{algocf}{Algorithm}{Algorithms}

\usepackage{verbatim}
\usepackage{subcaption}
\usepackage{bbding}
\usepackage{numprint}
\npdecimalsign{.}
\usepackage{gensymb}

\usepackage{xspace}
\usepackage{paralist}
\usepackage{lastpage}
\usepackage[shortlabels]{enumitem}

\usepackage{tikz}
\usetikzlibrary{shapes,positioning,arrows,shadows}

\usepackage[binary-units = true,per-mode=fraction]{siunitx} 

\usepackage{nicefrac}	

\usepackage{mfirstuc} 
\usepackage{algorithm2e}

\usepackage{etoolbox}

\setdefaultenum{(i)}{}{}{}

\usepackage{framed} 
\definecolor{shadecolor}{rgb}{.67,0.85,0.90}
\newenvironment{relaxedbox}{\begin{shaded*}}{\end{shaded*}}

\renewcommand*{\backref}[1]{}
\renewcommand*{\backrefalt}[4]{%
    \ifcase #1 %
    \relax 
    \or
    (cited on p.~#2).%
    \else
    (cited on pp.~#2).%
    \fi%
}

\journalname{Software and Systems Modeling}
\newcommand{\myTitle}{A Systematic Approach to Constructing Families of Incremental Topology Control Algorithms Using Graph Transformation}
\newcommand{\myShortTitle}{Constructing Families of Incremental Topology Control Algorithms Using GT}

\newcommand{\wrt}{w.r.t.\xspace}
\newcommand{\eg}{e.g.\xspace}
\newcommand{\idest}{i.e.\xspace}
\newcommand{\vs}{vs.\xspace}

\newcommand{\reftofig}[1]{\textsf{#1}}
\newcommand{\hide}[1]{\textcolor{white}{#1}}

\newcommand{\shortcomingGap}[0]{S1\xspace}
\newcommand{\shortcomingVariability}[0]{S2\xspace}

\newcommand{\MaxpowerTC}{Maxpower \TC}
\newcommand{\ktc}[0]{kTC\xspace}
\newcommand{\ktcParameterK}{\ensuremath{k}\xspace}
\newcommand{\ektc}[0]{\mbox{e-kTC}\xspace}
\newcommand{\lktc}[0]{\mbox{l*-kTC}\xspace}
\newcommand{\topology}[0]{\ensuremath{G}\xspace}
\newcommand{\ektcTopology}[0]{\ensuremath{G_{\ektc}}\xspace}
\newcommand{\maxpowerTopology}[0]{\ensuremath{G_{\text{maxpower}}}\xspace}
\newcommand{\TC}[0]{TC\xspace}
\newcommand{\WSNs}[0]{WSNs\xspace}
\newcommand{\WSN}[0]{WSN\xspace}
\newcommand{\CE}[0]{context event\xspace}
\newcommand{\CEs}[0]{context events\xspace}

\newcommand{\nodeName}[1]{\ensuremath{n_{\text{#1}}\xspace}}
\newcommand{\nodeNameLong}[1]{node~\nodeName{#1}\xspace}

\newcommand{\linkName}[1]{\ensuremath{e_{#1}\xspace}}

\newcommand{\ACT}[0]{\texttt{Active}\xspace}
\newcommand{\INACT}[0]{\texttt{Inactive}\xspace}
\newcommand{\UNCL}[0]{\texttt{Unclassified}\xspace}

\newcommand{\constraint}[1]{\ensuremath{C_{#1}}\xspace}
\newcommand{\constraintLong}[1]{constraint~\constraint{y}\xspace}

\newcommand{\inactiveLinkConstraintKTC}[0]{\ensuremath{\hyperref[fig:generic-graph-constraints-inactive]{C_{\textrm{i}}}}\xspace}
\newcommand{\inactiveLinkConstraintKTCLong}[0]{in\-active-link constraint \inactiveLinkConstraintKTC}
\newcommand{\InactiveLinkConstraintKTCLong}[0]{In\-active-Link Constraint \inactiveLinkConstraintKTC}

\newcommand{\activeLinkConstraintKTC}[0]{\ensuremath{\hyperref[fig:generic-graph-constraints-active]{C_{\textrm{a}}}}\xspace}
\newcommand{\activeLinkConstraintKTCLong}[0]{active-link constraint \activeLinkConstraintKTC}
\newcommand{\ActiveLinkConstraintKTCLong}[0]{Active-Link Constraint \activeLinkConstraintKTC}

\newcommand{\unclassifiedLinkConstraint}[0]{\ensuremath{\hyperref[fig:unclassified-link-constraint]{C_{\textrm{u}}}}\xspace}
\newcommand{\unclassifiedLinkConstraintLong}[0]{no-un\-classi\-fied-links constraint \unclassifiedLinkConstraint}
\newcommand{\premiseUnclassifiedLinkConstraint}{\ensuremath{p_{\textrm{no-u}}}\xspace}

\newcommand{\noParallelLinksConstraint}[0]{\ensuremath{C_{\text{no-par-links}}}\xspace}

\newcommand{\noLoopsConstraint}[0]{\ensuremath{C_{\text{no-loops}}}\xspace}

\newcommand{\premiseInactiveLinkConstraint}[0]{\ensuremath{p_{\textrm{i}}}\xspace}

\newcommand{\premiseActiveLinkConstraint}[0]{\ensuremath{p_{\textrm{a}}}\xspace}
\newcommand{\premiseActiveLinkConstraintLong}{premise~\premiseActiveLinkConstraint}

\newcommand{\conclusionInactiveLinkConstraint}[0]{\ensuremath{c_{\textrm{i}}}\xspace}

\newcommand{\match}[1]{\ensuremath{m\ifstrempty{#1}{}{(#1)}}\xspace}
\newcommand{\linkOrder}[1]{\ensuremath{\prec_{#1}}\xspace}

\newcommand{\pattern}[0]{\ensuremath{p}\xspace}
\newcommand{\GT}[0]{GT\xspace}

\newcommand{\AC}[1]{\ensuremath{\text{AC}_{#1}}\xspace}

\newcommand{\NAC}[1]{\ensuremath{\text{NAC}_{#1}}\xspace}
\newcommand{\NACs}[0]{NACs\xspace}

\newcommand{\NACia}[1]{\ensuremath{\text{NAC}_{\textrm{i},\textrm{a},#1}}\xspace}
\newcommand{\NACaa}[1]{\ensuremath{\text{NAC}_{\textrm{a},\textrm{a},#1}}\xspace}

\newcommand{\PACii}[1]{\ensuremath{\text{PAC}_{\textrm{i},\textrm{i},#1}}\xspace}

\newcommand{\PACui}[1]{\ensuremath{\text{PAC}_{\textrm{u},\textrm{i},#1}}\xspace}

\newcommand{\PAC}[1]{\ensuremath{\text{PAC}_{#1}}\xspace}

\newcommand{\LHS}[1]{\ensuremath{\text{LHS}_{#1}}\xspace}
\newcommand{\RHS}[1]{\ensuremath{\text{RHS}_{#1}}\xspace}
\newcommand{\RHSActivationRule}[0]{\RHS{\textrm{a}}}
\newcommand{\RHSInactivationRule}[0]{\RHS{\textrm{i}}}

\newcommand{\guardSuccess}[0]{\reftofig{[Success]}\xspace}
\newcommand{\guardFailure}[0]{\reftofig{[Failure]}\xspace}

\newcommand{\GTrule}[1]{\ensuremath{R_{#1}}\xspace}
\newcommand{\GTruleLong}[1]{\GT rule~\GTrule{#1}}

\newcommand{\GToperation}[1]{\texttt{#1}\xspace}

\newcommand{\activationRule}[0]{\ensuremath{R_{\textrm{a}}}\xspace}
\newcommand{\activationRuleLong}[0]{activation rule~\activationRule}
\newcommand{\ActivationRuleLong}[0]{Activation Rule~\activationRule}

\newcommand{\inactivationRule}[0]{\ensuremath{R_{\textrm{i}}}\xspace}
\newcommand{\inactivationRuleLong}[0]{inactivation rule~\inactivationRule}
\newcommand{\InactivationRuleLong}[0]{Inactivation Rule~\inactivationRule}

\newcommand{\unclassificationRule}[0]{\ensuremath{R_{\textrm{u}}}\xspace}
\newcommand{\unclassificationRuleLong}[0]{un\-clas\-si\-fi\-ca\-tion rule~\unclassificationRule}

\newcommand{\nodeAdditionRule}[0]{\ensuremath{R_{\textrm{+n}}}\xspace}
\newcommand{\nodeAdditionRuleLong}[0]{node addition rule~\nodeAdditionRule}

\newcommand{\nodeRemovalRule}[0]{\ensuremath{R_{\textrm{-n}}}\xspace}
\newcommand{\nodeRemovalRuleLong}[0]{node removal rule~\nodeRemovalRule}

\newcommand{\linkAdditionRule}[0]{\ensuremath{R_{\textrm{+e}}}\xspace}
\newcommand{\linkAdditionRuleLong}[0]{link addition rule~\linkAdditionRule}

\newcommand{\linkRemovalRule}[0]{\ensuremath{R_{\textrm{-e}}}\xspace}
\newcommand{\linkRemovalRuleLong}[0]{link removal rule~\linkRemovalRule}

\newcommand{\attributeModificationRule}[1]{\ensuremath{R_{\textrm{mod-#1}}}\xspace}

\newcommand{\weightModificationRule}[0]{\attributeModificationRule{w}}
\newcommand{\energyModificationRule}[0]{\attributeModificationRule{E}}
\newcommand{\hopCountModificationRule}[0]{\attributeModificationRule{h}}
\newcommand{\latModificationRule}[0]{\attributeModificationRule{lat}}
\newcommand{\longModificationRule}[0]{\attributeModificationRule{long}}

\newcommand{\weightModificationRuleLong}[0]{weight modification rule~\weightModificationRule}
\newcommand{\energyModificationRuleLong}[0]{energy modification rule~\energyModificationRule}
\newcommand{\hopCountModificationRuleLong}[0]{hop count modification rule~\hopCountModificationRule}
\newcommand{\latModificationRuleLong}[0]{latitude modification rule~\latModificationRule}
\newcommand{\longModificationRuleLong}[0]{longitude modification rule~\longModificationRule}

\newcommand{\findUnclassifiedLinkRule}[0]{\ensuremath{R_{\textrm{find-u}}}\xspace}
\newcommand{\findUnclassifiedLinkRuleLong}[0]{un\-classi\-fied-link-iden\-ti\-fi\-ca\-tion rule~\findUnclassifiedLinkRule}

\newcommand{\findClassifiedLinkRule}[0]{\ensuremath{R_{\textrm{find-ai}}}\xspace}
\newcommand{\findClassifiedLinkRuleLong}[0]{classi\-fied-link-iden\-ti\-fi\-ca\-tion rule~\findClassifiedLinkRule}

\newcommand{\LSM}{link state modification\xspace}
\newcommand{\LSMs}{link state modifications\xspace}

\newcommand{\handleUnclassification}[0]{\handlerOperation{\unclassificationRule}}

\newcommand{\leftPattern}[0]{\ensuremath{p_\ell}\xspace}
\newcommand{\rightPattern}[0]{\ensuremath{p_r}\xspace}

\newcommand{\gluing}[2]{\ensuremath{g^{#2}_{#1}}\xspace}

\newcommand{\gluingII}[2]{\gluing{\textrm{i},\textrm{i},#1}{#2}}

\newcommand{\gluingAA}[2]{\gluing{\textrm{a},\textrm{a},#1}{#2}}

\newcommand{\postcondition}[1]{\ensuremath{\text{PC}_{#1}}\xspace}

\newcommand{\postconditionII}[1]{\postcondition{\textrm{i},\textrm{i},#1}}

\newcommand{\premise}[2]{\ensuremath{p^{#2}_{#1}}\xspace}

\newcommand{\conclusion}[2]{\ensuremath{c^{#2}_{#1}}\xspace}

\newcommand{\conclusionII}[2]{\conclusion{\textrm{i},\textrm{i},#1}{#2}}

\newcommand{\nodeVariable}[0]{\ensuremath{n}\xspace}
\newcommand{\nodeVariableBaseStation}[0]{\ensuremath{n_0}\xspace}
\newcommand{\nodeVariableA}[0]{\ensuremath{n_\text{A}}\xspace}
\newcommand{\nodeVariableB}[0]{\ensuremath{n_\text{B}}\xspace}

\newcommand{\nodeVariablea}[0]{\ensuremath{n_\text{a}}\xspace}
\newcommand{\nodeVariableb}[0]{\ensuremath{n_\text{b}}\xspace}
\newcommand{\nodeVariablec}[0]{\ensuremath{n_\text{c}}\xspace}
\newcommand{\nodeVariablez}[0]{\ensuremath{n_\text{z}}\xspace}
\newcommand{\nodeVariableOne}[0]{\ensuremath{n_1}\xspace}
\newcommand{\nodeVariableTwo}[0]{\ensuremath{n_2}\xspace}

\newcommand{\linkVariable}[0]{\linkName{}\xspace}
\newcommand{\linkVariableAB}[0]{\linkName{\text{AB}}\xspace}
\newcommand{\linkVariableab}[0]{\linkName{\text{ab}}\xspace}
\newcommand{\linkVariableac}[0]{\linkName{\text{ac}}\xspace}
\newcommand{\linkVariablecb}[0]{\linkName{\text{cb}}\xspace}
\newcommand{\linkVariablebc}[0]{\linkName{\text{bc}}\xspace}

\newcommand{\linkVariablecd}[0]{\linkName{\text{cd}}\xspace}
\newcommand{\linkVariableyz}[0]{\linkName{\text{yz}}\xspace}
\newcommand{\linkVariableOneTwo}[0]{\linkName{\text{12}}\xspace}
\newcommand{\linkVariableOneThree}[0]{\linkName{\text{13}}\xspace}
\newcommand{\linkVariableOneFour}[0]{\linkName{\text{14}}\xspace}

\newcommand{\linkVariableThreeOne}[0]{\linkName{\text{31}}\xspace}
\newcommand{\linkVariableTwoThree}[0]{\linkName{\text{23}}\xspace}
\newcommand{\linkVariableThreeTwo}[0]{\linkName{\text{32}}\xspace}
\newcommand{\linkVariableFourThree}[0]{\linkName{\text{43}}\xspace}

\newcommand{\handlerOperation}[1]{\GToperation{handle-#1}}

\newcommand{\violationIdentificationRule}[1]{\GTrule{\text{find-violation-}#1}}

\newcommand{\OK}[0]{\checkmark} 
\newcommand{\notOK}{\ensuremath{\times}}
\newtheorem{theorem}{Theorem}
\newenvironment{proof-sketch}{\noindent{\emph{Sketch of Proof.}}}{\qed}
\newenvironment{proof}{\noindent{\emph{Proof.}}}{\qed}

\newcommand{\coneCount}[0]{\ensuremath{n_{\text{cone}}}\xspace}
\newcommand{\naturalNumbers}[0]{\ensuremath{\mathbb{N}}\xspace}

\DeclareMathOperator{\atan}{atan}

\DeclareMathOperator{\stateOp}{s}
\newcommand{\state}[1]{\ensuremath{\stateOp(#1)}\xspace}

\DeclareMathOperator{\idOp}{id}
\newcommand{\id}[1]{\ensuremath{\idOp(#1)}\xspace}

\DeclareMathOperator{\hopcountOp}{hops}
\newcommand{\hopcount}[1]{\ensuremath{\hopcountOp\ifstrempty{#1}{}{(#1)}}\xspace}

\DeclareMathOperator{\weightOp}{w}
\newcommand{\weight}[1]{\ensuremath{\weightOp\ifstrempty{#1}{}{(#1)}}\xspace}
\newcommand{\weightsquared}[1]{\ensuremath{\weightOp^2(#1)}\xspace}
\newcommand{\weightThreshold}[0]{\ensuremath{w_{\text{thres}}}\xspace}

\DeclareMathOperator{\longitudeOp}{long}
\newcommand{\longitude}[1]{\ensuremath{\longitudeOp(#1)}\xspace}

\DeclareMathOperator{\latitudeOp}{lat}
\newcommand{\latitude}[1]{\ensuremath{\latitudeOp(#1)}\xspace}

\DeclareMathOperator{\angleOp}{\alpha}
\newcommand{\angleCmd}[1]{\ensuremath{\angleOp(#1)}\xspace}

\DeclareMathOperator{\energyOp}{E}
\newcommand{\energy}[2]{\ensuremath{\energyOp_{#2}\ifstrempty{#1}{}{(#1)}}\xspace}

\newcommand{\cardinality}[1]{\ensuremath{\left\vert{#1}\right\vert}\xspace}

\DeclareMathOperator{\lifetimeOp}{L}
\DeclareMathOperator{\powerOp}{P}
\newcommand{\remainingLifetime}[2]{\ensuremath{\lifetimeOp_{#1}\ifstrempty{#2}{}{(#2)}}\xspace}
\newcommand{\remainingLifetimeLong}[2]{remaining $#1$-lifetime~\remainingLifetime{#1}{#2}}

\newcommand{\remainingLifetimePct}[2]{\ensuremath{\lifetimeOp_{\SI{#1}{\percent}}\ifstrempty{#2}{}{(#2)}}\xspace}
\newcommand{\remainingLifetimePctLong}[2]{remaining \SI{#1}{\percent}-lifetime~\remainingLifetimePct{#1}{#2}}

\newcommand{\expectedPower}[2]{\ensuremath{\widehat{\powerOp}_{#1}\ifstrempty{#2}{}{(#2)}}\xspace}
\newcommand{\expectedPowerLong}[2]{expected transmission power~\expectedPower{#1}{#2}}

\newcommand{\expectedRemainingLifetime}[2]{\ensuremath{\widehat{\lifetimeOp}_{#1}\ifstrempty{#2}{}{(#2)}}\xspace}
\newcommand{\expectedRemainingLifetimeLong}[2]{expected \ifstrempty{#1}{}{$#1$-}lifetime~\expectedRemainingLifetime{#1}{#2}}

\DeclareMathOperator{\completeClassificationPredicate}{\phi_{\text{CC}}}
\DeclareMathOperator{\AConnectivityPredicate}{\phi_{\text{A-conn}}}

\DeclareMathOperator{\minWeightPredicateOp}{\phi_{\text{min-weight}}}
\newcommand{\minWeightPredicate}[1]{\ensuremath{\minWeightPredicateOp\ifstrempty{#1}{}{(#1)}}\xspace}
\newcommand{\minWeightPredicateLong}[1]{minimum-weight predicate~\minWeightPredicate{#1}}

\DeclareMathOperator{\tieBreakingPredicateOp}{\phi_{\text{tie-break}}}
\newcommand{\tieBreakingPredicate}[1]{\ensuremath{\tieBreakingPredicateOp\ifstrempty{#1}{}{(#1)}}\xspace}
\newcommand{\tieBreakingPredicateLong}[1]{tie-breaking predicate~\tieBreakingPredicate{#1}}

\DeclareMathOperator{\trianglePredicateOp}{\phi_\Delta}
\newcommand{\trianglePredicate}[1]{\ensuremath{\trianglePredicateOp\ifstrempty{#1}{}{(#1)}}\xspace}

\DeclareMathOperator{\tcPredicateOp}{\phi}
\newcommand{\tcPredicate}[2]{\ensuremath{\tcPredicateOp_{#1}\ifstrempty{#2}{}{(#2)}}\xspace}

\newcommand{\ktcPredicate}[1]{\tcPredicate{\;\text{\ktc}}{#1}}
\newcommand{\ktcMinWeightPredicate}[1]{\tcPredicate{\;\text{kTC+min-weight}}{#1}}
\newcommand{\ektcPredicate}[1]{\tcPredicate{\text{\ektc}}{#1}}
\newcommand{\lktcPredicate}[1]{\tcPredicate{\;\text{\lktc}}{#1}}
\newcommand{\xtcPredicate}[1]{\tcPredicate{\;\text{XTC}}{#1}}
\newcommand{\xtcPredicateLong}[1]{XTC predicate~\xtcPredicate{#1}}
\newcommand{\rngPredicate}[1]{\tcPredicate{\;\text{RNG}}{#1}}
\newcommand{\ggPredicate}[1]{\tcPredicate{\;\text{GG}}{#1}}
\newcommand{\yaoPredicate}[1]{\tcPredicate{\;\text{Yao}}{#1}}
\newcommand{\maxpowerPredicate}[1]{\tcPredicate{\;\text{Maxpower}}{#1}}

\newcommand{\RQekTC}{RQ1}
\newcommand{\RQekTCLong}{\RQekTC---Performance of \ektc}

\newcommand{\RQMinWeightPredicate}{RQ2}
\newcommand{\RQMinWeightPredicateLong}{\RQMinWeightPredicate---Benefit of Minimum-Weight Predicate}

\newcommand{\eMoflon}[0]{\textsc{eMoflon}\xspace}
\newcommand{\Simonstrator}[0]{\textsc{Simonstrator}\xspace}
\newcommand{\PFS}[0]{\PFSLong}
\newcommand{\PFSLong}[0]{\textsc{PeerfactSim.KOM}\xspace}
\newcommand{\pct}[1]{\SI{#1}{\percent}\xspace}

\def\evalDataRoot{figures/evaluation/sosym_evalrun_2016-06-29T144307/output}
\def\rootOfAliveNodePlots{\evalDataRoot/alive-nodes_vs_simtime_comparison}

\newcommand{\aliveNodesSuffix}[4]{worldSize=#1_nodeCount=#2_kTCParK=#3_minDistInM=#4}

\newcommand{\evalParSeedCount}{\numprint{5}\xspace}
\newcommand{\evalParK}{\numprint{1.41}\xspace}
\newcommand{\evalParTCInterval}{\SI{10}{\minute}\xspace}
\newcommand{\evalParSimulationDuration}{\SI{25}{\hour}\xspace}
\newcommand{\evalParBatterySource}{\SI{130}{\joule}\xspace}

\newcommand{\evalParMeasurementsPerRun}{\numprint{150}\xspace}

\newcommand{\evalParTransmissionRange}{\SI{130}{\meter}\xspace}
\newcommand{\evalParNodeCount}{100\xspace}
\newcommand{\evalParNumberOfConfigurations}{2\xspace}
\newcommand{\evalParMsgInterval}{\SI{30}{\second}\xspace}
\newcommand{\evalParMsgSize}{\SI{1}{\kilo\byte}\xspace}

\newcommand{\metricOutdegree}{\ensuremath{d_{\text{out}}}\xspace}
\newcommand{\metricWorldSize}{\ensuremath{a}\xspace}

\def\cfgSmallMedium{\texttt{n100w500}\xspace}
\def\cfgSmallSparse{\texttt{n100w750}\xspace}

\newcommand{\relativeMetric}[2]{\ensuremath{\delta_{#1}\ifstrempty{#2}{}{({#2})}}\xspace}

\newcommand{\topologySize}[1]{\ensuremath{S\ifstrempty{#1}{}{(#1)}}\xspace}
\newcommand{\meanTopologySize}[1]{\ensuremath{\overline{\topologySize{#1}}}}
\newcommand{\relativeTopologySize}[1]{\relativeMetric{\text{Size}}{#1}}

\newcommand{\relativeExecutionTime}[1]{\relativeMetric{T}{#1}}
\newcommand{\executionTime}[1]{\ensuremath{T\ifstrempty{#1}{}{(#1)}}\xspace}
\newcommand{\meanExecutionTime}[1]{\ensuremath{\overline{\executionTime{#1}}}}

\newcommand{\relativeLsmCount}[1]{\relativeMetric{N}{#1}}
\newcommand{\lsmCount}[1]{\ensuremath{N\ifstrempty{#1}{}{(#1)}}}
\newcommand{\meanLsmCount}[1]{\ensuremath{\overline{\lsmCount{#1}}}}

\newcommand{\relativeRemainingLifetime}[2]{\relativeMetric{\remainingLifetime{#1}{}}{#2}}
\newcommand{\relativeRemainingLifetimePct}[2]{\relativeMetric{\remainingLifetimePct{#1}{}}{#2}}

\def\extractFirstInternal#1 #2\relax{#1}
\def\extractSecondInternal#1 #2\relax{#2}
\def\extractFirst#1{\expandafter\extractFirstInternal#1\relax}
\def\extractSecond#1{\expandafter\extractSecondInternal#1\relax}
\newcommand{\bstVal}[1]{\textbf{#1}}
\newcommand{\GwThresh}[0]{\ensuremath{G_{\weightThreshold}}}

\usetikzlibrary{positioning, shadows, shapes, arrows}

\tikzset{main node/.style={circle, draw, fill=black, text=white,font=\sffamily\scriptsize, minimum size = 0.4cm, inner sep= 1pt}}
\tikzset{active/.style={->, line width=0.9pt, font=\scriptsize}}
\tikzset{inactive/.style={->, densely dotted, line width=0.9pt, font=\scriptsize}}
\tikzset{unclassified/.style={->, densely dashed,line width=0.9pt, font=\scriptsize}}
\tikzset{activeOrInactive/.style={->, dashdotted,line width=0.9pt, font=\scriptsize}}
\tikzset{any/.style={->, draw={rgb:black,5;white,6}, line width=0.9pt, font=\scriptsize}}
\tikzset{induction/.style={->, >=latex, line width=0.9pt, font=\scriptsize}}

\tikzset{rcBig/.style={rectangle,draw, fill=white, minimum height=4.1cm,text width=4.6cm, align = center}}
\tikzset{rcGraphSmall/.style={rectangle,draw, fill=white, minimum height=1.1cm,text width=4cm, align = center}}
\tikzset{rcGraph/.style={rectangle,draw, fill=white, minimum height=3.25cm,text width=4cm, align = center}}
\tikzset{rcGraphInner/.style={rectangle,draw, fill=white, minimum height=0.705cm,text width=3.55cm, align = left}}
\tikzset{rcMini/.style={rectangle,draw, fill=white, minimum height=0.1cm,text width=0.55cm, align = left}}
\tikzset{rcMiniMedium/.style={rectangle,draw, fill=white, minimum height=0.1cm,text width=0.95cm, align = left}}
\tikzset{rcMiniLong/.style={rectangle,draw, fill=white, minimum height=0.1cm,text width=1.25cm, align = left}}

\tikzset{rc/.style={rectangle,draw, fill=white, minimum height=2.35cm,text width=3.25cm, align = center}}
\tikzset{rcSmall/.style={rectangle,draw, fill=white, minimum height=0.705cm,text width=3.35cm, align = center}}

\tikzset{rcMedium/.style={rectangle,draw, fill=white, minimum height=3.5cm,text width=3cm, align = center}}
\tikzset{rcGraphMedium/.style={rectangle,draw, fill=white, minimum height=1.9cm,text width=2.6cm, align = center}}

\tikzset{rcShort/.style={rectangle,draw, fill=white, minimum height=1.7cm,text width=2.25cm, align = center}}
\tikzset{rcShortDouble/.style={rectangle,draw, fill=white, minimum height=1.7cm,text width=4.35cm, align = center}}
\tikzset{rcShortGraph/.style={rectangle,draw, fill=white, minimum height=1.1cm,text width=1.9cm, align = center}}
\tikzset{rcLongGraph/.style={rectangle,draw, fill=white, minimum height=2.3cm,text width=1.9cm, align = center}}

\tikzstyle{bigArrow}=[line width=0.9cm,font=\sffamily\small, color={rgb:black,5;white,6}]
\tikzset{bigArrowHead/.style={fill={rgb:black,5;white,6},shape=single arrow,text width=0.35mm,text height=10.5ex}}
\tikzset{sectionSummaryBox/.style={align=center,text width=4.6cm,below}}

\begin{document}

\title{\myTitle%
\thanks{This work has been funded by the German Research Foundation (DFG) as part of projects A01 and C01 within the Collaborative Research Center (CRC) 1053 -- MAKI.
}
}

\titlerunning{\myShortTitle}

\author{Roland Kluge\and Michael Stein \and Gergely Varró \and Andy Schürr \and Matthias Hollick \and Max Mühlhäuser
}

\institute{R. Kluge \at
              Merckstraße 25, 64283 Darmstadt \\
              Tel.: +49-6151-16-22354\\
              Fax: +49-6151-16-22352\\
              \email{roland.kluge@es.tu-darmstadt.de} 
}

\date{Received: 2016-07-04 / Revised: 2016-11-22 / Accepted: 2017-02-16 / Published online: 2017-03-08}

\maketitle
   
\begin{abstract}

\begin{relaxedbox}
\noindent This document corresponds to the accepted manuscript of the article
\emph{Kluge, R., Stein, M., Varró, G., Schürr, A., Hollick, M., Mühlhäuser, M.: "A systematic approach to constructing families of incremental topology control algorithms using graph transformation," in: SoSyM 2017.}
The URL of the formal version is \url{https://doi.org/10.1007/s10270-017-0587-8}.
This document is made available under the CC-BY-NC-ND 4.0 license \url{http://creativecommons.org/licenses/by-nc-nd/4.0/}.
\end{relaxedbox}  

In the communication systems domain, constructing and maintaining network topologies via topology control algorithms is an important cross-cutting research area.
Network topologies are usually modeled using attributed graphs whose nodes and edges represent the network nodes and their interconnecting links.
A key requirement of topology control algorithms is to fulfill certain consistency and optimization properties to ensure a high quality of service.
Still, few attempts have been made to constructively integrate these properties into the development process of topology control algorithms.
Furthermore, even though many topology control algorithms share substantial parts (such as structural patterns or tie-breaking strategies), few works constructively leverage these commonalities and differences of topology control algorithms systematically.
In previous work, we addressed the constructive integration of consistency properties into the development process.
We outlined a constructive, model-driven methodology for designing individual topology control algorithms.
Valid and high-quality topologies are characterized using declarative graph constraints;
topology control algorithms are specified using programmed graph transformation.
We applied a well-known static analysis technique to refine a given topology control algorithm in a way that the resulting algorithm preserves the specified graph constraints.

In this paper, we extend our constructive methodology by generalizing it to support the specification of families of topology control algorithms.
To show the feasibility of our approach, we reneging six existing topology control algorithms and develop \ektc, a novel energy-efficient variant of the topology control algorithm \ktc.
Finally, we evaluate a subset of the specified topology control algorithms using a new tool integration of the graph transformation tool \eMoflon and the \Simonstrator network simulation framework.
\end{abstract}
\section{Introduction}
\label{sec:introduction}

In the communication systems domain, wireless sensor networks (\WSNs)
\cite{Santi2005,Wang08} are a highly active research area. 
For instance, \WSNs are applied to monitor physical or environmental conditions using distributed, autonomous, battery-powered sensor nodes that cooperatively transmit their collected data to a central location.
To improve important properties (\eg, the battery lifetime of these devices), a topology control (\TC) algorithm~\cite{Santi2005} inactivates redundant communication links of a \WSN.
Key requirements on a \TC algorithm are to
\begin{inparaenum}
\item handle continuously changing network topologies,
\item operate in a highly distributed environment, in which each node can only observe and modify its local neighborhood, and
\item guarantee important local and global formal properties (\eg, bounded node degree or connectivity of the topology) for their neighborhood and the whole network, respectively.
\end{inparaenum}

The design and implementation of a \TC algorithm are, therefore, challenging and elaborate tasks, which are typically carried out by highly skilled experts.
The development of a new \TC algorithm is usually an iterative process.
In each iteration, %
\begin{inparaenum}
\item a new variant must be individually designed and implemented for a
distributed environment,
\item the preservation of required formal properties must be proved, and
\item performance measurements must be carried out in a corresponding
runtime environment, which is either a network simulator or a hardware testbed.
\end{inparaenum}
On the one hand, the \emph{specification} of a \TC algorithm often builds on a mathematically well-founded framework (\eg, graph or game theory), which allows the \TC developer to prove formal properties.
On the other hand, the \emph{implementation} of a \TC algorithm is typically written in a general-purpose programming language (\eg, Java for simulation~\cite{RSRS15} or C for testbed evaluation~\cite{Dunkels2004}).
Additionally, the \TC algorithm and the runtime environment (often continuously) interact:
The \TC algorithm activates and inactivates links in the topology, and the runtime environment causes context events (\eg, node removals due to battery depletion).
A \emph{dynamic \TC algorithm} has to handle such context events.
In many application scenarios of \WSNs (\eg, environmental monitoring), these context events are small compared to the size of the entire topology.
Therefore, it is crucial that a dynamic \TC algorithm reacts to context events in an \emph{incremental} manner, \idest, by retaining unaltered parts of the topology as much as possible.

\paragraph{Current shortcomings}
State-of-the-art \TC literature reveals that the traditional development process of \TC algorithms exposes two major shortcomings.
\begin{enumerate}
\item[\shortcomingGap]
A systematic mapping between the specification and the implementation is missing.
This makes it extraordinarily difficult to verify that both representations are indeed equivalent.
Incremental \TC algorithms are considerably more difficult to develop compared to their batch version, which complicates the traceability between specification and implementation even more.

\item[\shortcomingVariability]
Novel \TC algorithms tend to build on former \TC algorithms.
Still, these inherent commonalities and differences of \TC algorithms are not specified systematically.
This reduces reusability among and comparability of \TC algorithms.
Moreover, such a systematic specification could also enable us to prove formal properties not only for individual \TC algorithms but for whole families of \TC algorithms at once.
\end{enumerate} 
Both shortcomings are well-known research challenges in the communication systems domain.
For instance, \shortcomingGap is addressed in~\cite{Kluge2016,Qadir2015,Mori2013,Martins2010,Zave2012,Zave2008,Dohler2007} (see also \Cref{sec:relatedWorkCorrectness}), and \shortcomingVariability is addressed in \cite{Portocarrero2014,Anaya2014,Saller2013,Ortiz2012,Saller2012,Fuentes2011,Anguera2010,Delicato2009,Bencomo2008} (see also \Cref{sec:relatedWorkVariability}).

\paragraph{Previous work on \shortcomingGap}
In~\cite{KVS15,Kluge2016}, we showed that model-driven principles~\cite{Beydeda2005}, as applied in many success stories
\cite{VSBHH13,Hermann2013}, are suitable to address \shortcomingGap.
We describe topologies as graph-based models and possible operations of topology control algorithms as declarative graph transformation (\GT) rules~\cite{Rozenberg1997}.
Although this approach provides a well-defined procedure for modeling the static and dynamic aspects of \TC algorithms in general, it does not ensure that all required formal properties are fulfilled for the resulting topology.
To this end, a well-known, constructive, static analysis technique~\cite{HW95}
has been established in the \GT community to formulate structural invariants in terms of graph constraints and to guarantee that these graph constraints hold.
Graph constraints specify positive or negative graph patterns, which must be present in or missing from a valid graph, respectively. 
Based on these graph constraints and a set of \GT rules, a refinement algorithm enriches the \GT rules with additional application conditions.
These application conditions are derived from the graph constraints and ensure that the applying the refined \GT rules never produces invalid graphs \wrt the graph constraints.
This technique could previously only be applied to scenarios where invariants must hold permanently (\eg, \cite{Koch2002}).
In the \WSN domain, the situation is different because context events inevitably violate the specified graph constraints.
To address this problem, we first relaxed the original constraints by introducing appropriate intermediate states of links.
Then, we ensured that the specification of the \TC algorithm always preserves the relaxed constraints, using the constructive approach described in~\cite{HW95}.
Finally, we iteratively modified the state configuration of the topology to enforce the original, strong constraints.
We illustrated the proposed constructive methodology by re-engineering a \emph{single} existing \TC algorithm, \ktc~\cite{SWBM12}.
Yet, \shortcomingVariability, \idest, describing commonalities and differences of \TC algorithms, still remained an open issue in~\cite{KVS15,Kluge2016}.

\paragraph{Contribution}
In this paper, we tackle \shortcomingVariability by generalizing the constructive, model-driven methodology for designing \TC algorithms using graph transformation~\cite{KVS15,Kluge2016} to support the development of families of \TC algorithms.
This paper has four major contributions:
\begin{enumerate}
\item 
We model commonalities and differences of \TC algorithms by extracting common structural constraints and specifying the individual part of each \TC algorithm in terms of individual attribute constraints.
We lift all steps described in our original approach~\cite{KVS15} to operate on abstract representations of \TC algorithm families.

\item 
To demonstrate the applicability of our approach, we re-engineer six existing \TC algorithms (kTC~\cite{SWBM12}, l-kTC~\cite{SKSVSM15,Stein2016b}, XTC~\cite{Wattenhofer2004}, RNG~\cite{Karp2000}, GG~\cite{Wang08}, Yao Graph~\cite{Yao1982}) and propose \ektc, a novel, energy-aware variant of \ktc~\cite{SWBM12} that has been inspired by the CTCA algorithm~\cite{XH12}.

\item 
We extend the constructive approach with a step that systematically derives context event handlers, which repair all constraint violations that may result from the context events.

\item 
We perform a comparative, simulation-based evaluation of \ktc and \ektc to showcase an integration of the \GT tool \eMoflon~\cite{LAS14} and the \Simonstrator network simulation environment~\cite{RSRS15}.
\end{enumerate}

\paragraph{Structure}
\Cref{fig:overview} maps the major contributions of this paper to the following sections.
\begin{figure}
    \begin{center}
        \resizebox{1.0\textwidth}{!}{%
            \begin{tikzpicture}[node distance = 4.925cm]
            
            \sffamily
            \normalsize
            
            \node [rcBig,color=white] (0) at (0,0) {};
            
            \draw[bigArrow] ([shift=({-0.5cm,0mm})]0.west) -- ++(16.25,0) --  ++(0,-2.85) --++ (-16.6,0) --++ (0,-2.8) --++ (16.5,0);
            
            \node [rcBig] (1) at (0,0) {};
            \node[sectionSummaryBox] (1TopText) at (1.north) {
            Shortcomings in current development of topology control algorithms:
            gap between specification and implementation (\shortcomingGap) and missing reuse (\shortcomingVariability)\\
            \lbrack \Cref{sec:introduction}\rbrack
            }; 
            \node[align=left, below of=1, node distance=0.7cm, anchor=north, text width=4.6cm] (1BottomText) {};
            
            \node [rcBig, right of=1] (2) {};
            \node[sectionSummaryBox] (2TopText) at (2.north) {
            Model wireless sensor network topologies and
            topology control algorithms\\
            \lbrack \Cref{sec:background}\rbrack
            }; 
            \node[align=left, below of=2, node distance=0.7cm, anchor=north, text width=4.6cm] (2BottomText) {NEW:\\ Modular specification of six \TC algorithms + \ektc};
            
            \node [rcBig, text width=4.6cm, right of=2] (3) {};
            \node[sectionSummaryBox] (3TopText) at (3.north) {
            Model consistency properties using graph constraints\\
            \lbrack \Cref{sec:constraints}\rbrack
            }; 
            \node[align=left, below of=3, node distance=0.7cm, anchor=north, text width=4.6cm] (3BottomText) {NEW:\\ Lift proves to families of \TC algorithms};
            
            \node [rcBig, below of=1, node distance=5.65cm] (4) {};
            \node[sectionSummaryBox] (4TopText) at (4.north) {Model topology modifications using \GT rules and programmed \GT\\ 
            \lbrack \Cref{sec:gratra}\rbrack
            };
            \node[align=left, below of=4, node distance=0.7cm, anchor=north, text width=4.6cm] (4BottomText) {};
            
            \node [rcBig, right of=4] (5) {};
            \node[sectionSummaryBox] (5TopText) at (5.north) {Derive constraint-preserving graph transformation rules\\
            \lbrack \Cref{sec:refinement}\rbrack
            };
            \node[align=left, below of=5, node distance=0.7cm, anchor=north, text width=4.6cm] (5BottomText) {NEW:\\ Derive context event handlers};
            
            \node [rcBig, right of=5] (6) {};
            \node[sectionSummaryBox] (6TopText) at (6.north) {Simulation-based evaluation of network lifetime\\
            \lbrack \Cref{sec:evaluation}\rbrack
            };
            \node[align=left, below of=6, node distance=0.7cm, anchor=north, text width=4.6cm] (6BottomText) {NEW:\\ Entirely};
            
            \node[bigArrowHead, text width=0.055mm, right of=6, node distance=2.95cm](arrow1) {};
            
            \end{tikzpicture}}%
    \end{center}
    \caption{Structure of this paper (NEW: extensions compared to \cite{KVS15,Kluge2016})}
    \label{fig:overview}
\end{figure}
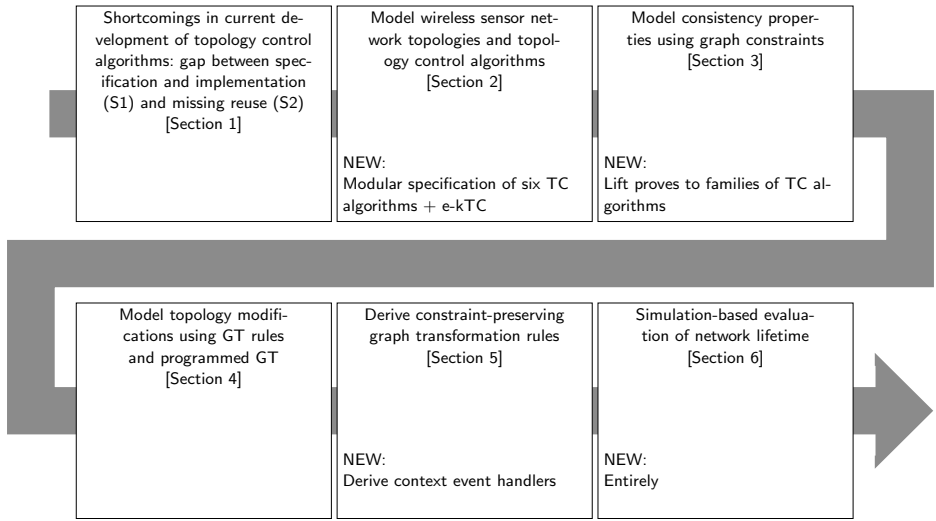
In \Cref{sec:background}, we specify valid topologies using metamodeling and introduce the six existing and one novel \TC algorithm using first-order logic predicates.
In \Cref{sec:constraints} we specify the \TC algorithms using graph constraints and conduct the prove of connectivity based an abstract specification of the family of \TC algorithms.
In \Cref{sec:gratra}, we specify topology modifications using \GT rules and \TC algorithms using programmed \GT.
In \Cref{sec:refinement}, we refine the \GT rules based on the graph constraints to ensure that the refined \GT rules preserve the graph constraints.
Additionally, we derive handler operations for the \CE rules and ensure that the refined \TC algorithm terminates.
In \Cref{sec:evaluation}, we present the results of a simulation-based evaluation.
In \Cref{sec:relatedWork}, we survey related work and conclude this paper in \Cref{sec:conclusion}.

\section{Metamodeling and Topology Control}
\label{sec:background}
In this section, we introduce basic concepts of metamodeling and \TC.
Afterwards, we introduce the considered \TC algorithms and analyze them \wrt recurring substructures.

\subsection{Basic Metamodeling Concepts}
\label{sec:metamodeling-intro}

A \emph{model} describes a set of related entities as a graph whose nodes are \emph{objects} and whose edges are \emph{references}.
A \emph{metamodel} specifies all well-formed models of the considered domain as a multi-graph whose nodes are the \emph{classes}, which describe possible entities and serve as \emph{object type}, and whose edges are \emph{associations}, which describe possible relations between entities and serve as \emph{reference type}.
Object and reference types have to be \emph{compatible}, \idest, the types of the source and target of a reference are the source and target class of its corresponding association.
A class may have multiple typed \emph{attribute}s, which represent properties of its instances.
An association end is labeled with a \emph{role name}, which further describes the corresponding relation, and a \emph{multiplicity}, which restricts the number of corresponding references in a model.

\subsection{Topologies}
\label{sec:topology-control-intro}

A (network) \emph{topology} represents the state of a communication system as an attributed graph consisting of \emph{nodes} and (communication) \emph{links}~\cite{Santi2005}.
In this paper, we consider topologies of \WSNs, \idest, nodes in the topology correspond to battery-powered wireless sensor nodes, and links correspond to the possible direct wireless communication connections between sensor nodes.
This implies that a topology is a simple graph, which neither contains loops nor parallel links, \idest, the source and target node of a link are unequal, and each pair of nodes is connected by at most one link.
We denote links with the letter $e$\footnote{We use $e$ instead of $l$ or $\ell$ for better readability.} in running text, \eg, \linkVariableOneTwo, \linkVariableab, \linkVariableAB, and as arrow-headed lines in compact notation, \eg, \tikz[baseline,anchor=base]{\node[shape=circle,draw,inner sep=2pt,text=white,fill=black] (act1) {$1$}; \node[right of=act1,shape=circle,draw,inner sep=2pt,text=white,fill=black] (act2) {$2$}; \draw[-angle 45,very thick,black,solid] (act1) -- (act2);}.
By convention, a link \linkVariableab has source node \nodeVariablea and target node \nodeVariableb.
A \emph{path} $P(\nodeVariablea, \nodeVariablez) = (\linkVariableab, \linkVariablebc, \dots, \linkVariableyz)$ from \nodeNameLong{a} to \nodeNameLong{z} is a list of links where the target node of one link in $P$ is the source node of its successor link in $P$.
In the following, we introduce node and link properties that are required to model the \TC algorithms in this paper.
A sensor node \nodeVariablea exposes the following properties:
\begin{itemize}
\item 
An integer-valued unique \emph{identifier} $\id{\nodeVariablea} = a$ allows to distinguish \nodeVariablea from other nodes.
The identifier of a node is shown in subscript notation in running text, \eg, \nodeVariableOne, and as white label inside the corresponding solid black circle in compact notation, \eg, \tikz[baseline=(char.base)]{\node[shape=circle,draw,inner sep=2pt,text=white,fill=black] (char) {$1$};}.
\item 
A real-valued \emph{energy} property \energy{\nodeVariablea}{} stores the current energy level of node \nodeVariablea, which is typically measured in Joule.
\item 
The real-valued \emph{latitude} \latitude{\nodeVariablea} and \emph{longitude} \longitude{\nodeVariablea} capture the position of node \nodeVariablea (\eg, Euclidean or GPS coordinates).
\item 
The integer-valued \emph{hop count} \hopcount{\nodeVariablea} stores the shortest distance (\wrt the number of hops) between \nodeVariablea and a dedicated second node \nodeVariableBaseStation.
A \emph{hop} is the traversal of a single link.
This property is required in application scenarios such as data collection, where each sensor node periodically sends collected data of its environment to a dedicated base station node \nodeVariableBaseStation.
A routing protocol (\eg, RPL~\cite{Winter2012}, AODV~\cite{CB04}) operates on top of the output topology of the current \TC algorithm and determines the path between \nodeVariablea and the base station.
\end{itemize}
A link \linkVariableab exposes the following properties:
\begin{itemize}
\item 
A real-valued generic \emph{weight} \weight{\linkVariableab} stores the cost of using \linkVariableab for message transfer.
For example, the weight of \linkVariableab may be derived from the distance of its incident nodes or the received signal strength indicator (RSSI) at \nodeVariableb.
\item
The real-valued \emph{angle} \angleCmd{\linkVariableab} of a link \linkVariableab can be derived from the positions of its incident nodes as follows:
\begin{align*}
\angleCmd{\linkVariableab} &= \atan(
\latitude{\nodeVariablea} - \latitude{\nodeVariableb}, \longitude{\nodeVariablea} - \longitude{\nodeVariableb}) + 180\degree
\end{align*}
With $\atan$, we denote the arcus tangens operator, which maps a pair of latitudinal and longitudinal differences to the corresponding angle.
\item
The \emph{state} \state{\linkVariableab} stores the processing state of \linkVariableab during the execution of the \TC algorithm.
Details follow in \Cref{sec:topology-control}.
\end{itemize}

\paragraph{Topology metamodel}
\Cref{fig:topology-metamodel} depicts the metamodel of topologies and contains three classes, \reftofig{Topology}\footnote{We use sans-serif font when referring to metamodel elements.}, \reftofig{Node}, and \reftofig{Link} (depicted as rectangular boxes).
Its eight associations (depicted as four bidirectional arrows) specify that 
\begin{inparaenum}
\item 
each \reftofig{Node} and \reftofig{Link} is contained in a single \reftofig{Topology},
\item 
each \reftofig{Topology} contains an unlimited number of \reftofig{Node}s and \reftofig{Link}s, and
\item a \reftofig{Node} serves as the unique \reftofig{source} (\reftofig{target}) of any of its zero or more \reftofig{outgoing} (\reftofig{incoming}) \reftofig{Link}s.
\end{inparaenum}
The class \reftofig{Node} has two integer-valued attributes (\reftofig{id} and \reftofig{hopCount}), and three real-valued attributes (\reftofig{energy}, \reftofig{latitude}, \reftofig{longitude}).
The class \reftofig{Link} has a real-valued \reftofig{weight} and a \reftofig{state} that can take values \reftofig{Active}, \reftofig{Inactive}, and \reftofig{Unclassified}, specified in the enumeration type \reftofig{LinkState}.
All attributes and types correspond to the aforementioned node and link properties of the same name.
\begin{figure}
    \begin{center}
        \includegraphics[width=.85\textwidth]{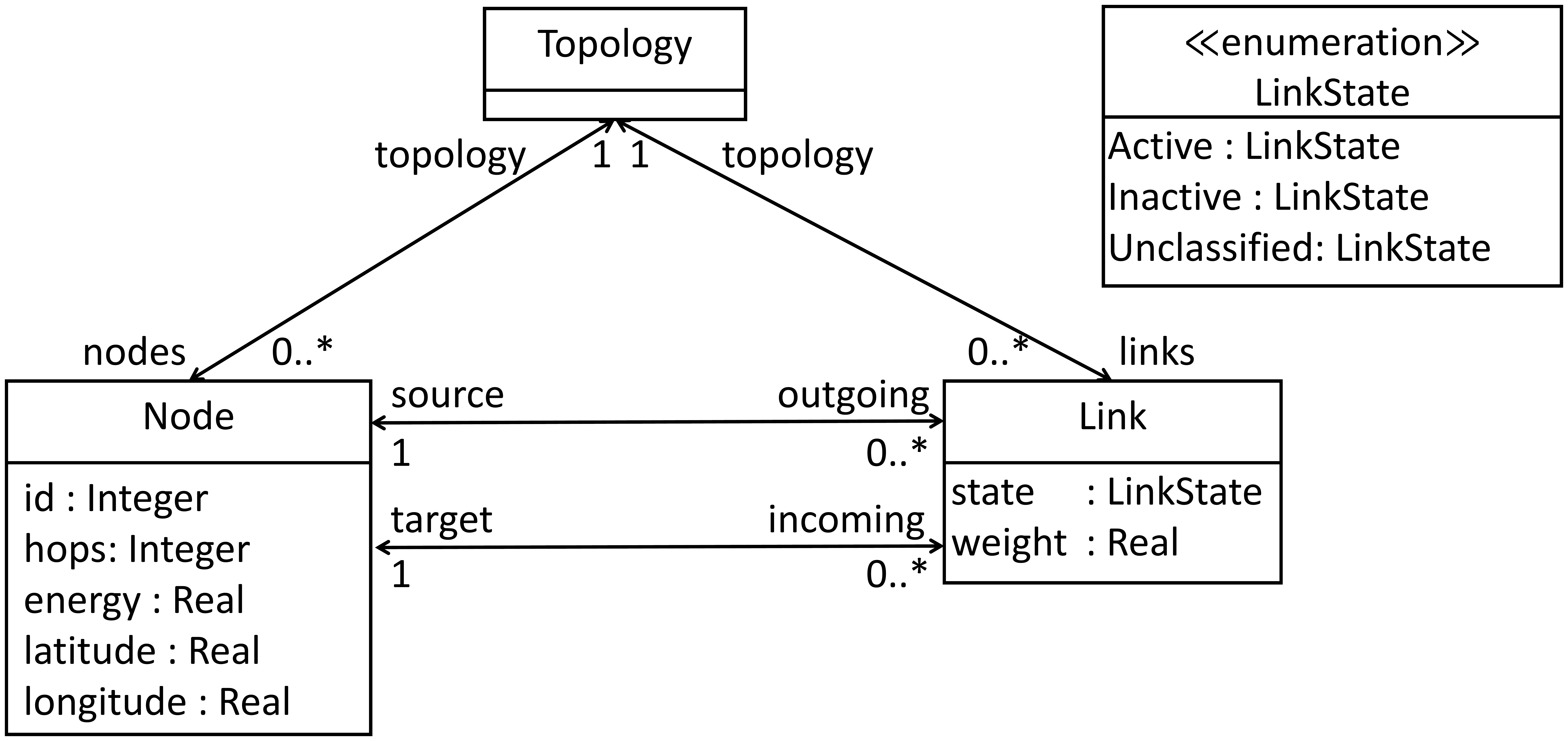}
        \caption{Topology metamodel}
        \label{fig:topology-metamodel}
    \end{center}
\end{figure}

\Cref{fig:topology-example-triangle} shows a \emph{(directed) triangle} of links in object and compact notation.
In this example, latitude and longitude are Euclidean coordinates, link weights represent the Euclidean distance between the incident nodes, and the hop count is relative to node \nodeVariableOne.
Throughout this paper, we assume that every node and link is part of a single topology $G$ with node set $V$ and edge set $E$.
For brevity, we use the compact notation and depict only the relevant attribute values (\eg, the link weight in this case) in the following.
\begin{figure*}
    \begin{center}     
        \begin{subfigure}[t]{\textwidth}
            \includegraphics[width=.99\textwidth]{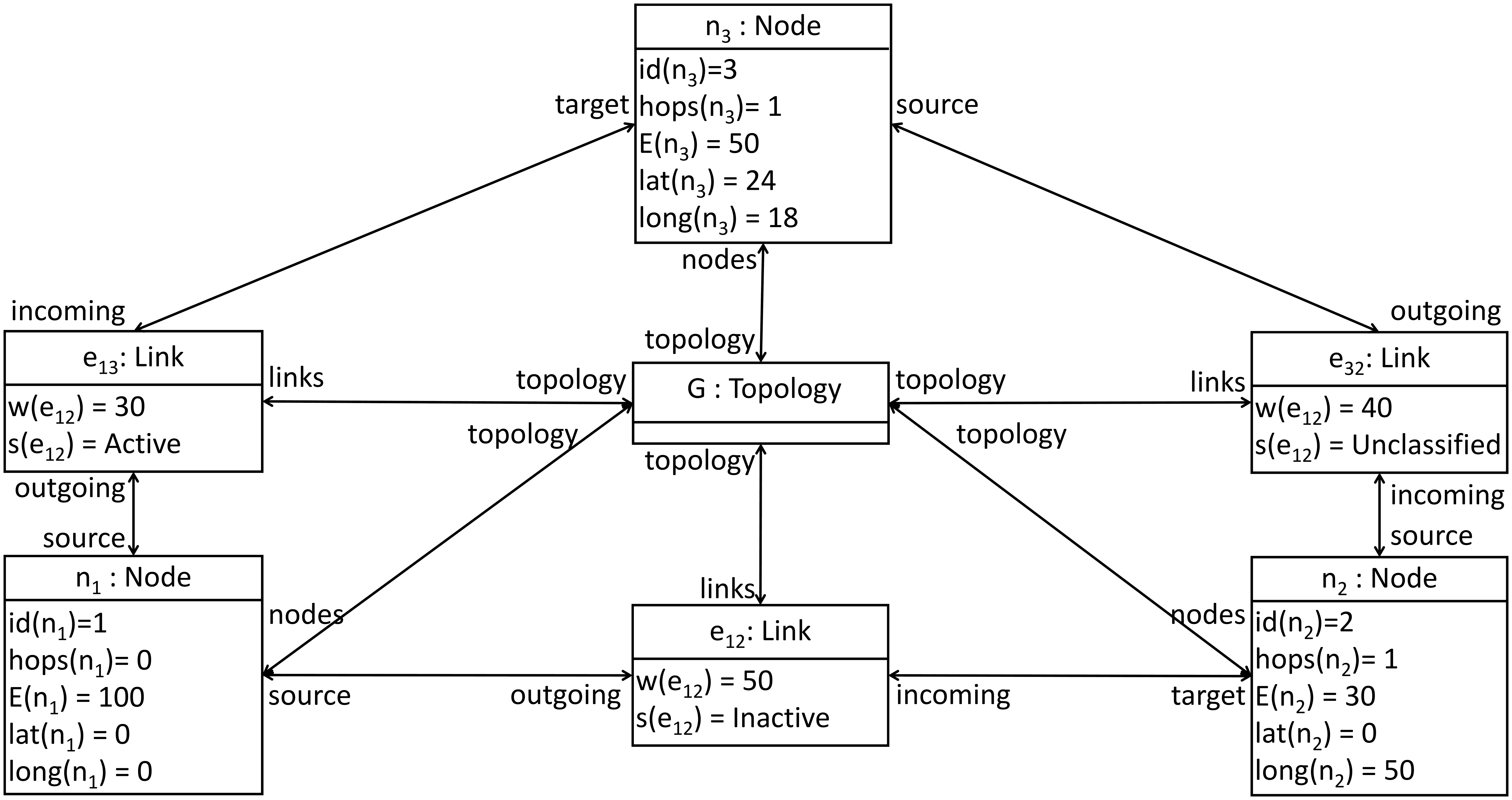}
            \caption{Object notation}
            \label{fig:topology-example-triangle-abstract}
        \end{subfigure}
        
        \begin{subfigure}[t]{\textwidth}
            \begin{center}
                \includegraphics[width=.3\textwidth]{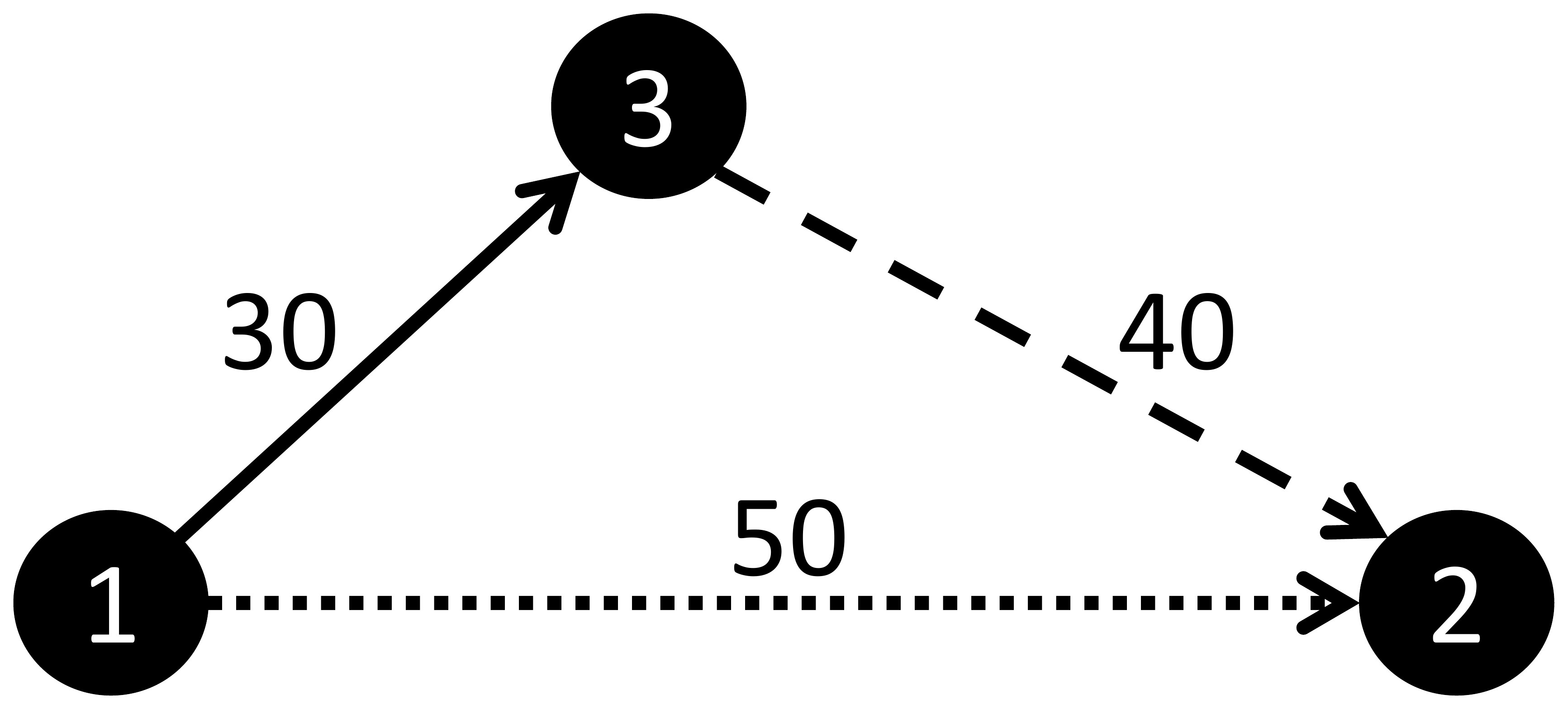}        
                \hspace{4em}
                \includegraphics[width=.3\textwidth]{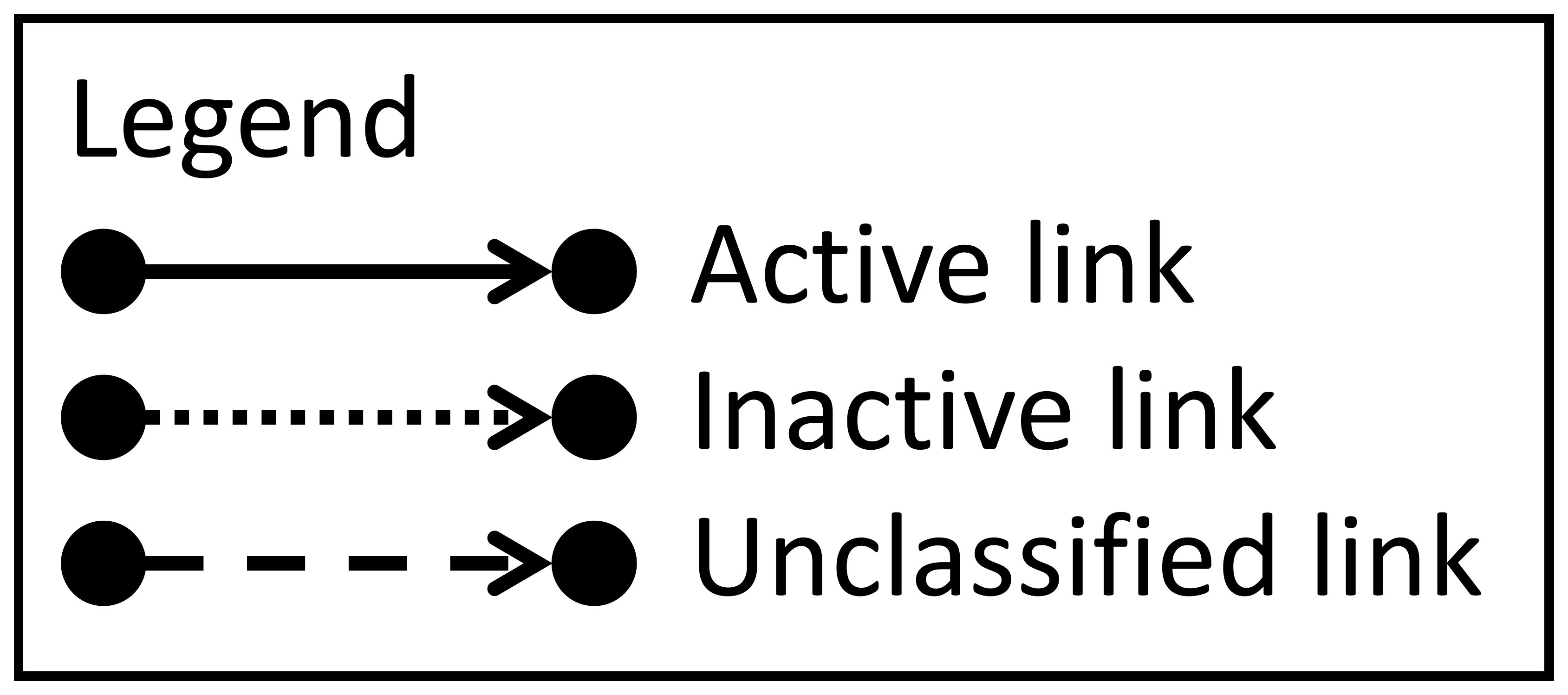}      
            \end{center}
            \caption{Compact notation}
            \label{fig:topology-example-triangle-concrete}
        \end{subfigure}
        \caption{Triangle topology in object and compact notation. Hop count is relative to \nodeVariableOne.}
        \label{fig:topology-example-triangle}
    \end{center}
\end{figure*}

\subsection{Topology Control}
\label{sec:topology-control}
Topology control (\TC{}) is the discipline of adapting wireless sensor network topologies to optimize network metrics.
As described earlier, wireless sensor nodes are typically battery-powered, and often the energy source is not (easily) exchangeable or rechargeable.
This makes prolonging the network lifetime a key optimization goal for \WSNs~\cite{Santi2005}.

\Cref{fig:topology-control-steps} sketches the three phases of the \TC process: topology monitoring, planning, and execution.
\begin{figure}[htbp]
    \begin{center}
        \includegraphics[width=\textwidth]{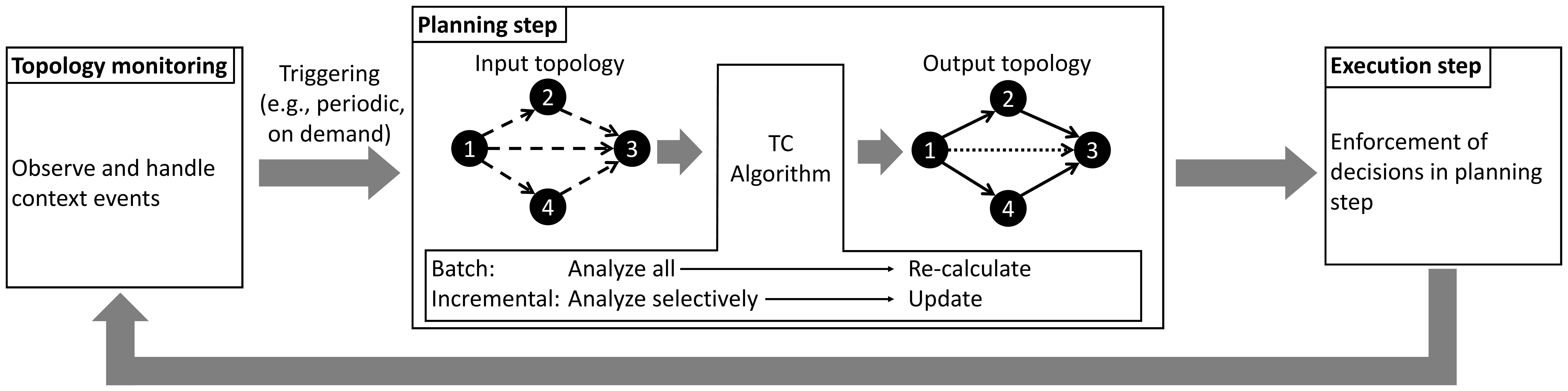}
        \caption{Topology control process}
        \label {fig:topology-control-steps}
    \end{center}
\end{figure}
The \emph{topology monitoring} detects \emph{context events}, which are external modifications of the physical topology.
In this paper, we consider the following six types of \emph{context event}s: node addition, node removal, link addition, link removal, node property modification, and link property modification.
The execution of the \TC process is \emph{triggered} either periodically or on-demand, \eg, when a batch of context events has finished.
In the \emph{planning step}, the \TC algorithm analyzes the input topology and produces a corresponding output topology.
The output topology contains all links of the input topology that are necessary to fulfill the specified consistency properties (\eg, reduced node degree or connectivity).
We distinguish between batch and incremental \TC algorithms:
A \emph{batch \TC algorithm} analyzes the entire input topology and outputs an entire output topology.
An \emph{incremental \TC algorithm} selectively analyzes the modified parts of the topology and updates the output topology accordingly.
This behavior requires that all links that are modified, added, and removed between two iterations of the planning step are marked.
In the \emph{execution step}, the sensor node ensures that all links in the output topology are available for message transfer.
The planning step is a link classification problem, whereas the execution step and the topology monitoring are highly platform-specific tasks.
Therefore, we focus on the planning step in this paper and leave the necessary refinement of the metamodel for the other two steps as future work.

In \Cref{fig:topology-control-steps}, we use a joint representation for the input and output topology.
The physically possible links in the topology are shown as directed lines. 
The decision of the \TC algorithm is stored as \emph{state} \state{\linkVariableab} attribute for each link \linkVariableab.
This representation allows us to store previous decisions of the \TC algorithm for each link, which is essential for incremental \TC algorithms.
More precisely, we say that a link \linkVariableab in the output topology is 
\begin{itemize}

\item 
\emph{active} if $\state{\linkVariableab}{=}\ACT$ if it is part of the output topology (denoted as solid line in compact notation, \eg, \tikz[baseline=(char.base)]{\node[shape=circle,draw,inner sep=2pt,text=white,fill=black] (act1) {$1$}; \node[right of=act1,shape=circle,draw,inner sep=2pt,text=white,fill=black] (act2) {$2$}; \draw[-angle 45,very thick,black,solid] (act1) -- (act2);}), 

\item 
\emph{inactive} if $\state{\linkVariableab}{=}\INACT$ if is not part of the output topology (denoted as dotted line in compact notation \eg, \tikz[baseline=(char.base)]{\node[shape=circle,draw,inner sep=2pt,text=white,fill=black] (act1) {$1$}; \node[right of=act1,shape=circle,draw,inner sep=2pt,text=white,fill=black] (act2) {$2$}; \draw[-angle 45,very thick,black,dotted] (act1) -- (act2);}),

\item
\emph{classified} if $\state{\linkVariableab} \in \{\ACT, \INACT\}$  if the \TC algorithm has made a decision for \linkVariableab (denoted as mixed dotted-solid line in compact notation, \eg, \tikz[baseline=(char.base)]{\node[shape=circle,draw,inner sep=2pt,text=white,fill=black] (act1) {$1$}; \node[right of=act1,shape=circle,draw,inner sep=2pt,text=white,fill=black] (act2) {$2$}; \draw[-angle 45,very thick,black,densely dashdotted] (act1) -- (act2);}),

\item
\emph{unclassified} $\state{\linkVariableab} = \UNCL$  if either the \TC algorithm has not considered link \linkVariableab yet or a context event has invalidated the decision of the \TC algorithm (denoted as dashed line in compact notation, \eg, \tikz[baseline=(char.base)]{\node[shape=circle,draw,inner sep=2pt,text=white,fill=black] (act1) {$1$}; \node[right of=act1,shape=circle,draw,inner sep=2pt,text=white,fill=black] (act2) {$2$}; \draw[-angle 45,very thick,black,dashed] (act1) -- (act2);}), and

\item
\emph{undefined} if we either do not know or do not care about the state of \linkVariableab (denoted as gray line in compact notation, \eg, \tikz[baseline=(char.base)]{\node[shape=circle,draw,inner sep=2pt,text=white,fill=black] (act1) {$1$}; \node[right of=act1,shape=circle,draw,inner sep=2pt,text=white,fill=black] (act2) {$2$}; \draw[-angle 45,thick,gray,solid] (act1) -- (act2);}).
\end{itemize}
A \emph{link state modification} is the modification of the state of a single link, \idest, activation, inactivation, or unclassification of the link.

\subsection{Specifying Valid Output Topologies with First-Order Logic Predicates}
\label{sec:predicates}

In the following, we specify required properties of output topologies in terms of first-order logic predicates.
We begin with two general properties that must hold for any \TC algorithm.
Afterwards, we introduce additional algorithm-specific conditions for seven \TC algorithms (six existing \TC algorithms and one new \TC algorithm variant).

\subsubsection{General Required Properties of Output Topologies}
\label{sec:general-properties}
Upon termination of every \TC algorithm, each link in the topology should be classified and the output topology should be connected.
These requirements are described by the following two predicates.

\paragraph{Complete Classification Constraint $\completeClassificationPredicate$:}
A \TC algorithm should make a definite decision for each link in the topology, \idest, the output topology of every \TC algorithm should only contain classified links.
This postcondition ensures that a \TC algorithm may only terminate after completely classifying the input topology;
more formally:
\begin{align}\label{eqn:no-unclassified-links}
    \completeClassificationPredicate(G(V,E)) \Leftrightarrow \forall \linkVariableab \in E: \state{\linkVariableab} \in \{\ACT, \INACT\}
\end{align}

\paragraph{A-connectivity Predicate $\AConnectivityPredicate$:}
The output topology must be connected, \idest, each pair of nodes $\nodeVariablea, \nodeVariableb \in V$ must be connected by a path $P_{\text{out}}(\nodeVariablea,\nodeVariableb)$ of active links if a path $P_{\text{in}}(\nodeVariablea, \nodeVariableb)$ of edges exists in the input topology.
This requirement can be described by the following A-connectivity predicate $\AConnectivityPredicate$:
\begin{align}\label{eqn:active-link-connectivity}
    \AConnectivityPredicate(G(V,E))\Leftrightarrow  &\forall (\nodeVariablea, \nodeVariableb) \in V \times V: 
     \notag\\
    &\big(
    \exists P_\text{in}(\nodeVariablea,\nodeVariableb) \Rightarrow 
    \exists P_\text{out}(\nodeVariablea,\nodeVariableb): \forall \linkVariable \in P_{\text{out}}:
        \state{\linkVariable} = \ACT 
    \big)
\end{align}
Note that A-connectivity can only be evaluated based on a global view of the topology, whereas complete classification can be checked based on local knowledge of each node's outgoing links.

\subsubsection{Algorithm-Specific Properties}
\label{sec:algorithm-specific-predicates}

Each \TC algorithm has specific optimization goals, which jointly describe when a link may be inactive in a valid output topology;
all links that do not fulfill these conditions have to be active.
As an example, we consider the \TC algorithm \ktc~\cite{SWBM12}.
In a valid output topology of \ktc, a link \linkVariableab is inactive if and only if 
\begin{inparaenum}
\item 
it is the weight-maximal link in a triangle, together with classified links \linkVariableac and \linkVariablecb, and 
\item its weight is additionally \ktcParameterK times larger than the weight of the weight-minimal link in the same triangle;
\end{inparaenum}
more formally:
\begin{align}\label{eqn:ktc-inactive-link-constraint}
\begin{split}
\forall \linkVariableab \in E:\;&\;\state{\linkVariableab} = \INACT \\
& \Leftrightarrow\; \linkVariableab \text{ is in a triangle with classified links } \linkVariableac, \linkVariablecb \\
& \wedge\; \weight{\linkVariableab} \geq \max\left(\weight{\linkVariableac}, \weight{\linkVariablecb}\right) \\
& \wedge\; \weight{\linkVariableab} \geq k \cdot \min\left(\weight{\linkVariableac}, \weight{\linkVariablecb}\right). 
\end{split}
\end{align}
The core idea of \ktc is that it is often beneficial to use multiple shorter (\idest, more energy-efficient) links (here: \linkVariableac and \linkVariablecb) instead of one long link (here: \linkVariableab) for transferring a message because the required transmission power grows at least quadratically with the length of a link~\cite{Fri46}.

A closer look at \Cref{eqn:ktc-inactive-link-constraint} reveals that the algorithm-specific condition contains a structural predicate (the first line, here: a triangle) and an additional attribute predicate (the remaining two lines), which refers to the links identified by the structural predicate.
In fact, this is a recurring property in specifications of \TC algorithms.
We express this separation in the following reformulation of \Cref{eqn:ktc-inactive-link-constraint}.
\begin{align}\label{eqn:ktc-inactive-link-constraint-split}
\begin{split}
\forall \linkVariableab \in E:\;&\;\state{\linkVariableab} = \INACT \\
& \Leftrightarrow\; \exists \linkVariableac, \linkVariablecb :   \trianglePredicate{\linkVariableab, \linkVariableac, \linkVariablecb} \wedge\ \ktcPredicate{\linkVariableab, \linkVariableac, \linkVariablecb}\\ 
&\text{with}\\
\trianglePredicate{\linkVariableab, \linkVariableac, \linkVariablecb} &= \linkVariableab \text{ is in a triangle with classified links } \linkVariableac, \linkVariablecb\\
\ktcPredicate{\linkVariableab, \linkVariableac, \linkVariablecb} &= \weight{\linkVariableab} \geq \max\left(\weight{\linkVariableac}, \weight{\linkVariablecb}\right) \\
& \wedge\; \weight{\linkVariableab} \geq k \cdot \min\left(\weight{\linkVariableac}, \weight{\linkVariablecb}\right).
\end{split}
\end{align}
The \emph{directed-triangle predicate} \trianglePredicate{} reflects the structural condition that an inactive link \linkVariableab must be part of a triangle.
The \emph{\ktc predicate \ktcPredicate{}} specifies the condition that \linkVariableab must be the weight-maximal link and at least \ktcParameterK~times larger than the weight-maximal link among \linkVariableac and \linkVariablecb.

\paragraph{Example: Incremental \TC using \ktc}
\Cref{fig:example-ktc} shows the evolution of a sample topology (\Cref{fig:example-ktc-A}).
For conciseness, we show a link and its reverse link as a single double-headed line.
The topology is first optimized by invoking \ktc ($k=2$) (\Cref{fig:example-ktc-B}).
Then, two links \linkName{79} and \linkName{97} are added (\eg, because an obstacle between \nodeNameLong{7} and \nodeNameLong{9} has moved out of the way), and node \nodeName{10} is removed (\eg, because its battery is empty).
The resulting topology is shown in \Cref{fig:example-ktc-C}.
The \CE handling has unclassified the new links \linkName{79} and \linkName{97} and all links around the removed \nodeNameLong{10}.
Finally, the topology is processed by \ktc again (\Cref{fig:example-ktc-D}).
Now, the added links \linkName{79} and \linkName{97} as well as the formerly inactive links \linkName{3,11}, \linkName{11,3}, \linkName{9,11} and \linkName{11,9} are active, and the links \linkName{3,9} and \linkName{9,3} are inactive.
\begin{figure*}
    \begin{center}
        \subcaptionbox{Initial topology\label{fig:example-ktc-A}}[.4\textwidth]{\includegraphics[width=.4\textwidth]{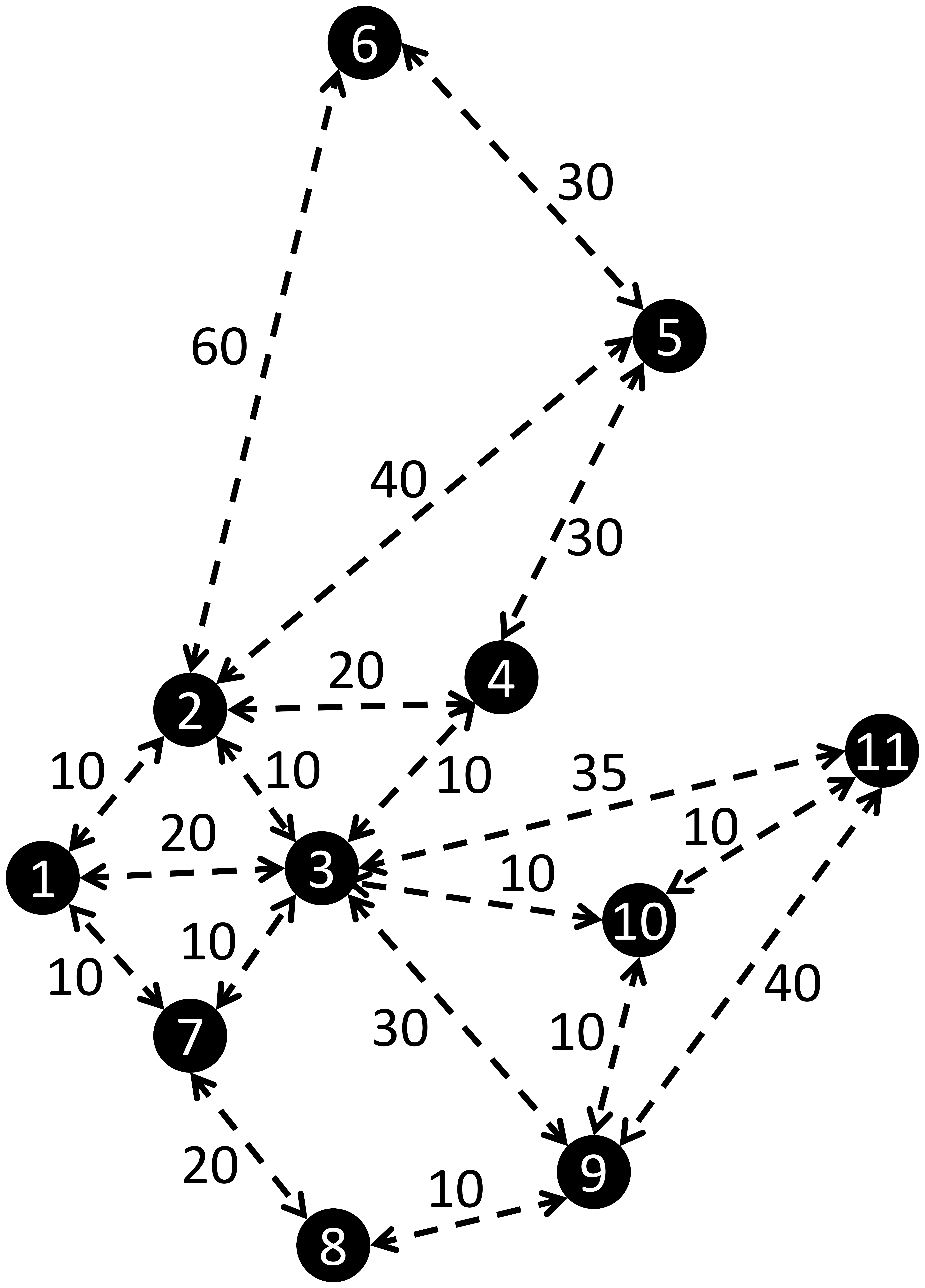}}\hspace{1em}
        \subcaptionbox{After initial \TC \label{fig:example-ktc-B}}[.4\textwidth]{\includegraphics[width=.4\textwidth]{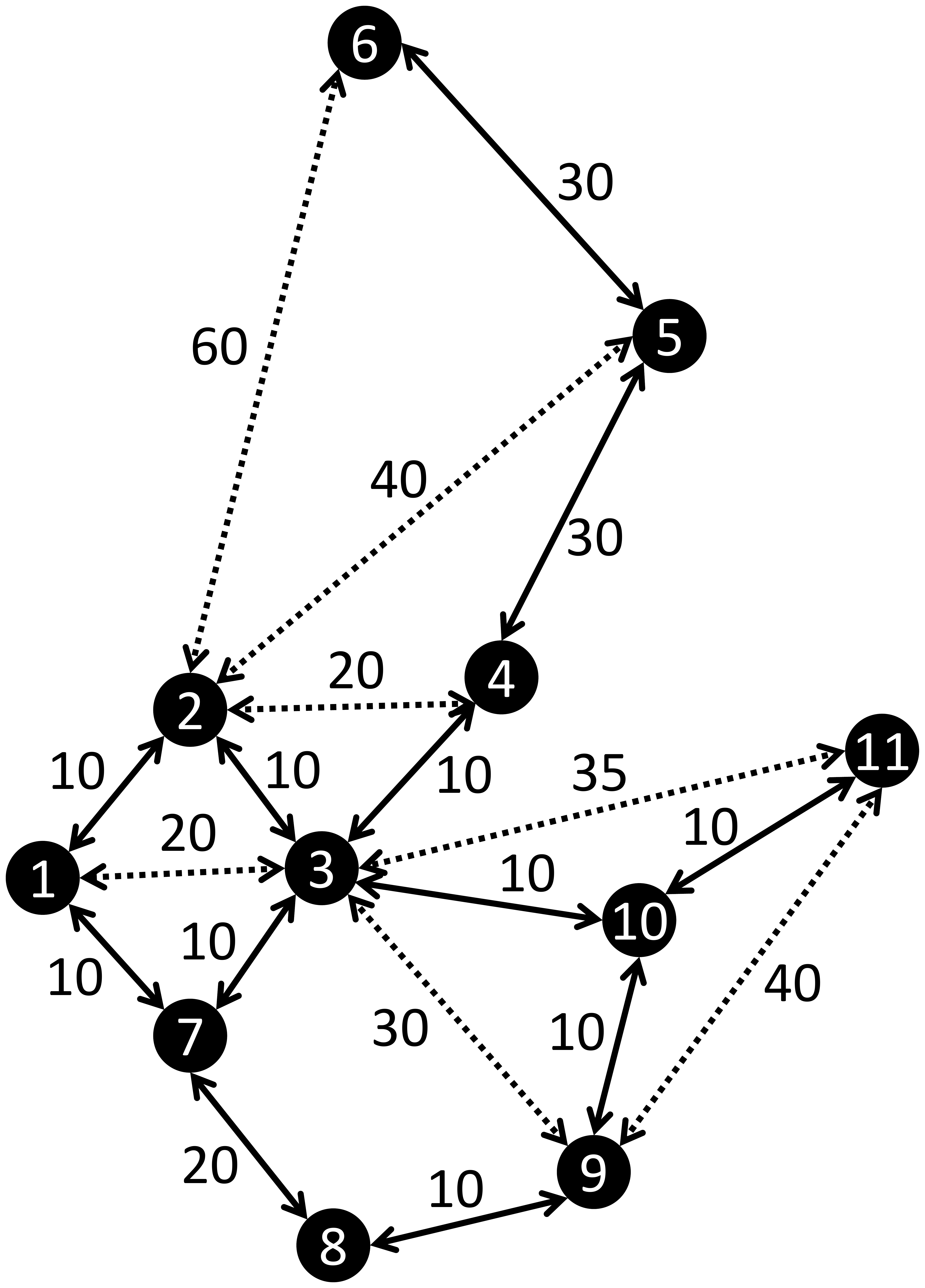}}
        
        \subcaptionbox{After addition of \linkName{79}, \linkName{97} and rem. of \nodeName{10}\label{fig:example-ktc-C}}[.4\textwidth]{\includegraphics[width=.4\textwidth]{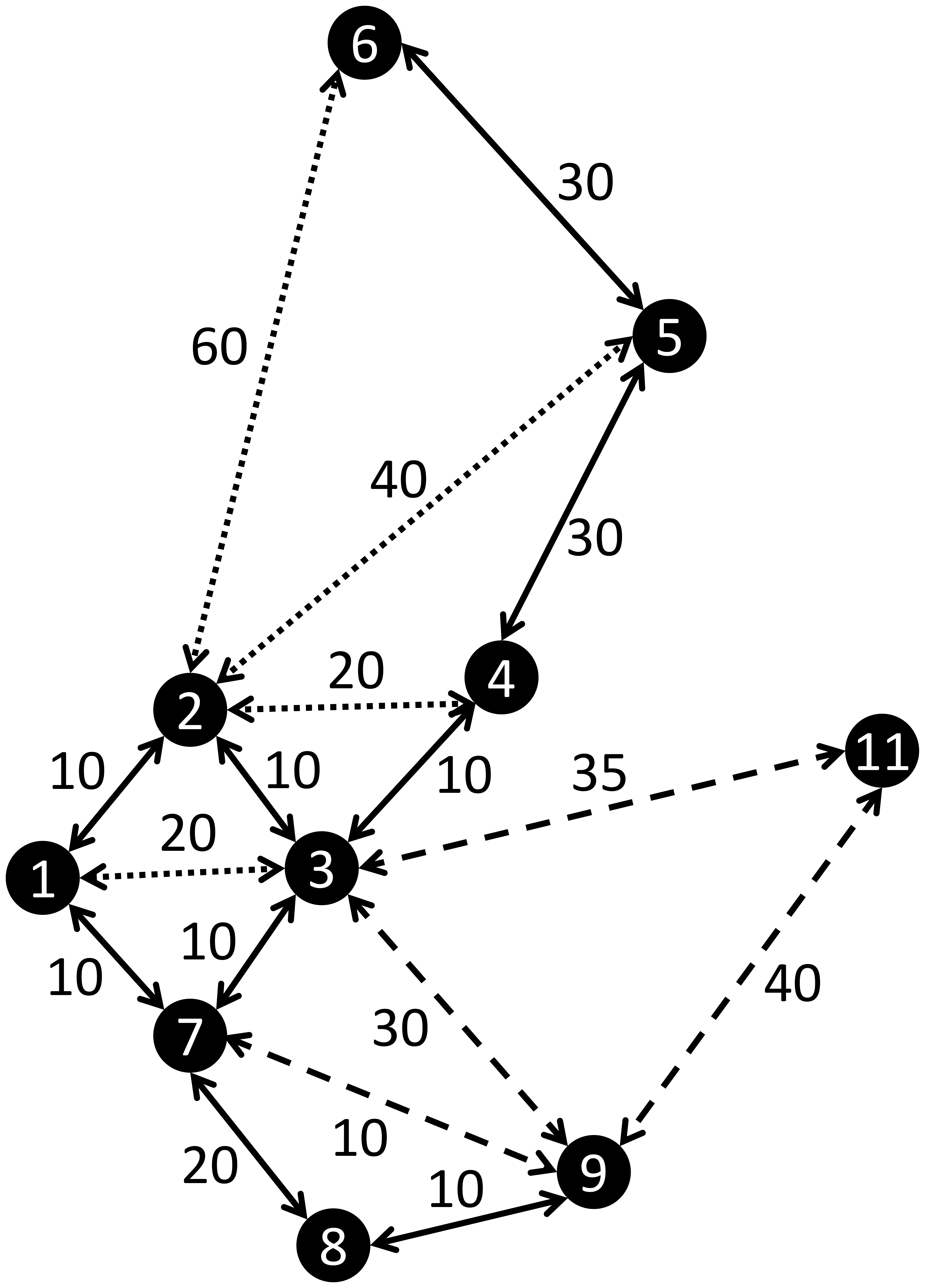}}\hspace{1em}
        \subcaptionbox{After incremental \TC\label{fig:example-ktc-D}}[.4\textwidth] {\includegraphics[width=.4\textwidth]{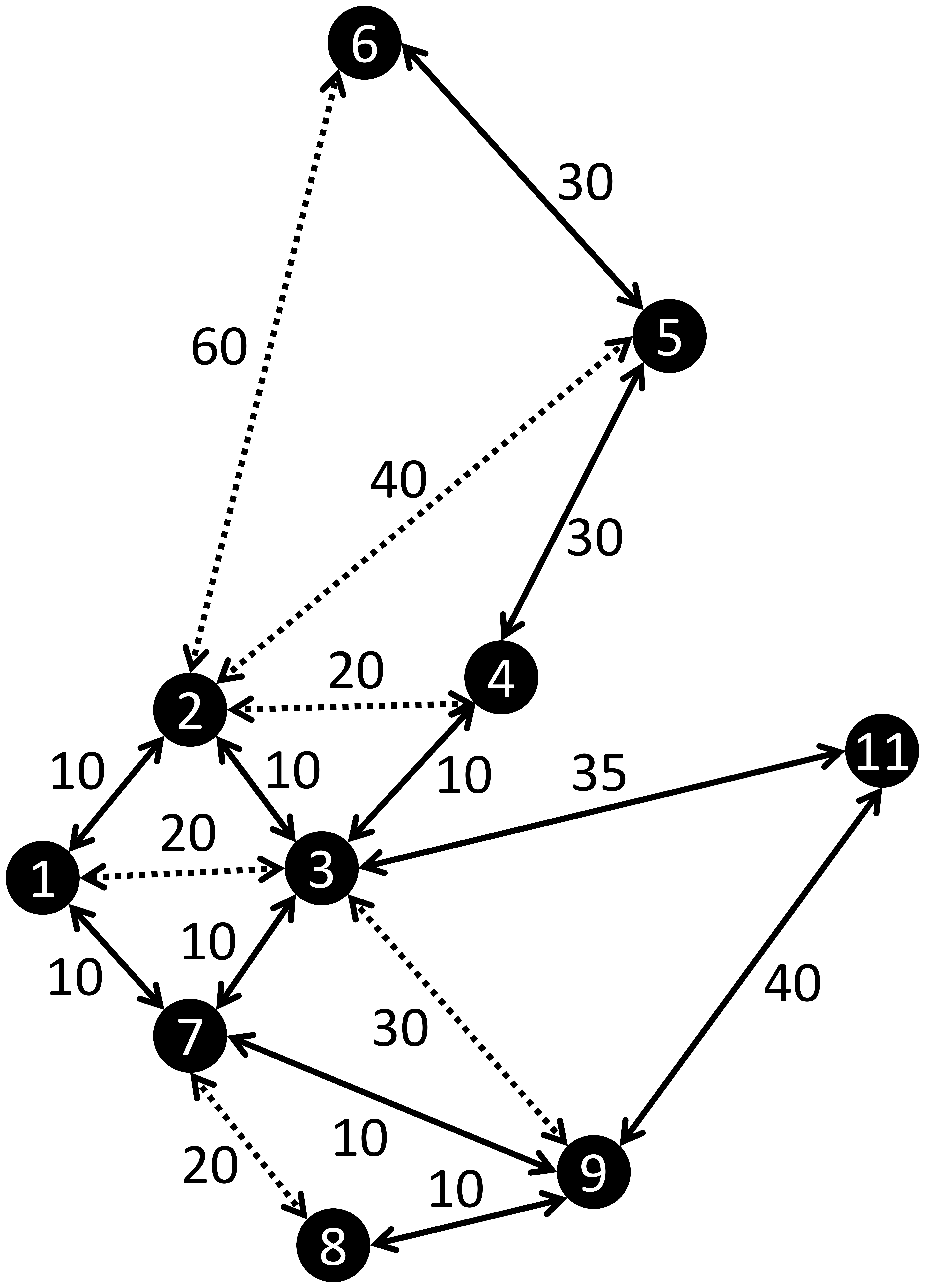}}
        
        \subcaptionbox{Legend}[.4\textwidth] {\includegraphics[width=.4\textwidth]{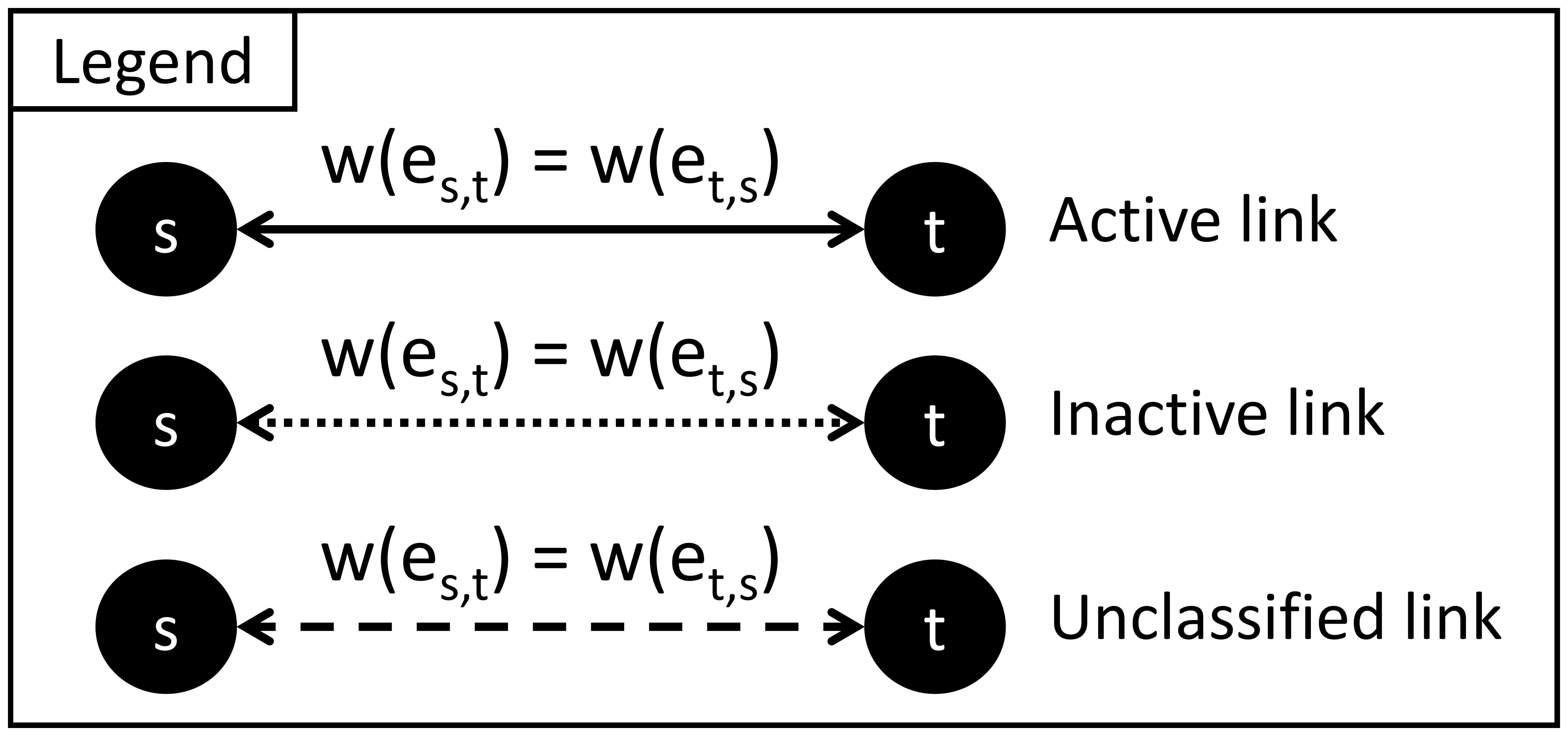}}
        
    \end{center}
    \caption{Example of incremental \TC with \ktc ($k=2$)}
    \label{fig:example-ktc}
\end{figure*}

\subsubsection{The Tie-Breaking Predicate \tieBreakingPredicate{}}
\label{sec:tie-breaking-predicate}

A recurring issue while developing \TC algorithms is that more than one link, \eg, in a triangle, may fulfill the algorithm-specific predicate, which may cause multiple links to be inactivated.
In case of \ktc, this may even lead to a disconnected output topology.
As a resort, tie breaker are applied in such situations.
For instance, a link \linkVariableab is only inactivated if it has the largest identifier compared to all other links in the triangle that fulfill the algorithm-specific predicate.
For triangles of weighted links, the \emph{\tieBreakingPredicateLong{}} is defined as follows:
\begin{align}\label{eqn:tie-breaking-predicate}
\begin{split}
\tieBreakingPredicate{\linkVariableab, \linkVariableac, \linkVariablecb} &= \left(\weight{\linkVariableab} = \weight{\linkVariableac}  
\Rightarrow \id{\linkVariableab} > \id{\linkVariableac}\right)\\
&\;\wedge\;\left(\weight{\linkVariableab} = \weight{\linkVariablecb}
\Rightarrow \id{\linkVariableab} > \id{\linkVariablecb}\right)\\
\end{split}
\end{align}
The following \Cref{eqn:ktc-inactive-link-constraint-tie-breaker} shows how the \tieBreakingPredicateLong{} can be used to compose a variant of \ktc that is guaranteed to inactivate only one weight-maximal link in each triangle.
\begin{align}\label{eqn:ktc-inactive-link-constraint-tie-breaker}
\begin{split}
\state{\linkVariableab} = \INACT 
& \Leftrightarrow \exists \linkVariableac, \linkVariablecb :   \trianglePredicate{\linkVariableab, \linkVariableac, \linkVariablecb} \wedge \ktcPredicate{\linkVariableab, \linkVariableac, \linkVariablecb} \wedge \tieBreakingPredicate{\linkVariableab, \linkVariableac, \linkVariablecb}\\ 
&\text{with}\\
\trianglePredicateOp{\linkVariableab, \linkVariableac, \linkVariablecb} &= \linkVariableab \text{ is in a triangle with classified links } \linkVariableac, \linkVariablecb\\
\ktcPredicate{\linkVariableab, \linkVariableac, \linkVariablecb} &= \weight{\linkVariableab} \geq \max\left(\weight{\linkVariableac}, \weight{\linkVariablecb}\right) \\
& \wedge\; \weight{\linkVariableab} \geq k \cdot \min\left(\weight{\linkVariableac}, \weight{\linkVariablecb}\right)\\
\tieBreakingPredicate{\linkVariableab, \linkVariableac, \linkVariablecb} &= \left(\weight{\linkVariableab} = \weight{\linkVariableac}  
\Rightarrow \id{\linkVariableab} > \id{\linkVariableac}\right)\\
&\;\wedge\;\left(\weight{\linkVariableab} = \weight{\linkVariablecb}
\Rightarrow \id{\linkVariableab} > \id{\linkVariablecb}\right)\\
\end{split}
\end{align}

\subsubsection{Maxpower Topology Control}
\label{sec:maxpower}
The \MaxpowerTC algorithm activates all links in a topology.
Its name derives from the fact that its output topology contains all links that are available if the node transmits with maximum power.
\MaxpowerTC is a generally accepted baseline for performing network evaluations.
The predicate \maxpowerPredicate{} of this algorithm is \emph{false} because it never inactivates a link:
Note that \MaxpowerTC does not even require the additional triangle-identifying predicate \trianglePredicate{}, \idest, the full specification of \MaxpowerTC looks as follows:
\begin{align*}
\forall \linkVariableab \in E: \state{\linkVariableab} &= \INACT \Leftrightarrow \maxpowerPredicate{\linkVariableab}\\ &\text{with } \maxpowerPredicate{\linkVariableab}= \text{false}.
\end{align*}

\subsubsection{XTC Algorithm}
The idea behind the XTC algorithm \cite[Sec.~3]{Wattenhofer2004} is that a large link weight indicates a low link quality.
A link \linkVariableab is inactive in the output topology of XTC if there exist links of higher quality, \idest, smaller weight, that connect the source with the target of \linkVariableab, possibly via multiple intermediate links.
This is equivalent to the following property:
A link in the output topology of the XTC algorithm is inactive if it is the weight-maximal link in some triangle;
more formally:
$\weight{\linkVariableab} > \max\left(\weight{\linkVariableac}, \weight{\linkVariablecb}\right)$.
In \cite[Sec.~4]{Wattenhofer2004}, the authors refine the XTC algorithm to using the same tie breaking predicate as in \Cref{sec:tie-breaking-predicate}:
Whenever the triangle contains multiple links of the same minimum quality (\idest, of the same maximum weight), only the link with the unique maximum ID is considered.
Therefore, we define the \xtcPredicateLong{} as
\begin{align*}
\xtcPredicate{\linkVariableab, \linkVariableac, \linkVariablecb}&= \weight{\linkVariableab} \geq \max(\weight{\linkVariableac}, \weight{\linkVariablebc})\\
&\;\wedge\; \tieBreakingPredicate{\linkVariableab, \linkVariableac, \linkVariablecb}.
\end{align*}

\subsubsection{Gabriel Graph Algorithm}
The output graph of the Gabriel Graph (GG) algorithm~\cite{Rodoplu1999,Wang08} conforms to the following geometric definition.
A graph is a GG if for each link \linkVariableab, the circle with diameter \linkVariableab and center between \nodeVariablea and \nodeVariableb contains no nodes apart from \nodeVariablea and \nodeVariableb~\cite{Gabriel1969}.
The original formulation is position-based, \idest, each node requires knowledge about its latitude and longitude.
By Thales' theorem~\cite[p.~50]{Agricola2008}, the following equivalent formulation can be obtained.
In each triangle, a link is inactive if its squared weight is smaller than the sum of the squared weights of the other links \linkVariableac and \linkVariablecb:
\begin{align*}
\ggPredicate{\linkVariableab, \linkVariableac, \linkVariablecb} = \weightsquared{\linkVariableab} > \weightsquared{\linkVariableac} +  \weightsquared{\linkVariablecb}.
\end{align*}

\subsubsection{Relative Neighborhood Algorithm}
The output topology of the Relative Neighborhood Graph (RNG) algorithm~\cite{Karp2000,Wang08} is an RNG, which is defined as follows.
In each triangle, the weight of the weight-maximal link \linkVariableab in the triangle must be less than or equal to the weight of the other links \linkVariableac and \linkVariablecb.
This means that a link \linkVariableab is inactive if it is part of a triangle with shorter links \linkVariableac and \linkVariablecb:
\begin{align*}
\rngPredicate{\linkVariableab, \linkVariableac, \linkVariablecb} &= \weight{\linkVariableab} > \weight{\linkVariableac} \wedge \weight{\linkVariableab} > \weight{\linkVariablecb}
\end{align*}
Note that the RNG predicate is similar to XTC, which may lead to the impression that the output topology of XTC is (almost) identical to RNG.
However, this is true only if 
\begin{inparaenum}
\item 
the link weight (as used by XTC) correlates strictly negatively with the Euclidean distance (used by RNG), and
\item 
link-distances are unique because RNG applies the $>$-operator while XTC applies the $\geq$-operator with ID-based tie breaking.
\end{inparaenum}

\subsubsection{Local \ktc Algorithm}
\lktc~\cite{Stein2016b} is a variant of \ktc that is tailored to many-to-one communication scenarios (\eg, data collection).
Here, the hop count attribute \hopcount{\nodeVariablea} stores the length (in hops) of the shortest path from \nodeVariablea to a dedicated base station node.
In the output topology of \lktc, a link \linkVariableab is inactive if
\begin{inparaenum}
\item \linkVariableab fulfills \ktcPredicate{} and
\item if its inactivation does not extend the length of the path to the base station by more than a factor $a$:
\end{inparaenum}
\begin{align}\label{eqn:lStarkTCPredicate}
\begin{split}
\lktcPredicate{\linkVariableab, \linkVariableac, \linkVariablecb} &= \ktcPredicate{\linkVariableab, \linkVariableac, \linkVariablecb}  \\
& \wedge \min(\hopcount{\nodeVariablea}, \hopcount{\nodeVariableb}, \hopcount{\nodeVariablec}) \geq 0\\
& \Big(\hopcount{\nodeVariablea} = \hopcount{\nodeVariableb} \Rightarrow true\\
& \wedge \hopcount{\nodeVariablea} > \hopcount{\nodeVariableb} \Rightarrow \frac{\hopcount{\nodeVariablec} + 1}{\max(1, \hopcount{\nodeVariablea})} < a\\
& \wedge \hopcount{\nodeVariablea} < \hopcount{\nodeVariableb} \Rightarrow \frac{\hopcount{\nodeVariablec} + 1}{\max(1, \hopcount{\nodeVariableb})} < a
\Big)
\end{split}
\end{align}
The second line of \Cref{eqn:lStarkTCPredicate} ensures that the hop count is defined for each participating node.
The third line covers the case that \nodeVariablea and \nodeVariableb have the same distance to the base station.
The fourth line considers the case when \nodeVariablea is farther away from the base station than \nodeVariableb.
In this case, we may estimate the path length after inactivating \linkVariableab as $\hopcount{\nodeVariablec} + 1$.
As \nodeVariablea may be the base station itself (\idest, $\hopcount{\nodeVariablea} = 0$), we ensure that the denominator is always at least 1.
The fifth line is symmetric to the fourth line.
More details of the algorithm can be found in \cite[p.5]{Stein2016b}.

\subsubsection{Yao Graph Algorithm}
\label{sec:yao-algorithm}
The Yao graph algorithm~\cite{Yao1982} is the only location-dependent \TC algorithm considered in this paper.
This means that it requires information about the latitude and longitude of each node.
The Yao graph algorithm separates the environment of a node into \emph{cone}s of uniform angle.
If we denote the cone count with \coneCount, each cone covers an angle of $\frac{\SI{360}{\degree}}{\coneCount}$.
A link \linkVariableac is inactive if a link \linkVariableac in the same cone exists that has a smaller weight:
\newcommand{\alphaCone}{\ensuremath{\alpha_\text{\,cone}}}
\begin{align*}
\yaoPredicate{\linkVariableab, \linkVariableac, \linkVariablecb}&= \weight{\linkVariableab} > \weight{\linkVariableac}\\
&\;\wedge\;\exists x \in \{1, 2, \dots, \coneCount\}: \\
&\Big(\alphaCone \cdot (x - 1) \leq \angleCmd{\linkVariableab}  < \alphaCone \cdot x\\
&\;\wedge\;\alphaCone \cdot (x - 1) \leq \angleCmd{\linkVariableac}  < \alphaCone \cdot x
\Big) \\
&\text{ with } \alphaCone = \frac{360\degree}{\coneCount}\\	
\end{align*}

\subsubsection{\ektc Algorithm}\label{sec:ektc}

As the last \TC algorithm considered in this paper, we derive a novel, energy-aware variant of \ktc, called \ektc.
Its distinctive feature is that it considers the remaining energy of nodes.
We begin with an illustrative example that highlights a situation in which the output topology of \ktc is suboptimal \wrt network lifetime.
Afterwards, we present \ektc and show that it improves the network lifetime of the sample topology.

\paragraph{Network Lifetime}
In the \WSN community, extending the lifetime of a network is a key optimization goal.
There are many alternative definitions of network lifetime~\cite{Yunxia2005}.
We apply a definition that is tailored to the per-node lifetime.
A node \nodeVariablea is \emph{alive} if its remaining energy is positive, \idest, $\energy{\nodeVariablea}{} > 0$.
Likewise, a node \nodeVariablea is \emph{dead} if its battery is empty, \idest, $\energy{\nodeVariablea}{} = 0$.
The \emph{$d$-lifetime $L_d$} of a network is defined as the first point in time at which at least $d$ nodes are dead.\footnote{The variable $d$ alludes to the metric \emph{dead node count}.}
The following values of $d$ are of special interest:
\begin{itemize}
    \item $d = 1$ because the network is fully intact before $L_1$ and 
    \item $d = |V|$ because the network is no longer operational after $L_{|V|}$.
\end{itemize}
As a shorthand, $L_{x \si{\percent}}$ denotes the point in time when $x \si{\percent} \cdot |V|$ nodes are dead, \idest, $L_{x \si{\percent}} = L_{x \si{\percent} \cdot |V|}$.
For an energy-aware \TC algorithm, it is important to estimate the remaining lifetime of the topology.
This allows the \TC algorithm to proactively relieve nodes that would otherwise fail soon.
The \emph{expected $d$-lifetime \expectedRemainingLifetime{d}{G} of a topology $G$} estimates the $d$-lifetime of the topology.
For simplicity, we focus on \expectedRemainingLifetime{1}{G} in this paper.
The \emph{\expectedPowerLong{}{\linkVariableab}} for each link \linkVariableab represents the power that is required to reach \nodeVariableb from \nodeVariablea.
According to Friis' free space propagation model, \expectedPower{}{\linkVariableab} grows at least proportionally to the squared distance (here: weight) of \linkVariableab~\cite{Fri46}:
\begin{align*}
    \expectedPower{}{\linkVariableab} &\propto \weightsquared{\linkVariableab}
\end{align*}
We may estimate the \emph{\expectedRemainingLifetimeLong{}{\linkVariableab} of a node \nodeVariablea \wrt a link \linkVariableab} as follows:
\begin{align}\label{eqn:ExpectedRemainingLifetimePerLink}
\expectedRemainingLifetime{}{\linkVariableab} &= \frac{\energy{\nodeVariablea}{}}{\expectedPower{}{\linkVariableab}}
\end{align}
Here, we estimate the number of messages that can be transmitted with the remaining energy \energy{\nodeVariablea}{} of \nodeVariablea.
\Cref{eqn:ExpectedRemainingLifetimePerLink} presumes that transmitting a message consumes energy only at the sending node.
On real hardware, transmitting a message will also consume energy at the receiving node.
In this paper, we neglect this additional cost for simplicity, which is common in the network community (\eg, \cite{XH12}).
We define the \emph{\expectedRemainingLifetimeLong{}{\nodeVariablea} of a node \nodeVariablea} as the minimum \expectedRemainingLifetimeLong{}{\linkVariableab} of its outgoing links:
\begin{align*}
\expectedRemainingLifetime{}{\nodeVariablea} &= \min_{\linkVariableab \in E} \expectedRemainingLifetime{}{\linkVariableab}
\end{align*}
Finally, we lift this definition to topologies:
The \emph{\expectedRemainingLifetimeLong{1}{G} of a topology $G$} is the minimum \expectedRemainingLifetimeLong{}{\nodeVariablea} of its nodes:
\begin{align*}
\expectedRemainingLifetime{1}{G} &= \min_{\nodeVariablea \in V} \expectedRemainingLifetime{1}{\nodeVariablea}
\end{align*}

\paragraph{Motivating example with \ktc}

In the following, we analyze the \remainingLifetimeLong{1}{} of the sample topology shown in \Cref{fig:ktc-example-with-energy-consumption-A}.
Each node \nodeVariable is annotated with its remaining energy \energy{\nodeVariable}{}, and each link \linkVariableab is annotated with its weight \weight{\linkVariableab} and its \expectedRemainingLifetimeLong{1}{\linkVariableab}.
We simulate the behavior of the network over a number of discrete time steps until the first node runs out of energy, using the following workload.
In each time step, the remaining energy of a node decreases by the greatest required transmission power among all of its active outgoing links.
This is a simplified simulation, \eg, of a Gossiping protocol \cite{Jelasity2011} that broadcasts a message in each time step from a node to all its neighbors.
For instance, in each time step, the remaining energy of node \nodeName{4} decreases by 9.
\begin{figure}
    \begin{center}
       
        \subcaptionbox{Initial topology\label{fig:ktc-example-with-energy-consumption-A} }[.32\textwidth]{\includegraphics[width=.32\textwidth]{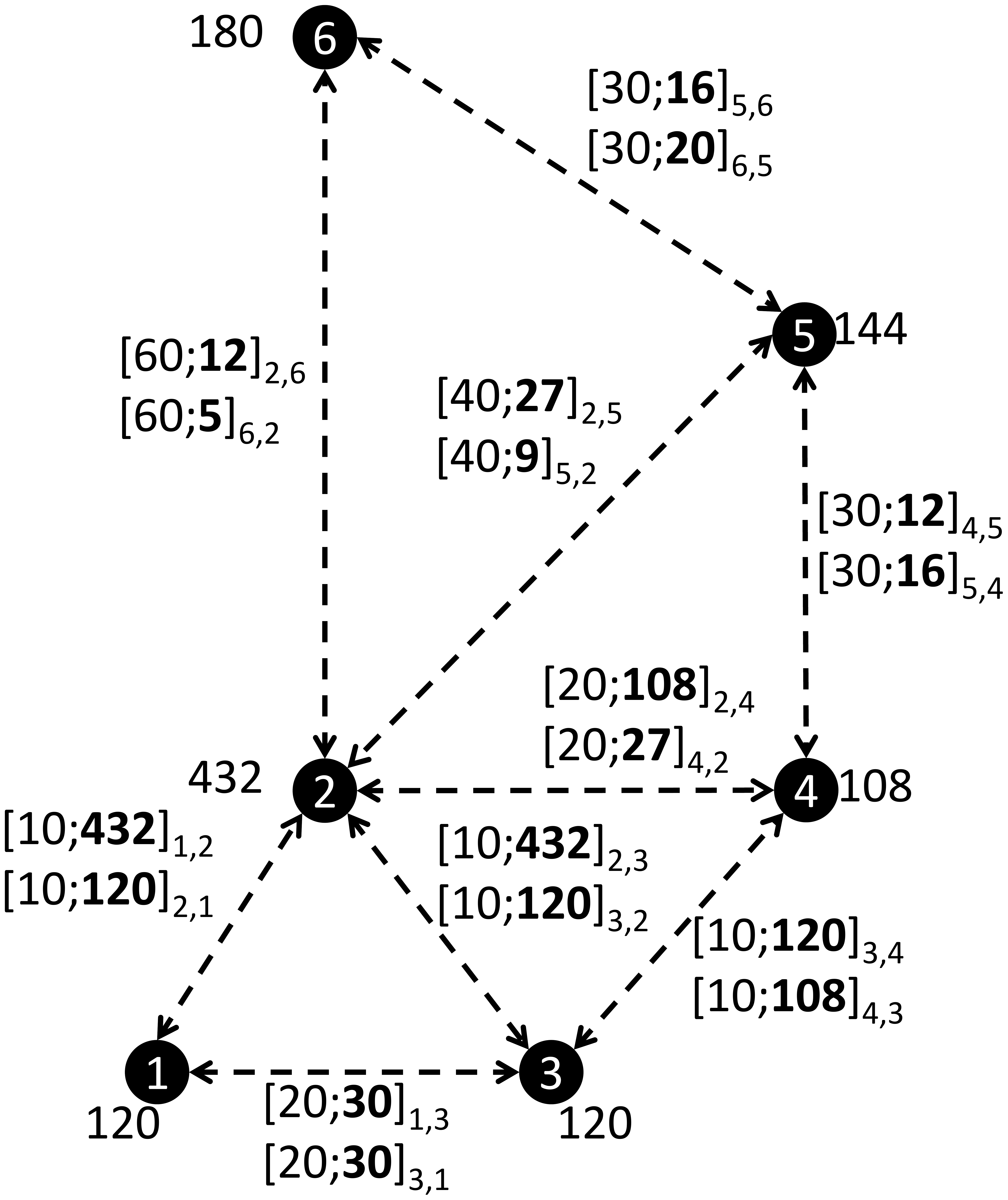}}
        \subcaptionbox{After executing \ktc\label{fig:ktc-example-with-energy-consumption-B}}[.32\textwidth]{\includegraphics[width=.32\textwidth]{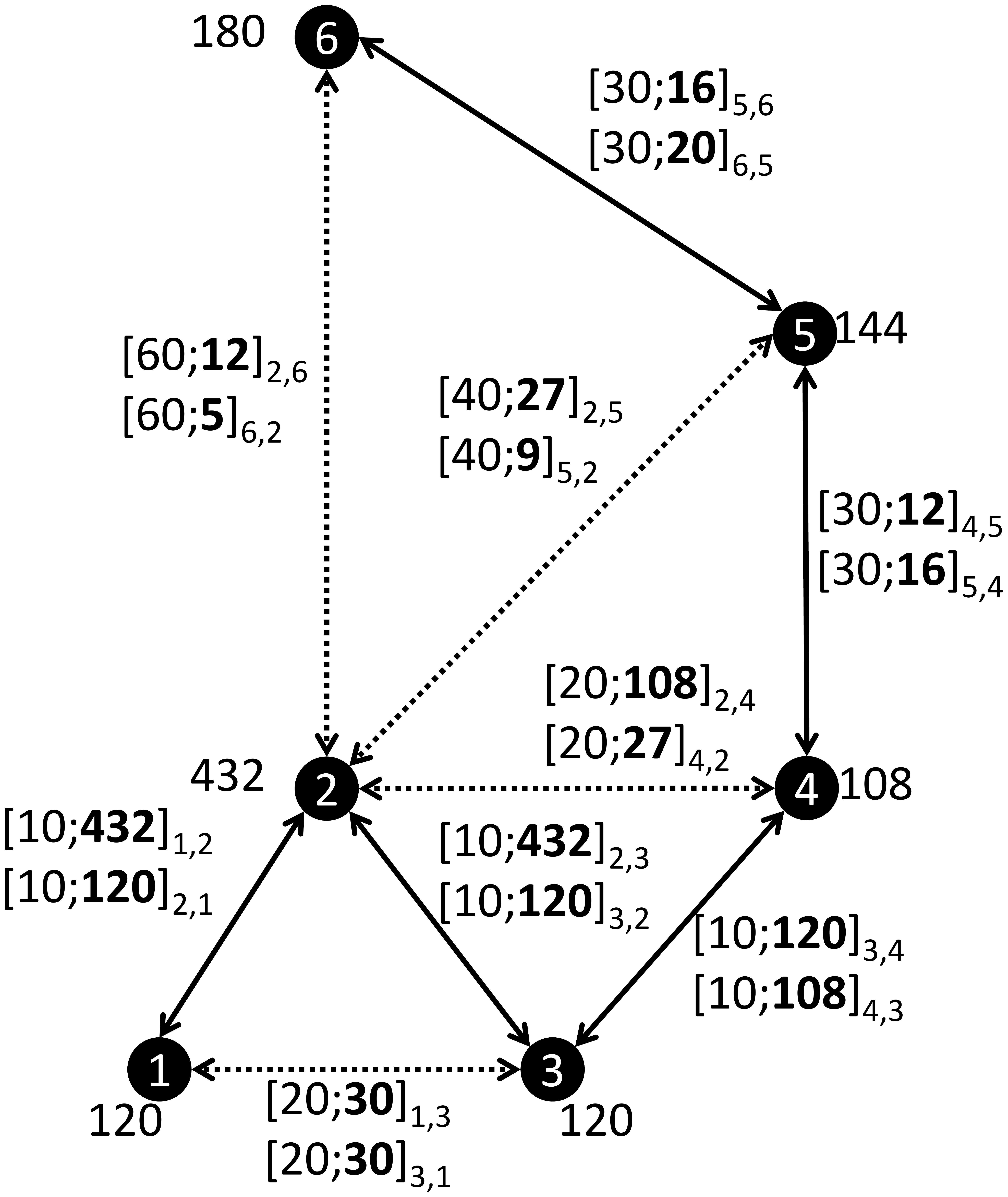}}
        \subcaptionbox{After 12 time steps\label{fig:ktc-example-with-energy-consumption-C}}[.32\textwidth]{\includegraphics[width=.32\textwidth]{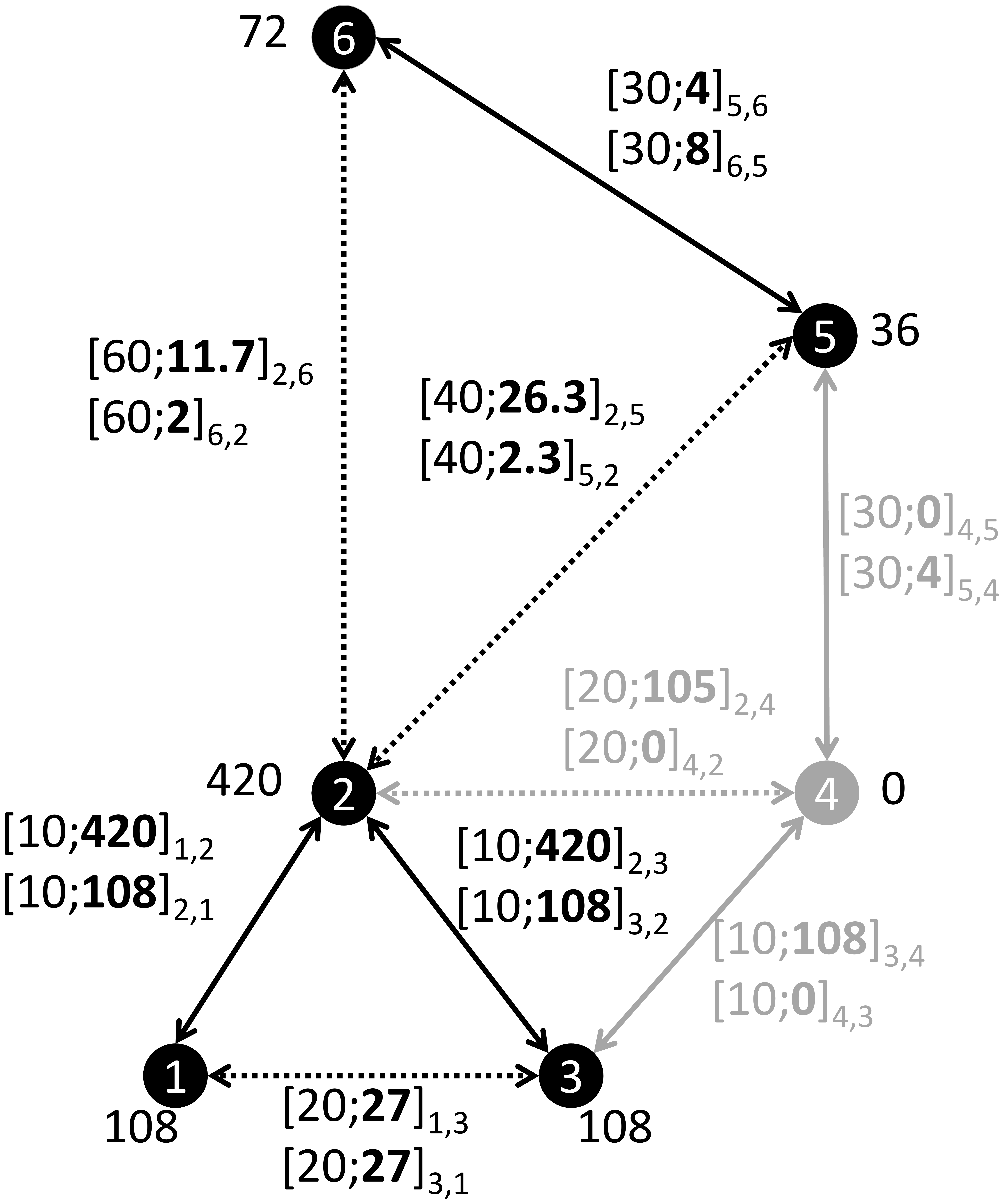}}

        \vspace{1em}
        \subcaptionbox{Legend\label{fig:ktc-example-with-energy-consumption-legend}}[.32\textwidth]{\includegraphics[width=.25\textwidth]{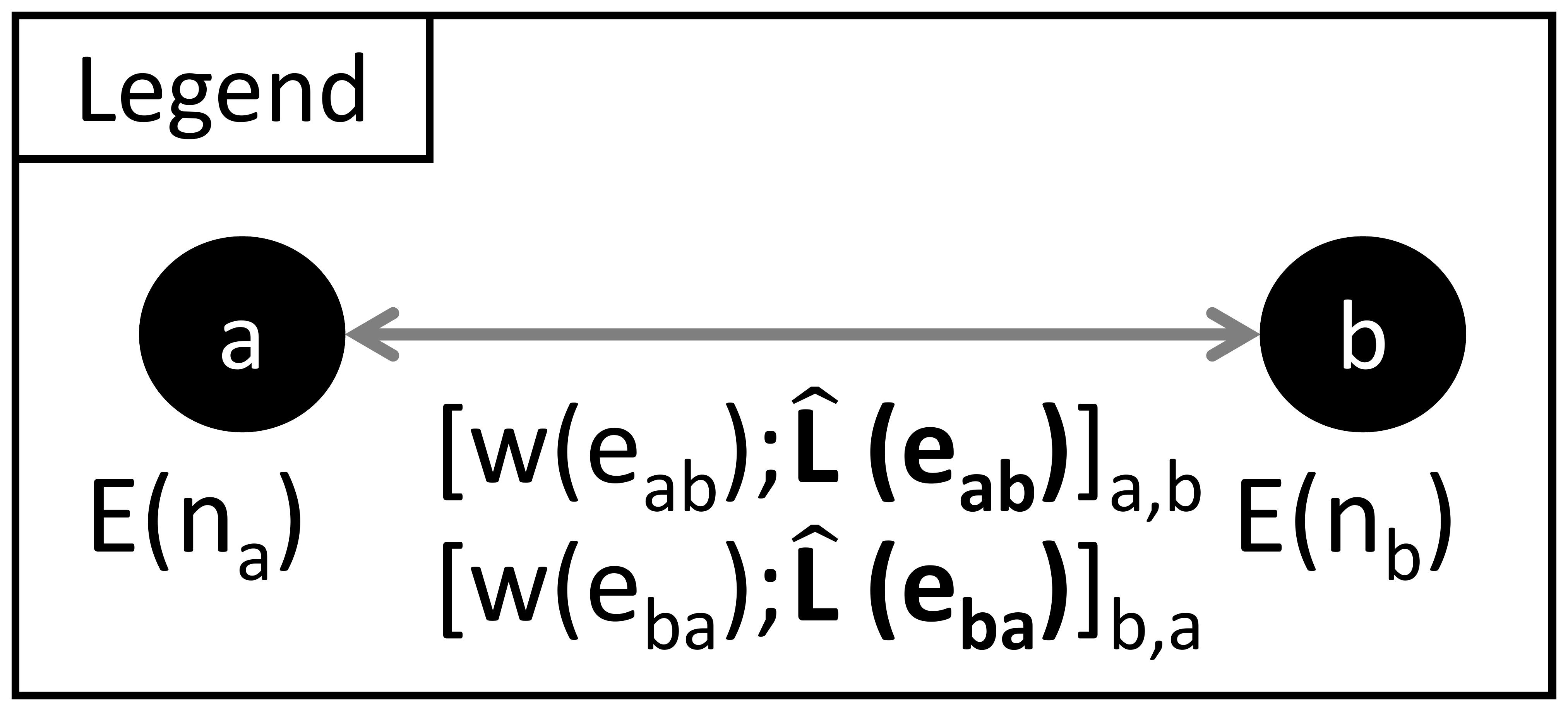}}
    \end{center}
    \caption{Applying \ktc to the sample topology ($k = 2$, $\remainingLifetime{1}{} = 12$)}
    \label{fig:ktc-example-with-energy-consumption}
\end{figure}
Only one execution of \ktc ($k = 2$) is required because the link weights are constant (\Cref{fig:ktc-example-with-energy-consumption-B}).
The example shows that \nodeName{4} is the first node to run out of energy after 12 time steps, \idest, $\remainingLifetime{1}{G} = 12$ (\Cref{fig:ktc-example-with-energy-consumption-C}).

\paragraph{Specification of \ektc}
The problem of unbalanced energy consumption is well-known in the \WSN literature.
A number of \TC algorithms have been proposed that explicitly take the remaining energy of nodes into account (\eg,~\cite{XH12,KMD07}).
Inspired by \cite{XH12}, we propose to modify the predicate of \ktc to take  the expected remaining lifetime of nodes into account.
We call this energy-aware variant \emph{\ektc}.
After executing \ektc, a link is inactive if and only if this link is part of a triangle in which it has the minimum expected remaining lifetime among the links in the triangle and if its expected remaining lifetime is at least \ktcParameterK times shorter than the maximum expected remaining lifetime of the other links in the triangle:
\begin{align}\label{eqn:ektc-inactive-link-condition}
\begin{split}
\ektcPredicate{\linkVariableab, \linkVariableac, \linkVariablecb} \;& = \;
\expectedRemainingLifetime{1}{\linkVariableab} \leq \min(\expectedRemainingLifetime{1}{\linkVariableac}, \expectedRemainingLifetime{1}{\linkVariablecb})\\
\;&\wedge \; \expectedRemainingLifetime{1}{\linkVariableab} \leq k \cdot \max(\expectedRemainingLifetime{1}{\linkVariableac}, \expectedRemainingLifetime{1}{\linkVariablecb}).
\end{split}
\end{align}

\paragraph{Motivating example with \ektc}

\Cref{fig:ektc-example} illustrates the processing of the same topology as in \Cref{fig:ktc-example-with-energy-consumption} with \ektc ($k=2$).
\begin{figure}
    \begin{center}
        \hfill
        \subcaptionbox{Initial topology\label{fig:ektc-example-A} }[.32\textwidth]{\includegraphics[width=.32\textwidth]{figures/pdf/example-battery-depletion-ktc-A_cropped.pdf}}
        \subcaptionbox{After executing \ektc \label{fig:ektc-example-B}}[.32\textwidth]{\includegraphics[width=.32\textwidth]{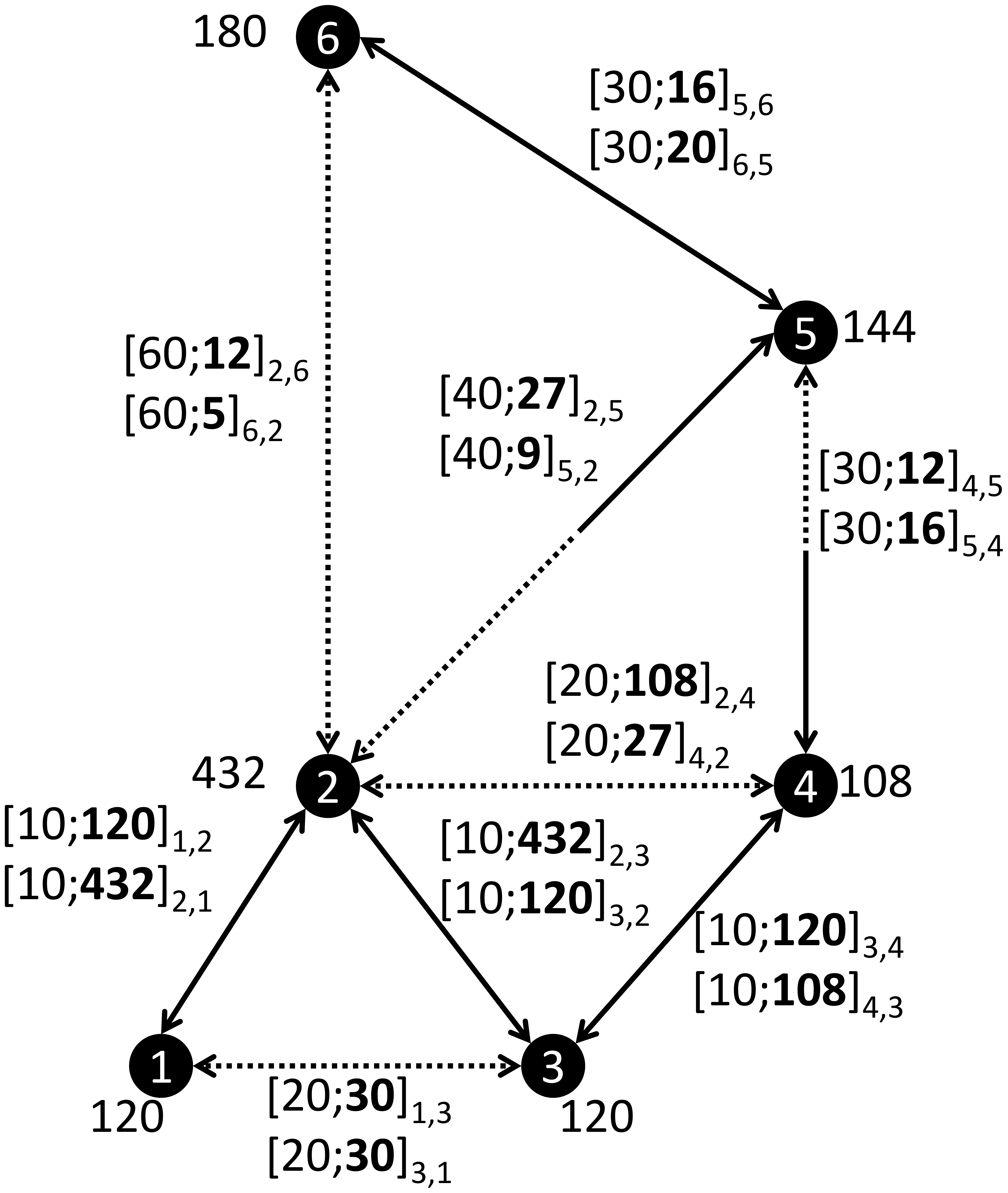}}
        \subcaptionbox{After 12 time steps\label{fig:ektc-example-C}}[.32\textwidth]{\includegraphics[width=.32\textwidth]{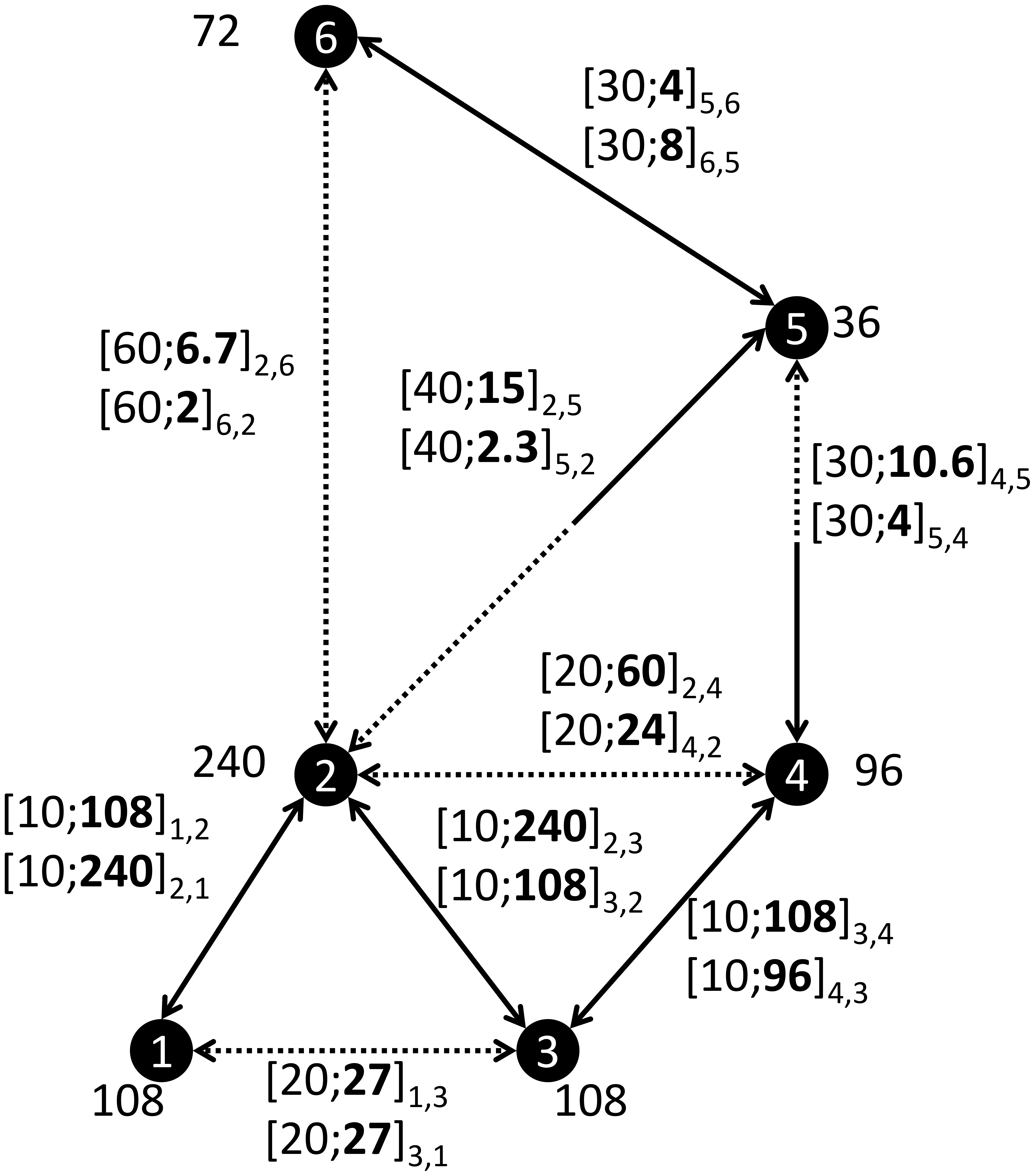}}
        
        \vspace{1em}
        \subcaptionbox{Legend\label{fig:ektc-example-with-energy-consumption-legend}}[.32\textwidth]{\includegraphics[width=.25\textwidth]{figures/pdf/example-battery-depletion-legend_cropped.pdf}}
    \end{center}
    \caption{Applying \ektc to the sample topology ($k = 2$, $\remainingLifetime{1}{} = 16$)}
    \label{fig:ektc-example}
\end{figure}
The expected remaining lifetime of a link changes in each time step and opposite links may have different states.
To establish comparability with \ktc, we invoke \ektc only once in the beginning (\Cref{fig:ektc-example-B}).
After 12 time steps, the remaining energy of node \nodeName{4} is now~96 (\Cref{fig:ektc-example-C}), and the minimal remaining energy among all nodes is~36.
Therefore, executing \ektc increases the \remainingLifetimeLong{1}{} of the sample topology.

\paragraph{Lifetime preservation with \ektc}
In fact, the benefit of applying \ektc in the previous example can be generalized:
For each triangle of links, consisting of \linkVariableab, \linkVariableac, and \linkVariablecb, \ektc preserves the expected lifetime of \nodeVariablea without decreasing the lifetime of \nodeVariableb and \nodeVariablec, compared to applying \MaxpowerTC.
The expected lifetime of \nodeVariablea even increases if the inactivated link \linkVariableab has the minimum expected lifetime among all outgoing links of \nodeVariablea.

\subsubsection{The Minimum-Weight Predicate \minWeightPredicate{}}
\label{sec:min-weight-predicate}
In the following, we introduce the novel \minWeightPredicateLong{}, which can be combined freely with all specified \TC algorithms.
It serves to reduce the memory footprint and the runtime of the \TC algorithm.
Working memory is a highly limited resource on wireless sensor nodes.
For this reason, keeping the entire neighborhood of a sensor node in working memory may be infeasible if the topology is dense.
Fortunately, it is often unnecessary to store links to close neighbors because the energy consumption is typically predominated by links to distant neighbors.
Additionally, reducing the size of the processed neighborhood may speed up the planning step.

The following \emph{minimum-weight predicate} \minWeightPredicate{} formalizes this reduction step.
The parameter \weightThreshold represents the configurable minimal weight of a link to be included in the considered neighborhood.
\begin{align}\label{eqn:minimum-weight-predicate}
\begin{split}
\minWeightPredicate{\linkVariableab, \linkVariableac, \linkVariablecb} 
&= \min\left(\weight{\linkVariableab}, \weight{\linkVariableac}, \weight{\linkVariablecb}\right) \geq \weightThreshold\\
\end{split}
\end{align}
This predicate may now be used to compose new variants of the previously specified \TC algorithms.
For instance, a modified version of \ktc with minimum-weight predicate is
\begin{align*}
\ktcMinWeightPredicate{\linkVariableab, \linkVariableac, \linkVariablecb} = \ktcPredicate{\linkVariableab, \linkVariableac, \linkVariablecb} \wedge \minWeightPredicate{\linkVariableab, \linkVariableac, \linkVariablecb}.\\ 
\end{align*}

\subsection{Summary of \TC Algorithms}
\Cref{fig:tc-algorithms-feature-model} and \Cref{tab:attribute-predicates-all} summarize the results of this section.
\Cref{fig:tc-algorithms-feature-model} illustrates the configuration options of a single sensor node:
the possible \TC algorithms, the minimum-weight optimization, and the relevant node and link properties.
Black solid lines indicate the hierarchical decomposition relation between configuration options.
Gray lines indicate dependencies from \TC algorithms to properties.
Dashed black lines indicate the relationship between the components of the considered \TC algorithms and their corresponding predicate.

The output topology of any \TC algorithm must fulfill the complete-classification predicate $\completeClassificationPredicate$ and the A-connectivity predicate $\AConnectivityPredicate{}$ (\Cref{sec:general-properties}).
The triangle-identifying predicate \trianglePredicate{} is common to all considered algorithms except for \MaxpowerTC.
In contrast, the predicate $\tcPredicate{A}{} \in \{\ktcPredicate{}, \maxpowerPredicate{}, \xtcPredicate{}, \ggPredicate{}, \rngPredicate{}, \lktcPredicate{}, \yaoPredicate{}, \ektcPredicate{}\}$ describes the algorithm-specific attribute constraints.
The auxiliary predicates \tieBreakingPredicate{} and \minWeightPredicate{} describe the ID-based tie-breaking (\Cref{sec:tie-breaking-predicate}) and the minimum-weight optimization (\Cref{sec:min-weight-predicate}).
\begin{figure}
    \begin{center}
        \includegraphics[width=\textwidth]{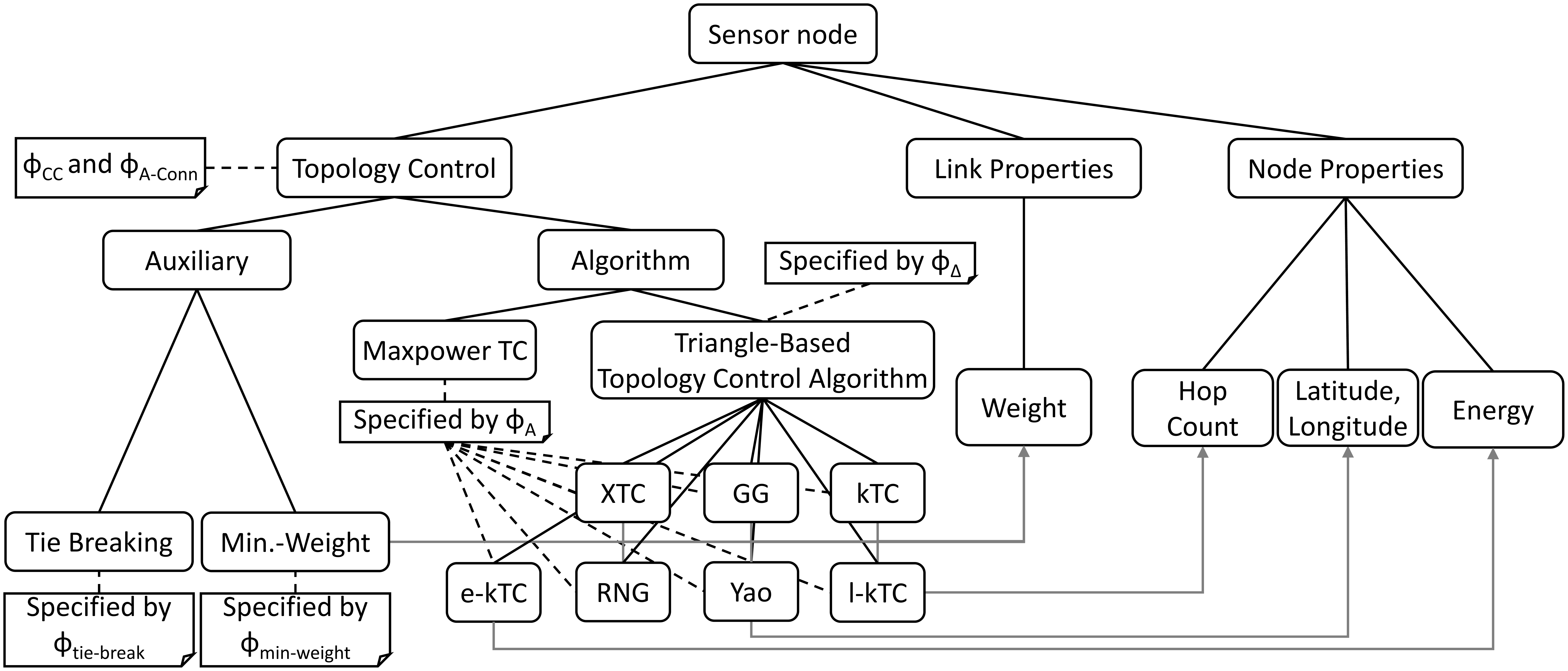}
        \caption{Overview of configuration options of a sensor node}
        \label{fig:tc-algorithms-feature-model}
    \end{center}
\end{figure}
\Cref{tab:attribute-predicates-all} summarizes the algorithm-specific predicates \tcPredicate{A}{}.
For simplicity, we write \tcPredicate{}{} instead of \tcPredicate{A}{} in the following.
\begin{table}
    \centering
    \caption{Attribute predicates of the considered \TC algorithms}
    \label{tab:attribute-predicates-all}
    \begin{tabular}{p{4cm}p{6.8cm}}
        \toprule
        \textbf{Algorithm $A$}&
        \textbf{Algorithm-Specific Predicate} \tcPredicate{A}{\linkVariableab, \linkVariableac, \linkVariablecb}
        \\
        \midrule
        \MaxpowerTC 
        &
        {\emph{false}}
        \\
        \midrule
        XTC \cite{Wattenhofer2004}
        &
        {
            $\begin{aligned}
            &\weight{\linkVariableab} \geq \max(\weight{\linkVariableac}, \weight{\linkVariablecb})\\
            &\;\wedge\;\left(\weight{\linkVariableab} = \weight{\linkVariableac}  
            \Rightarrow \id{\linkVariableab} > \id{\linkVariableac}\right)\\
            &\;\wedge\;\left(\weight{\linkVariableab} = \weight{\linkVariablecb}
            \Rightarrow \id{\linkVariableab} > \id{\linkVariablecb}\right)\\
            \end{aligned}$}
        \\
        \midrule
        GG \cite{Wang08}
        &
        {$\begin{aligned}
            \weightsquared{\linkVariableab} > \weightsquared{\linkVariableac} +  \weightsquared{\linkVariablecb}
            \end{aligned}$}\\
        \midrule
        RNG~\cite{Karp2000}
        &
        {$\begin{aligned}
            \weight{\linkVariableab} > \max(\weight{\linkVariableac}, \weight{\linkVariablecb})
            \end{aligned}$}\\
        \midrule
        \ktc \cite{SWBM12}
        &
        {$\begin{aligned}
            &\weight{\linkVariableab} \geq \max(\weight{\linkVariableac}, \weight{\linkVariablecb})\\
            &\;\wedge\;\weight{\linkVariableab} \geq k \cdot \min(\weight{\linkVariableac}, \weight{\linkVariablecb})\\
            &\;\wedge\;\left(\weight{\linkVariableab} = \weight{\linkVariableac}  
            \Rightarrow \id{\linkVariableab} > \id{\linkVariableac}\right)\\
            &\;\wedge\;\left(\weight{\linkVariableab} = \weight{\linkVariablecb}
            \Rightarrow \id{\linkVariableab} > \id{\linkVariablecb}\right)\\
            \end{aligned}$}\\
        \midrule
        \lktc \cite{Stein2016b}
        &
        {$\begin{aligned}
        &\ktcPredicate{\linkVariableab, \linkVariableac, \linkVariablecb}  \\
        &\wedge \min(\hopcount{\nodeVariablea}, \hopcount{\nodeVariableb}, \hopcount{\nodeVariablec}) \geq 0\\
        &\Big(\hopcount{\nodeVariablea} = \hopcount{\nodeVariableb} \Rightarrow true\\
        &\wedge \hopcount{\nodeVariablea} > \hopcount{\nodeVariableb} \Rightarrow \frac{\hopcount{\nodeVariablec} + 1}{\max(1, \hopcount{\nodeVariablea})} < a\\
        &\wedge \hopcount{\nodeVariablea} < \hopcount{\nodeVariableb} \Rightarrow \frac{\hopcount{\nodeVariablec} + 1}{\max(1, \hopcount{\nodeVariableb})} < a
        \Big)\end{aligned}$
        }\\
        \midrule
        Yao graph~\cite{Yao1982}
        &
        {$\begin{aligned}
            &\weight{\linkVariableab} > \weight{\linkVariableac}\;\wedge\\
            &\exists x \in \{1, 2, \dots, \coneCount\}: \\
            &\big(\; \frac{\SI{360}{\degree}}{\coneCount} \cdot (x - 1) \leq \angleCmd{\linkVariableab}  < \frac{\SI{360}{\degree}}{\coneCount} \cdot x \\	
            &\;\wedge\;\frac{\SI{360}{\degree}}{\coneCount} \cdot (x - 1) \leq \angleCmd{\linkVariableac}  < \frac{\SI{360}{\degree}}{\coneCount} \cdot x\;\big)\\	
            \end{aligned}$}\\
        \midrule
        \ektc \lbrack this paper\rbrack&
        \multicolumn{1}{l}{
            $\begin{aligned}
            &\expectedRemainingLifetime{1}{\linkVariableab} < \min(\expectedRemainingLifetime{1}{\linkVariableac}, \expectedRemainingLifetime{1}{\linkVariablecb})\\
            &\expectedRemainingLifetime{1}{\linkVariableab} \leq k \cdot \max(\expectedRemainingLifetime{1}{\linkVariableac}, \expectedRemainingLifetime{1}{\linkVariablecb})\\
            \end{aligned}$}\\
        \bottomrule
    \end{tabular}
\end{table}

\section{Characterizing Valid Topologies with Graph Constraints}
\label{sec:constraints}

In this section, we introduce graph patterns and graph constraints to specify the desired formal properties of topologies.
Compared to first-order logic predicates, the main advantage of graph constraints is that they can be constructively combined with \GT rules to produce refined constraint-preserving \GT rules.
The specification using graph constraints allows us to prove that all considered \TC algorithms produce connected output topologies.
As one main contributions of this paper, we lift the proof of connectivity to entire families of \TC algorithms.

\subsection{Graph Constraint Concepts}
A \emph{pattern} is a graph consisting of node and link variables together with a set of attribute constraints. A \emph{node (link) variable} serves as a placeholder for a node (link) in a topology.
An \emph{attribute constraint} is a predicate over attributes of node and link variables. 
A \emph{match \match{} of a pattern \pattern in a topology \topology} injectively maps the node and link variables of \pattern to the nodes and links of \topology, respectively, such that all attribute constraints are fulfilled.
Additionally, a match must preserve the end nodes of link variables, \idest, map the incident node variables of each link variable to the incident nodes of the corresponding link.

A \emph{graph constraint} \constraint{x} consists of a \emph{premise} pattern \premise{x}{} and a \emph{conclusion} \conclusion{x}{} that consists of zero or more \emph{conclusion pattern}s \conclusion{x,y}{}.
The premise is a subgraph of each conclusion pattern, and the attribute constraints of each conclusion pattern jointly imply the attribute constraints of the premise.
A graph constraint with no (at least one) conclusion patterns is called \emph{negative (positive) graph constraint}.
A \emph{graph constraint \constraint{x} is fulfilled on a topology \topology} (alternatively: a topology \topology fulfills a graph constraint \constraint{x}) if every match of its premise \premise{x}{} can be extended to a match of at least one conclusion pattern \conclusion{x,y}{}.
This implies that a topology fulfills a negative graph constraint \constraint{x} if it does not contain \emph{any} matches of \premise{x}{}.
Note that this requirement is sufficient but not necessary for positive graph constraints.
A \emph{graph constraint \constraint{x} is fulfilled at a match \match{} of a pattern \pattern in a topology \topology} if every match of \premise{x}{} that maps all node (link) variables that appear also in \pattern to the same nodes (links) in \topology can be extended to a match of at least one conclusion pattern \conclusion{x,y}{}.
A \emph{match~$\match{}^{\prime}$ of a pattern~$\pattern^\prime$ in \topology extends a match~\match{} of a pattern~\pattern} if every node (link) variable that appears in $\pattern^\prime$ and in \pattern is bound to the same node (link), respectively.

\paragraph{No-unclassified-links constraint \unclassifiedLinkConstraint}
The \unclassifiedLinkConstraintLong is equivalent to the complete-classification property $\completeClassificationPredicate$, which must be fulfilled by every \TC algorithm (\Cref{eqn:no-unclassified-links}).
Its premise \premiseUnclassifiedLinkConstraint matches any unclassified link \linkVariableab, and its conclusion is empty because it is a negative constraint.
This means that \unclassifiedLinkConstraint can only be fulfilled on a topology if this topology does not contain any unclassified links.
\begin{figure}
    \begin{center}
        \includegraphics[width=.3\textwidth]{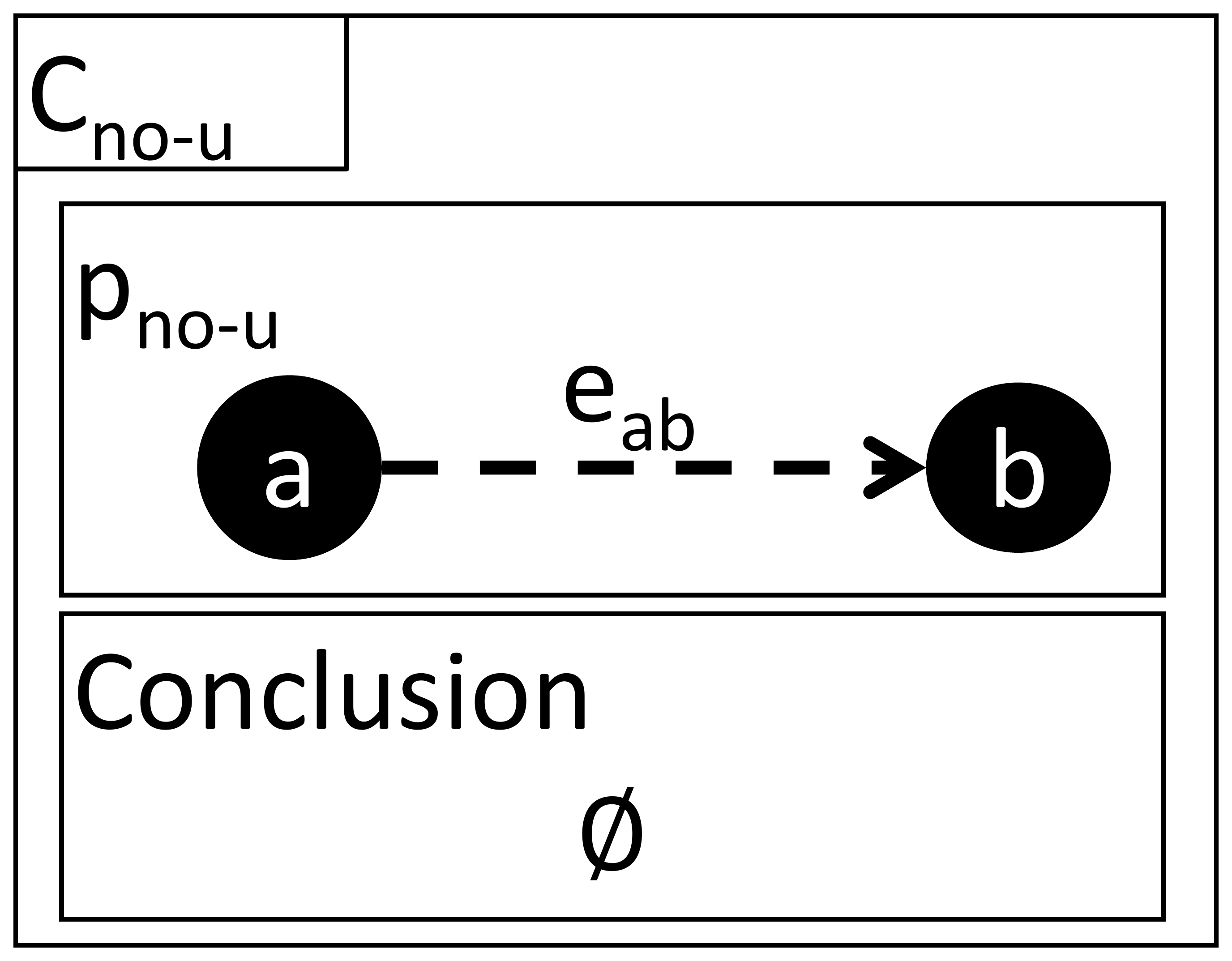}
        \caption{No-unclassified-links constraint \unclassifiedLinkConstraint}
        \label{fig:unclassified-link-constraint}
    \end{center}
\end{figure}

\paragraph{Algorithm-specific graph constraints}
\Cref{fig:generic-graph-constraints} shows two graph constraints that describe the algorithm-specific properties of a triangle-based \TC algorithm.
The \emph{\inactiveLinkConstraintKTCLong} (\Cref{fig:generic-graph-constraints-inactive}) states that each inactive link \linkVariableab must be part of a triangle together with classified links \linkVariableac and \linkVariablecb for which \tcPredicate{}{\linkVariableab, \linkVariableac, \linkVariablecb} is true.
Symmetrically, the \emph{\activeLinkConstraintKTCLong} (\Cref{fig:generic-graph-constraints-active}) states that each active link \linkVariableab may not be part of a triangle together with classified links \linkVariableac and \linkVariablecb for which \tcPredicate{}{\linkVariableab, \linkVariableac, \linkVariablecb} is true.
\begin{figure}
    \begin{center}
        \begin{subfigure}[t]{.3\textwidth}
            \includegraphics[width=.95\textwidth]{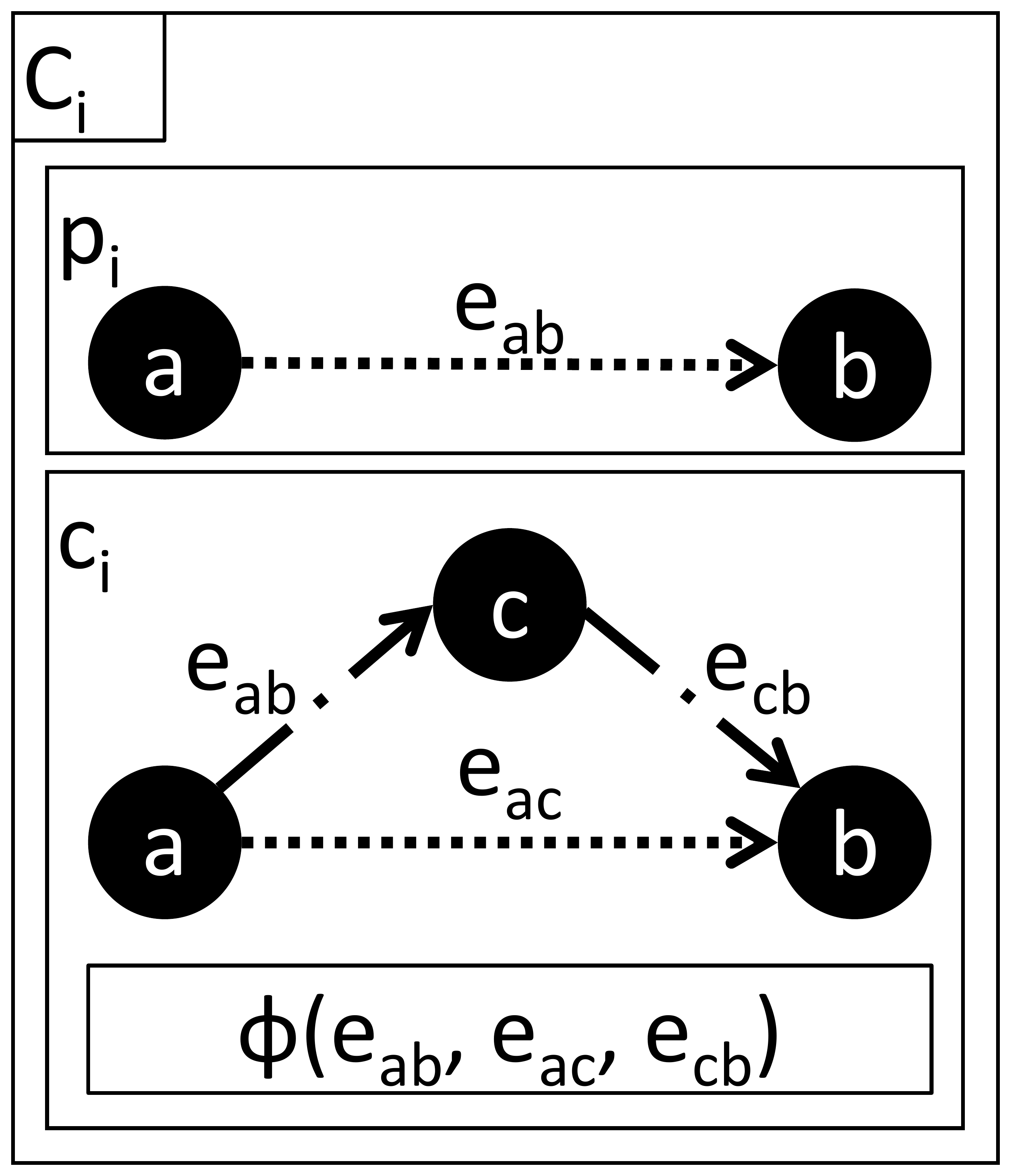}
            \subcaption{\xmakefirstuc{\inactiveLinkConstraintKTCLong}}
            \label{fig:generic-graph-constraints-inactive}
        \end{subfigure}
        \begin{subfigure}[t]{.3\textwidth}
            \includegraphics[width=.95\textwidth]{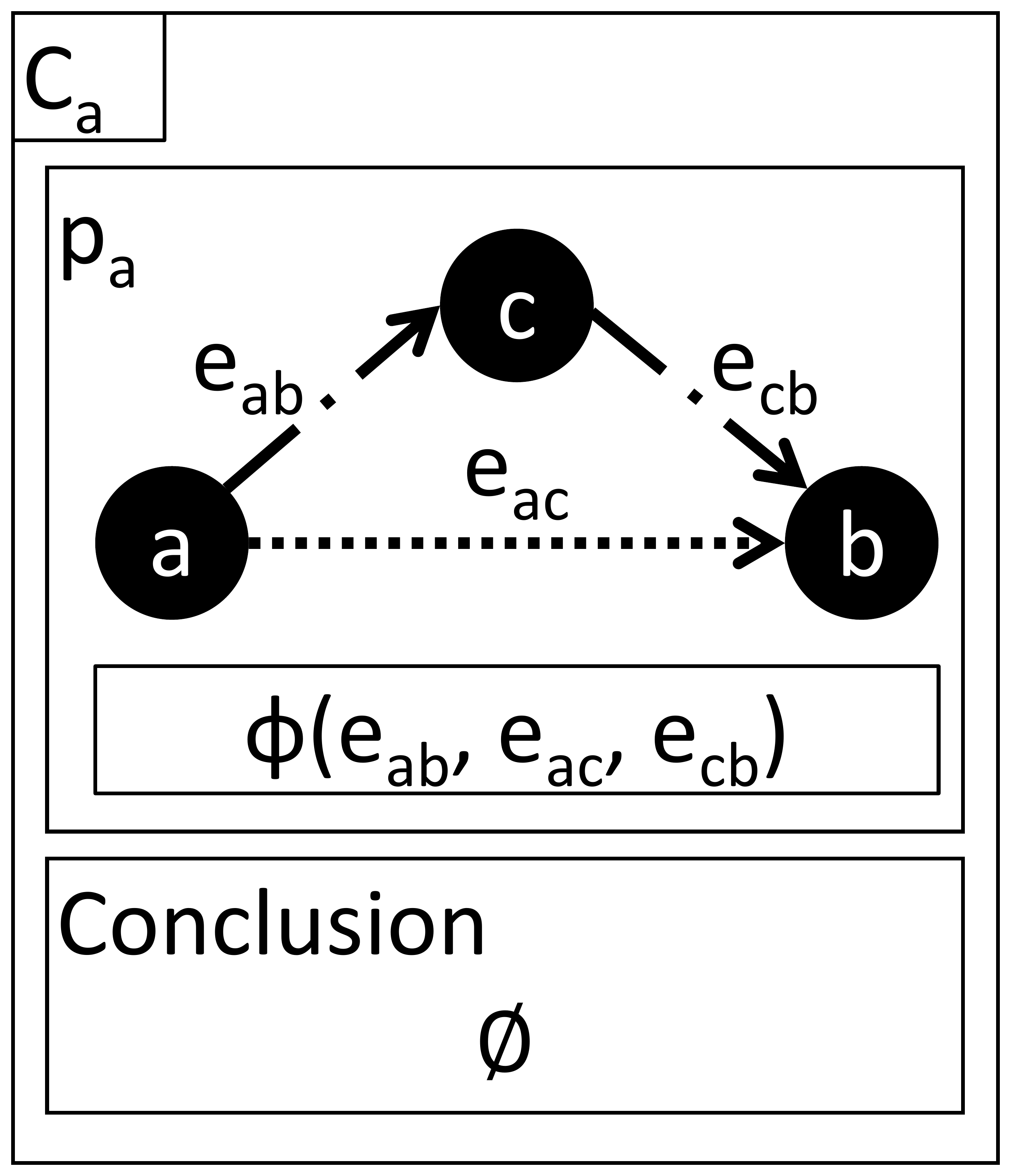}
            \subcaption{\xmakefirstuc{\activeLinkConstraintKTCLong}}
            \label{fig:generic-graph-constraints-active}
        \end{subfigure}
        
        \caption{
            Abstract algorithm-specific graph constraints for triangle-based \TC algorithms
        }
        \label{fig:generic-graph-constraints}
    \end{center}
\end{figure}

\subsection{Consistency of Topologies}
In the following, we categorize topologies according to which of the specified graph constraints they fulfill.
We distinguish between the following three levels of \emph{consistency}:
\begin{itemize}
\item 
A topology is \emph{strongly consistent} if it fulfills the \unclassifiedLinkConstraintLong, the \inactiveLinkConstraintKTCLong, and the \activeLinkConstraintKTCLong.
A valid output topology of a \TC algorithm must be strongly consistent.

\item 
A topology is \emph{weakly consistent} if it fulfills the \inactiveLinkConstraintKTCLong and the \activeLinkConstraintKTCLong.
For instance, an entirely unclassified topology is weakly consistent because it contains matches of the premises of neither \inactiveLinkConstraintKTC nor \activeLinkConstraintKTC.

\item 
A topology is \emph{inconsistent} if it is neither weakly nor strongly consistent, \idest, if it fails to fulfill at least one of the algorithm-specific constraints \inactiveLinkConstraintKTC and \activeLinkConstraintKTC.
\end{itemize}
In the following, we require that the topology is weakly consistent before and strongly consistent after invoking a \TC algorithm, \idest, weak consistency is the precondition and strong consistency is the postcondition of every \TC algorithm.

\subsection{Connectivity of Topologies}
\label{sec:connectivity}
Next, we show that the specified graph constraints already allow us to prove the important property that the output topology of a \TC algorithm must be A-connected (see also \Cref{sec:general-properties}).
A-connectivity is a hard constraint that may be violated by \CEs because \CEs may unclassify links.
For such situations, we need a second, softer constraint for defining connectivity in the presence of \CEs.
A topology is \emph{A-U-connected} if the subgraph consisting of its active and unclassified links is connected.
According to this definition, A-connectivity implies A-U-connectivity.
The idea behind A-U-connectivity is that unclassified links should be treated as if they were active.
The \CE handlers of a \TC algorithm must ensure that the topology remains A-U-connected.
The following \Cref{thm:strong-connectivity-follows-from-strong-consistency-and-weak-connectivity,thm:strong-connectivity-after-tc-algorithm} show that the output topology of the considered \TC algorithms is A-connected if the input topology is A-U-connected and weakly consistent.
This property must be proved \wrt the \TC algorithm $A$ because weak or strong consistency is always evaluated in the context of the current \TC algorithm $A$.
\begin{theorem}\label{thm:strong-connectivity-follows-from-strong-consistency-and-weak-connectivity}
    For each considered \TC algorithm $A$, a strongly consistent and A-U-connected topology is also A-connected.
\end{theorem}
\begin{proof-sketch}
    Let $A$ be the \TC algorithm whose \activeLinkConstraintKTCLong and \inactiveLinkConstraintKTCLong are fulfilled.
    From strong consistency follows that the topology fulfills the \unclassifiedLinkConstraintLong.
    Therefore, it suffices to show the following claim:
    The end nodes of each link are connected by a path of active links in the output topology.
    This trivially holds for the end nodes of active links.
    By induction, we show that the claim also holds for all inactive links.
    
    Let $e_{i_1} \linkOrder{A} e_{i_2} \linkOrder{A} \dots \linkOrder{A} e_{i_k}$ be a strict ordering of the inactive links in the topology such that in each match \match{} of the conclusion \conclusionInactiveLinkConstraint of the \inactiveLinkConstraintKTCLong, the link corresponding to \linkVariableab is larger (\wrt \linkOrder{A}) than the links corresponding to \linkVariableac and \linkVariablecb, \idest, $\match{\linkVariableac} \linkOrder{A} \match{\linkVariableab} \wedge \match{\linkVariablecb} \linkOrder{A} \match{\linkVariableab}$.    
    Corresponding definitions of \linkOrder{A} for all considered algorithms in this paper are shown in \Cref{tab:algorithm-specific-link-order}.
    
    \emph{Induction start:} The minimal inactive link \wrt \tcPredicate{A}{}, $e_{i_1}$, is part of a triangle with two active links that connect the end nodes of link $e_{i_1}$ because the \inactiveLinkConstraintKTCLong is fulfilled.
    Thus, the claim holds for link $e_{i_1}$.
    
    \emph{Induction step:} We now consider an inactive link $e_{i_{\ell+1}}$ with $1 \leq \ell \leq k{-}1$, which is part of a triangle with two classified links, $e_1$ and $e_2$.
    We assume that only $e_1$ is inactive.%
    \footnote{If both links are active, the claim follows trivially. 
        If both links are inactive, the argument applies for each link individually.}
    Link $e_{1}$ appears before $e_{i_{\ell+1}}$ in \linkOrder{A}, \idest, there is some $s \leq \ell$ such that $e_1 := e_{i_s}$.
    Since the claim has been proved for all inactive links that appear before $e_{i_{\ell+1}}$ in \linkOrder{A}, there is a path of active links between the end nodes of $e_{i_s}$.
    A path of active links between the end nodes of link $e_{\ell+1}$ can be constructed by joining the two paths between the end nodes of $e_1$ and $e_2$.
\end{proof-sketch}
\begin{table}
    \centering
    \caption{Algorithm-specific link order \linkOrder{A} for the considered algorithms}
    \label{tab:algorithm-specific-link-order}
    \begin{tabular}{p{.4\textwidth}p{.5\textwidth}}
        \toprule
        \textbf{Algorithm $A$}&
        \textbf{Link Order \linkOrder{A}}
        \\
        \midrule
        {
            \MaxpowerTC,
            XTC \cite{Wattenhofer2004},
            GG \cite{Rodoplu1999,Wang08},
            RNG \cite{Karp2000}, 
            kTC \cite{SWBM12},
            l-\ktc \cite{Stein2016b},
            Yao graph~\cite{Yao1982}
        }
        &{
            $\begin{aligned}
            \linkVariableab \linkOrder{A} \linkVariablecd &\Leftrightarrow \weight{\linkVariableab} < \weight{\linkVariablecd} \\
            &\;\vee\;
            \big(
            \weight{\linkVariableab} = \weight{\linkVariablecd}  \wedge 
            \id{\linkVariableab} < \id{\linkVariablecd}
            \big)
            \end{aligned}$}\\
        \midrule
        \ektc (\Cref{sec:ektc})&{
            $\begin{aligned}
            \linkVariableab \linkOrder{A} \linkVariablecd &\Leftrightarrow \expectedRemainingLifetime{1}{\linkVariableab} < \expectedRemainingLifetime{1}{\linkVariablecd} \\
            &\;\vee\;
            \left(
            \expectedRemainingLifetime{1}{\linkVariableab} = \expectedRemainingLifetime{1}{\linkVariablecd}  \wedge 
            \id{\linkVariableab} < \id{\linkVariablecd}
            \right)
            \end{aligned}$}\\
        \bottomrule
    \end{tabular}
\end{table}
\begin{theorem}\label{thm:strong-connectivity-after-tc-algorithm}
The output topology of each considered \TC algorithm $A$ is A-connected if the input topology is A-U-connected and weakly consistent.
\end{theorem}
\begin{proof-sketch}
By definition, a \TC algorithm turns a weakly consistent input topology into a strongly consistent output topology.
According to \Cref{thm:strong-connectivity-follows-from-strong-consistency-and-weak-connectivity}, a strongly consistent and A-U-connected topology is A-connected.
Therefore, the claim follows.
\end{proof-sketch}

\section{Specifying Topology Control with Programmed Graph Transformation}
\label{sec:gratra}

In this section, we specify \TC operations and \CEs using graph transformation (\GT) rules and \TC algorithms using programmed graph transformation~\cite{FNTZ98}, which carries out basic topology modifications by applying graph transformation rules whose execution order is defined by an explicit control flow.

\paragraph{Programmed Graph Transformation Concepts}

A \emph{graph transformation (\GT) rule}~\cite{EEPT06,Rozenberg1997} consists of a \emph{left-hand side
(\LHS{})} pattern, a \emph{right-hand side (\RHS{})} pattern, and \emph{application conditions (ACs)}.
A \emph{positive (negative) application condition (\PAC{} (\NAC{}))} is a positive (negative) graph constraint.
To enable meaningful attribute assignments, a predicate in the attribute constraints of the \RHS{} pattern can only be an equation with an attribute of a single node or link variable on its left side.
A \emph{\GT rule is applicable to a topology \topology} if a match of the \LHS{} pattern exists in the topology that fulfills all application conditions.
The \emph{application of a \GT rule at a match of its \LHS{} pattern to a topology \topology produces a topology $\topology'$} as follows.
\begin{enumerate}
\item All nodes (links) of the topology that map to a node (link) variable of the \LHS{} pattern and have a corresponding node (link) variable in the \RHS{} pattern are \emph{preserved}.
\item All nodes (links) of the topology that are assigned to a node (link) variable of the \LHS{} pattern and lack a corresponding node (link) variable in the \RHS{} pattern are \emph{removed},
\item For each node (link) variable in the \RHS{} that is missing in the \LHS{}, a corresponding node (link) is \emph{added} to the topology.
\item Node (link) attributes are \emph{assigned} in such a manner that the attribute constraints of the \RHS{} pattern are fulfilled (operator $=$).
\end{enumerate}

\paragraph{Specifying control flow using programmed \GT}
Control flow is specified in our approach with Story-Driven Modeling~\cite{FNTZ98}, an activity-diagram-based notation in which each activity node contains a graph transformation rule.
A \emph{regular activity node} (denoted by a single framed, rounded rectangle) with one unlabeled outgoing edge applies the contained \GT rule once to one arbitrary match.
A regular activity node with an outgoing \guardSuccess and \guardFailure edge applies the contained \GT rule and follows the \guardSuccess edge if the rule is
applicable at an arbitrary match, and it follows the \guardFailure
edge if the rule is inapplicable.
A \emph{foreach activity node} (denoted by a double-framed rounded rectangle) applies the contained \GT rule to all matches and traverses along the optional outgoing edge labeled with \guardSuccess for each match. When all the matches have been completely processed, the control flow continues along the \guardFailure outgoing edge.
\RHS{} node and link variables are \emph{bound} by a successful rule application and can be reused in subsequent activity nodes.

\paragraph{Topology control rules}
Three of the five \emph{\TC rules}, shown in \Cref{fig:topology-control-rules}, represent \LSMs inside the \GT specification of a \TC algorithm:
the \activationRuleLong, 
the \inactivationRuleLong, 
and the \unclassificationRuleLong.
The remaining two \TC rules are the \findUnclassifiedLinkRuleLong and the \findClassifiedLinkRuleLong, which serve to identify (but not to modify) unclassified and classified links and whose \LHS{} and \RHS{} are identical.
\begin{figure}
    \begin{center}
        \includegraphics[width=\textwidth]{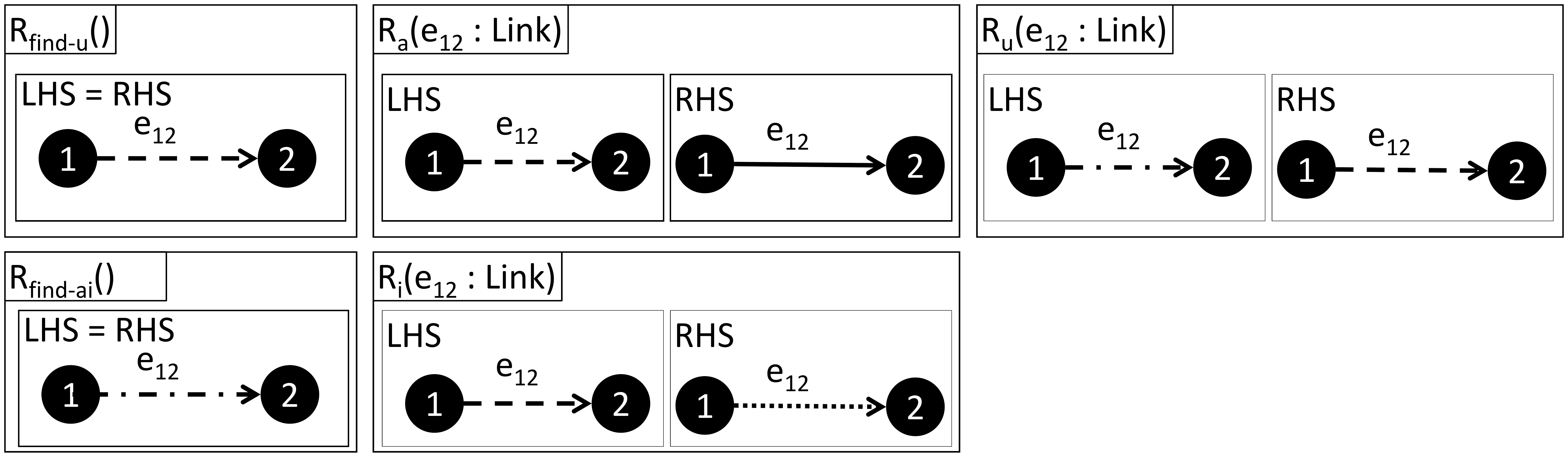}
        \caption{\TC rules}
        \label{fig:topology-control-rules}
    \end{center}
\end{figure}

\paragraph{Context event rules}
Four of the nine \emph{\CE rules}, as shown in \Cref{fig:context-event-rules}, specify structural modifications of a topology: the \nodeAdditionRuleLong, 
the \nodeRemovalRuleLong, 
the \linkAdditionRuleLong, 
and the \linkRemovalRuleLong.
The remaining five \CE rules specify attribute value modifications: 
the \weightModificationRuleLong, 
the \hopCountModificationRuleLong,
the \energyModificationRuleLong, 
the \latModificationRuleLong, and
the \longModificationRuleLong.
\begin{figure}
    \begin{center}
        \includegraphics[width=.99\textwidth]{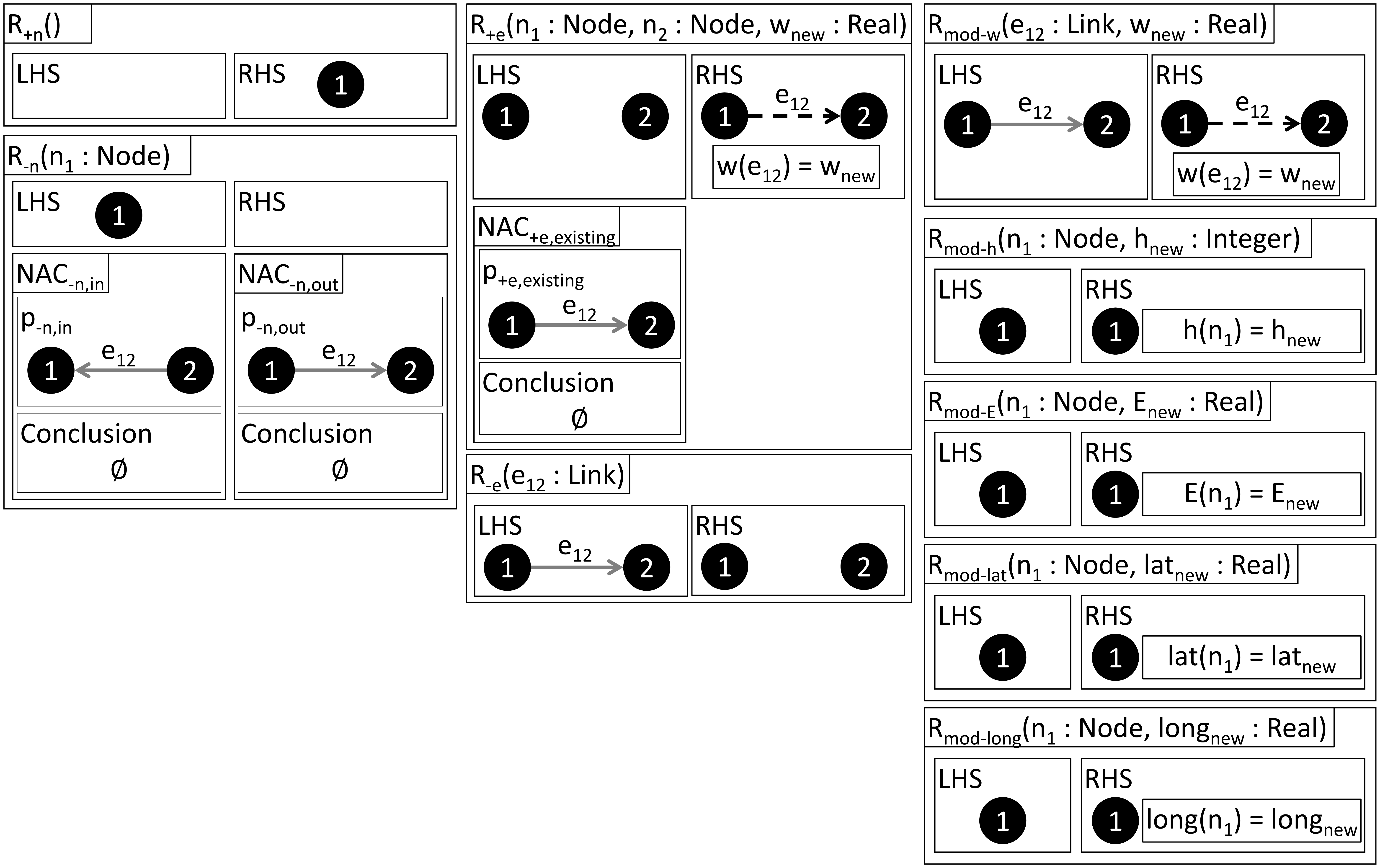}
        \caption{Context event rules}
        \label{fig:context-event-rules}
    \end{center}
\end{figure}

\paragraph{Specification of the \MaxpowerTC algorithm}
\Cref{fig:MaxpowerTCAlgorithm} shows a specification of the \MaxpowerTC algorithm (see \Cref{sec:maxpower}), which serves as starting point for the subsequent development steps.
Link variable \linkVariableOneTwo is bound by the \findUnclassifiedLinkRuleLong and reused in the activity nodes containing the \activationRuleLong and the \inactivationRuleLong.
Note that \inactivationRule is only shown here for completeness, even though it is unreachable because \activationRule is always applicable.
\begin{figure}
	\begin{center}
		\includegraphics[width=.95\textwidth]{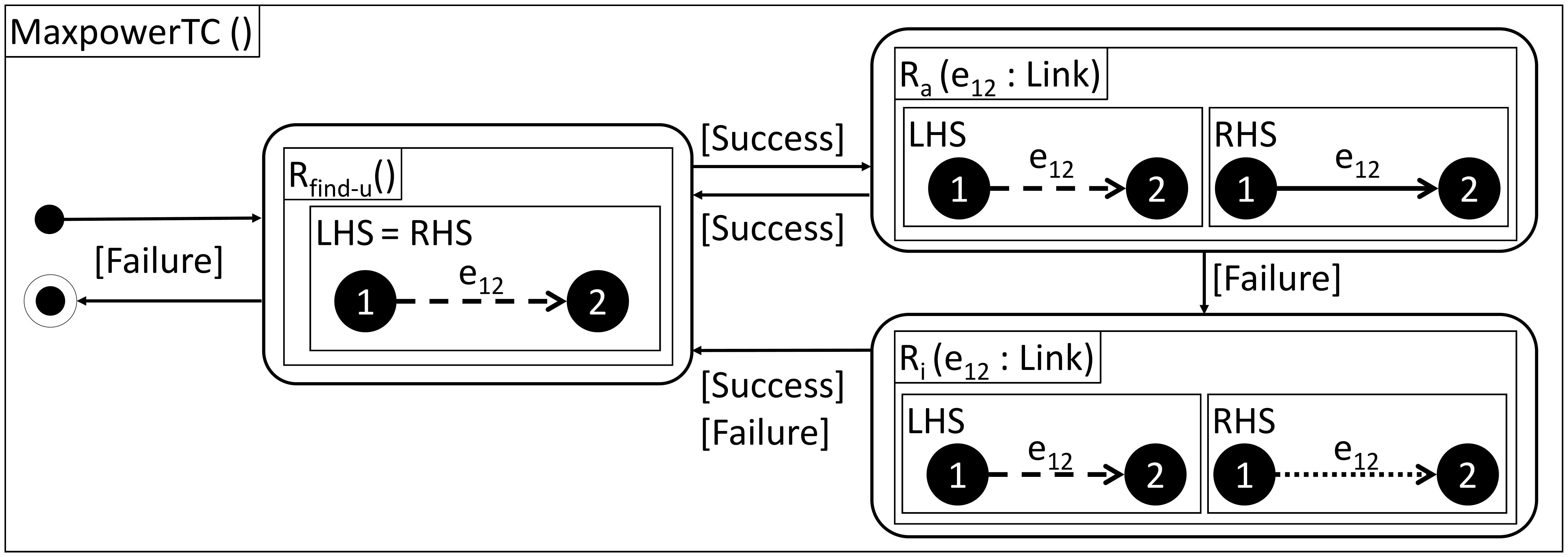}
		\caption{Specification of the \MaxpowerTC algorithm}
		\label{fig:MaxpowerTCAlgorithm}
	\end{center}
\end{figure}

\section{Refining the Graph Transformation Rules to Preserve the Graph Constraints}
\label{sec:refinement}

The fourth step of our design methodology combines the \TC and \CE rules with the graph constraints to produce a refined version of the \TC algorithm, which preserves all graph constraints.
This step is split into three substeps, as shown in \Cref{fig:overview-refinement}.
\begin{figure}
\begin{center}
\resizebox{1.0\textwidth}{!}{%
    \begin{tikzpicture}[node distance = 4.95cm]
    
    \sffamily
    \newcommand{\boundellipse}[3]
    {(#1) ellipse (#2 and #3)}
    
    \node [text width=1.6cm, align = center](c1) at (0,0) {Algorithm-\\specific\\constraints\\ \lbrack\Cref{fig:generic-graph-constraints}\rbrack};
    \draw[bigArrow] ([shift=({0.7cm,0mm})]c1.east) -- ++(4.15,0) --  ++(0,-1.25);	
    \node[bigArrowHead,text height=8.5ex, rotate=270](arrow1) at(5.8,-1.15) {};	
    
    \draw \boundellipse{0,0}{1.25}{1.075};
    
    \node [text width=1.6cm, align = center](c2) at(0,-3.0) {TC rules\\ \lbrack\Cref{fig:topology-control-rules}\rbrack};
    \draw[bigArrow] ([shift=({0.7cm,0mm})]c2.east) -- ++(14.5,0);
    \node[bigArrowHead,text height=8.5ex](arrow2) at(16.1,-3.01) {};
    \draw \boundellipse{0,-3.0}{1.25}{1.075};
    
    \node [text width=1.6cm, align = center](c3) at(0,-5.6) {Context\\event rules\\ \lbrack\Cref{fig:context-event-rules}\rbrack};
    \draw[bigArrow] ([shift=({0.7cm,0mm})]c3.east) -- ++(14.5,0);
    \node[bigArrowHead,text height=8.5ex](arrow2) at(16.1,-5.58) {};
    \draw \boundellipse{0,-5.6}{1.25}{1.075};
    
    \node [text width=1.6cm, align = center](c4) at(0,-8.2) {TC\\algorithm\\ \lbrack\Cref{fig:MaxpowerTCAlgorithm}\rbrack};
    \draw[bigArrow] ([shift=({0.7cm,0mm})]c4.east) -- ++(14.5,0);
    \node[bigArrowHead, text height=8.5ex](arrow2) at(16.1,-8.2) {};
    \draw \boundellipse{0,-8.2}{1.25}{1.075};
    
    \node[rc, minimum height=4.65cm, text width=4.55cm,node distance =2.75cm](r1) at (5.75, -4.28){Refine GT rules by adding\\application conditions to\\ensure constraint preservation\\of refined GT rules\\ \lbrack\Cref{sec:rule-refinement}\rbrack};
    
    \node[rc, minimum height=2.0cm, text width=4.55cm,node distance =2.75cm](r2) at (12.75, -5.60){Derive context event\\handlers to ensure permanent\\weak consistency\\ \lbrack\Cref{sec:deriving-context-event-handlers}\rbrack};
    
    \node[rc, minimum height=2.0cm, text width=4.55cm,node distance =2.75cm](r3) at (12.75, -8.20){Enforce termination by\\introducing \NAC{} handlers\\ \lbrack\Cref{sec:enforcing-termination}\rbrack};
    
    \node[rc, minimum height=7.28cm, text width=1.55cm,node distance =2.75cm](r4) at (17.75, -5.58){Complete\\TC\\algorithm\\ \lbrack\Cref{fig:algorithm-final}\rbrack};

\end{tikzpicture}}%
\end{center}
\caption{Overview of refinement step (\Cref{sec:refinement})}
\label{fig:overview-refinement}
\end{figure}
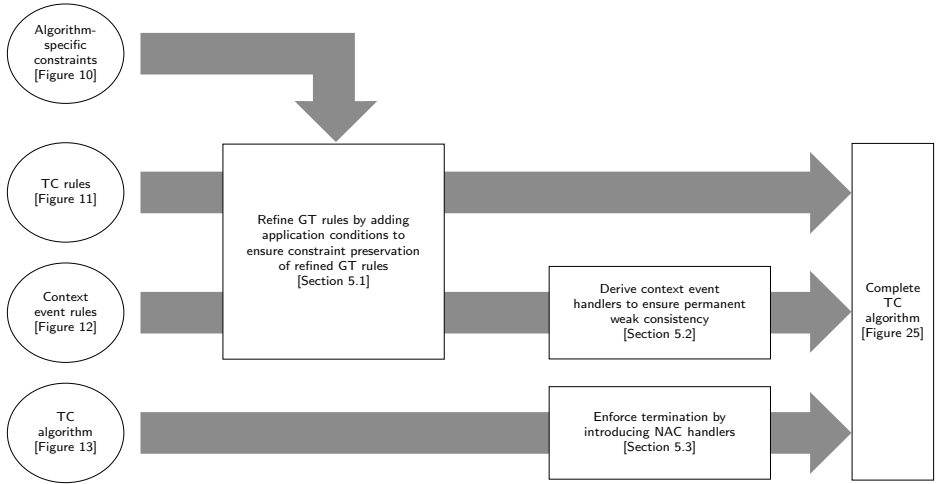
\begin{itemize}
\item
In \Cref{sec:rule-refinement}, we enrich the \TC and \CE rules step-by-step with additional application conditions that are derived from the graph constraints to ensure that the refined rules preserve all graph constraints.
After this step, the \TC algorithm preserves weak consistency.
\item 
In \Cref{sec:deriving-context-event-handlers}, we systematically transform the additional application conditions of the \CE rules into equivalent \GT-based handler operations.
These handler operations ensure that the topology is constantly weakly consistent.
This step is necessary because the applicability of \CE rules may not be restricted because they present uncontrollable effects of the environment.
The derivation of \CE handlers is one of the main contributions of this paper.
\item 
In \Cref{sec:enforcing-termination}, we refine the structure of the \TC algorithm to enforce its termination.
Shared application conditions of the refined \activationRuleLong and \inactivationRuleLong may lead to a non-termination of the \TC algorithm.
We systematically transform the shared application conditions into handler operations that ensure the termination of the \TC algorithm.
Finally, we prove that the resulting algorithm terminates.
\end{itemize}

\subsection{Refinement of \TC and Context Event Rules}
\label{sec:rule-refinement}
During this first refinement step, we
\begin{inparaenum}
\item 
analyze which of the specified \TC and \CE rules preserve or violate the \inactiveLinkConstraintKTCLong and the \activeLinkConstraintKTCLong, which constitute weak consistency, and 
\item 
apply a well-known static analysis technique~\cite{HW95} that produces additional application conditions that prevent the \GT rules from violating weak consistency.
\end{inparaenum}

\subsubsection{Analysis of Consistency Preservation}
The examples shown in \Cref{fig:constraint-violation-examples} illustrate situations in which applying one of the \GT rules violates weak consistency.
Every example is shown in the form of a weakly consistent initial topology, a rule application, and a constraint-violating final topology.
\begin{figure}
\newcommand{\subfigWidth}[0]{.7\textwidth}
\newcommand{\subfigPathPrefix}[0]{./figures/pdf/ConstraintViolationExamples}

\subcaptionbox{Applying \activationRule violates \activeLinkConstraintKTC.\label{fig:constraint-violation-example-actrule-actconstraint1}}[\textwidth]{\includegraphics[width=\subfigWidth]{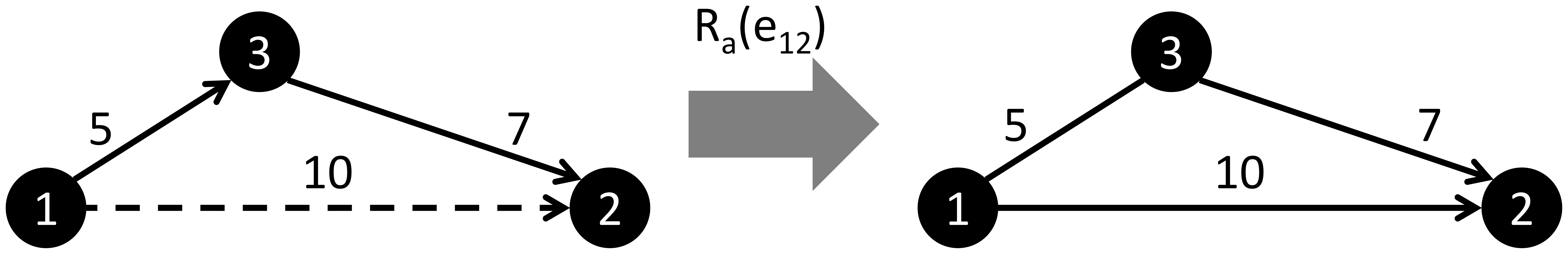}}

\subcaptionbox{Applying \activationRule violates \activeLinkConstraintKTC.\label{fig:constraint-violation-example-actrule-actconstraint2}}[\textwidth]{\includegraphics[width=\subfigWidth]{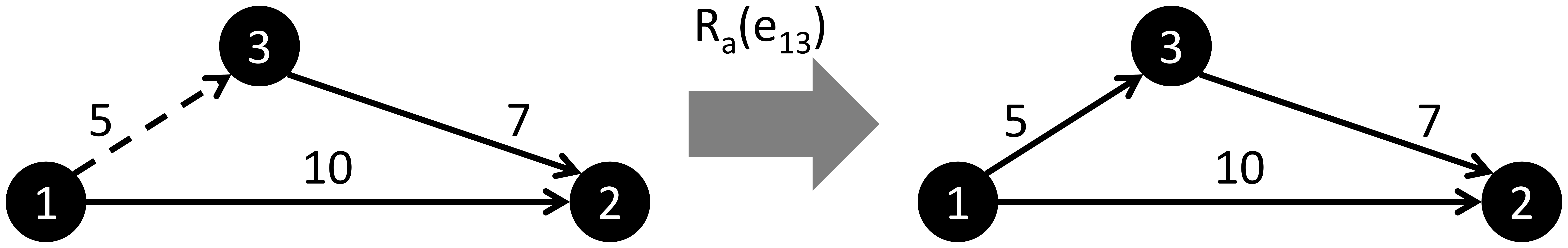}}

\subcaptionbox{Applying \inactivationRule violates \inactiveLinkConstraintKTC and \activeLinkConstraintKTC.\label{fig:constraint-violation-example-inactrule}}[\textwidth]{\includegraphics[width=\subfigWidth]{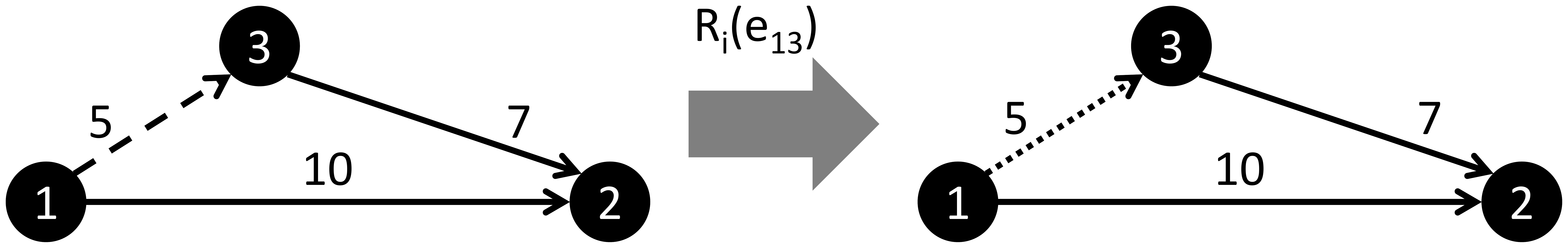}}

\subcaptionbox{Applying \unclassificationRule violates \inactiveLinkConstraintKTC.\label{fig:constraint-violation-example-unclrule}}[\textwidth]{\includegraphics[width=\subfigWidth]{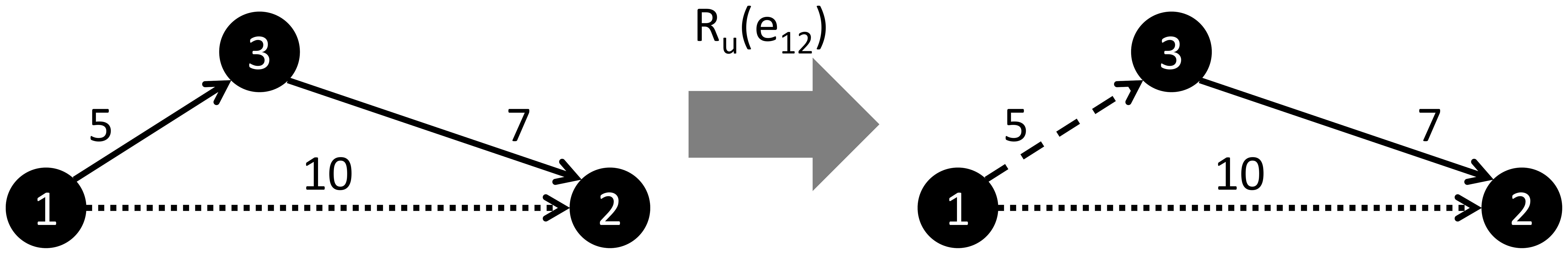}}

\subcaptionbox{Applying \linkRemovalRule violates \inactiveLinkConstraintKTC.\label{fig:constraint-violation-example-linkremrule}}[\textwidth]{\includegraphics[width=\subfigWidth]{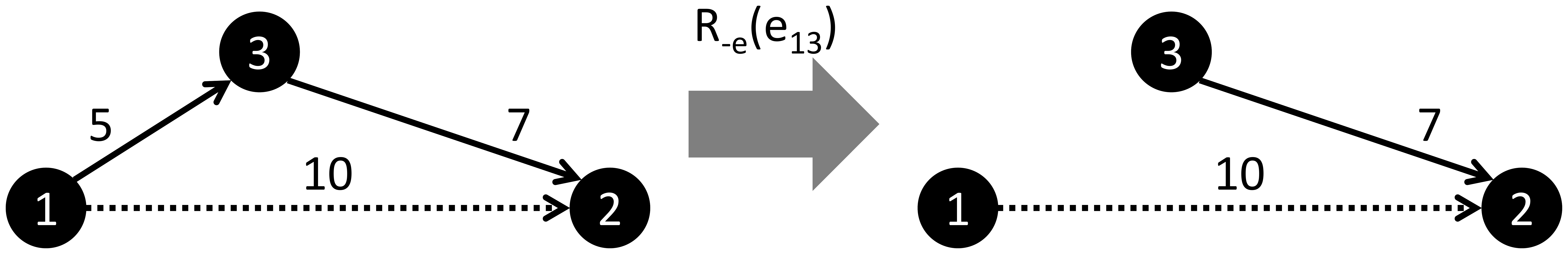}}

\subcaptionbox{Applying \weightModificationRule violates \inactiveLinkConstraintKTC.\label{fig:constraint-violation-example-modweightrule}}[\textwidth]{\includegraphics[width=\subfigWidth]{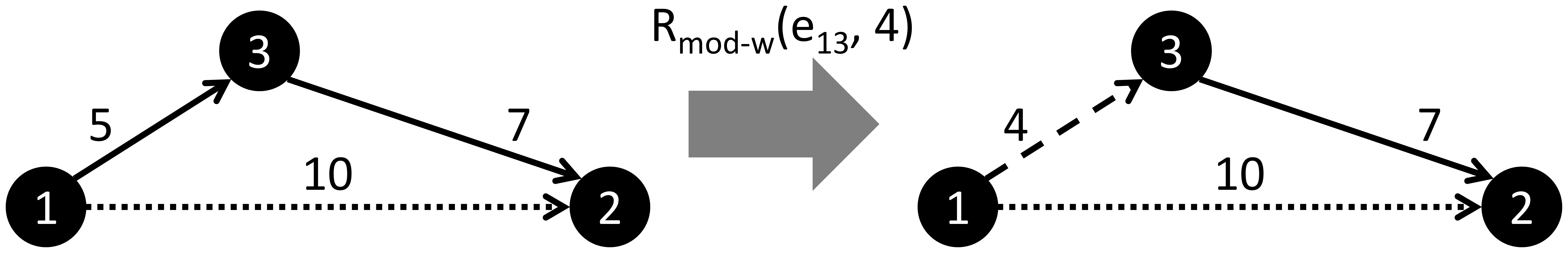}}

\caption{Sample constraint violations caused by applying \GT rules for \ktc with $k=2$}
\label{fig:constraint-violation-examples}
\end{figure}
These findings can be generalized to the results in \Cref{tab:constraint-preservation-analysis}.
\begin{table}
    \begin{center}
        \caption{Overview of constraint preservation (\OK: constraint-preserving, \notOK: constraint-violating, *: depends on \TC algorithm, $\dagger$: assuming unclassification of incident unclassified links)}
        \label{tab:constraint-preservation-analysis}
        \begin{tabular}{l|cc|p{4cm}}
            \toprule
            \textbf{Rule}
            & \textbf{\inactiveLinkConstraintKTC}
            & \textbf{\activeLinkConstraintKTC}
            & \textbf{Remarks}\\
            \midrule
            \activationRule&\OK&\notOK& \Cref{fig:constraint-violation-example-actrule-actconstraint1,fig:constraint-violation-example-actrule-actconstraint2}\\
            \inactivationRule&\notOK&\notOK&\Cref{fig:constraint-violation-example-inactrule}\\
            \unclassificationRule&\notOK&\OK&\Cref{fig:constraint-violation-example-unclrule}\\
            \findUnclassifiedLinkRule, \findClassifiedLinkRule&\OK&\OK&No modification\\
            \midrule
            \nodeAdditionRule, \nodeRemovalRule&\OK&\OK&Added/removed node is isolated\\
            \linkAdditionRule&\OK&\OK&Added link is unclassified\\
            \linkRemovalRule&\notOK&\OK&\Cref{fig:constraint-violation-example-linkremrule}\\
            \weightModificationRule&$\text{\notOK}^*$&\OK&\Cref{fig:constraint-violation-example-modweightrule}, nalogous to \unclassificationRule\\
            \latModificationRule, \longModificationRule&$\text{\notOK}^*$&$\text{\OK}^{*\dagger}$&Analogous to \unclassificationRule\\
            \energyModificationRule, \hopCountModificationRule&$\text{\notOK}^*$&$\text{\OK}^{*\dagger}$&Analogous to \unclassificationRule\\
            \bottomrule
        \end{tabular}
    \end{center}
\end{table}
\Cref{fig:constraint-violation-examples} provides examples for each entry in the table that corresponds to a constraint-violating \GT rule (denoted by a cross mark, \notOK).
In the following, we discuss why the remaining \GT rules are guaranteed to preserve the graph constraints, as indicated by the nine check marks (\OK) in the table.
\begin{itemize}
\item 
The \activationRuleLong preserves the \inactiveLinkConstraintKTCLong because applying \activationRule neither creates a new match of the premise of \inactiveLinkConstraintKTC nor destroys a match of the conclusion \conclusionInactiveLinkConstraint of \inactiveLinkConstraintKTC.
\item
The \unclassificationRuleLong preserves the \activeLinkConstraintKTCLong because unclassifying a link may never result in a new match of the \premiseActiveLinkConstraintLong of \activeLinkConstraintKTC.
\item 
The \findUnclassifiedLinkRuleLong and the \findClassifiedLinkRuleLong do not modify the topology and, therefore, preserve both constraints.
\item 
The \nodeAdditionRuleLong (\nodeRemovalRuleLong) only adds (removes) an isolated node to (from) the topology, which may neither produce a new match of \premiseActiveLinkConstraint or \premiseInactiveLinkConstraint nor destroy a match of \conclusionInactiveLinkConstraint.
\item 
The \linkAdditionRuleLong preserves \inactiveLinkConstraintKTC and \activeLinkConstraintKTC because the added link is unclassified.
\item 
The \linkRemovalRuleLong preserves the \activeLinkConstraintKTCLong because removing a link cannot establish a new match of its premise \premiseActiveLinkConstraint.
\end{itemize}
The situation is more difficult for the attribute modification rules.
We may assume that the sensor node is configured to suppress \CEs that modify attributes irrelevant for the current \TC algorithm.
Therefore, the following explanations only hold if the considered \TC algorithm depends on the particular attribute.
Independent of the considered \TC algorithm, the \weightModificationRuleLong preserves the \activeLinkConstraintKTCLong but does not preserve the \inactiveLinkConstraintKTCLong because this rule unclassifies the modified link and, thereby, behaves similar to the \unclassificationRuleLong.
For instance, whenever the hop count of a node \nodeVariablea is modified (by applying the \hopCountModificationRuleLong), we may need to evaluate for each incoming and outgoing link of \nodeVariablea whether its state needs to be updated.
A conservative resolution strategy is to unclassify all incident links of a node whenever one of its attribute values changes.
Under this assumption, the \hopCountModificationRuleLong, \energyModificationRuleLong,  the \latModificationRuleLong, and the \longModificationRuleLong behave equivalently to a sequence of applications of the \unclassificationRuleLong.

\subsubsection{Refinement of Graph Transformation Rules}
In this section, we show how the identified constraint-violating \GT rules can be refined to preserve the graph constraints.
In total, we have identified twelve problematic pairs of \GT rules and graph constraints, where a \GT rule may violate a particular graph constraint.
These pairs serve as input for the constructive refinement algorithm that has first been presented by Heckel and Wagner~\cite{HW95} for purely structural graph constraints and later extended by Deckwerth and Varró~\cite{DV14} to support attribute constraints.
The fundamental idea of translating global constraints into so-called \emph{weakest preconditions} of an algorithm dates back to Dijkstra~\cite{Dijkstra1976}.
\Cref{fig:overview-refinement-procedure} shows an overview of the refinement of a \GT rule \GTrule{x} and a graph constraint \constraint{x}.
\begin{figure}
    \begin{center}
        \includegraphics[width=.95\textwidth]{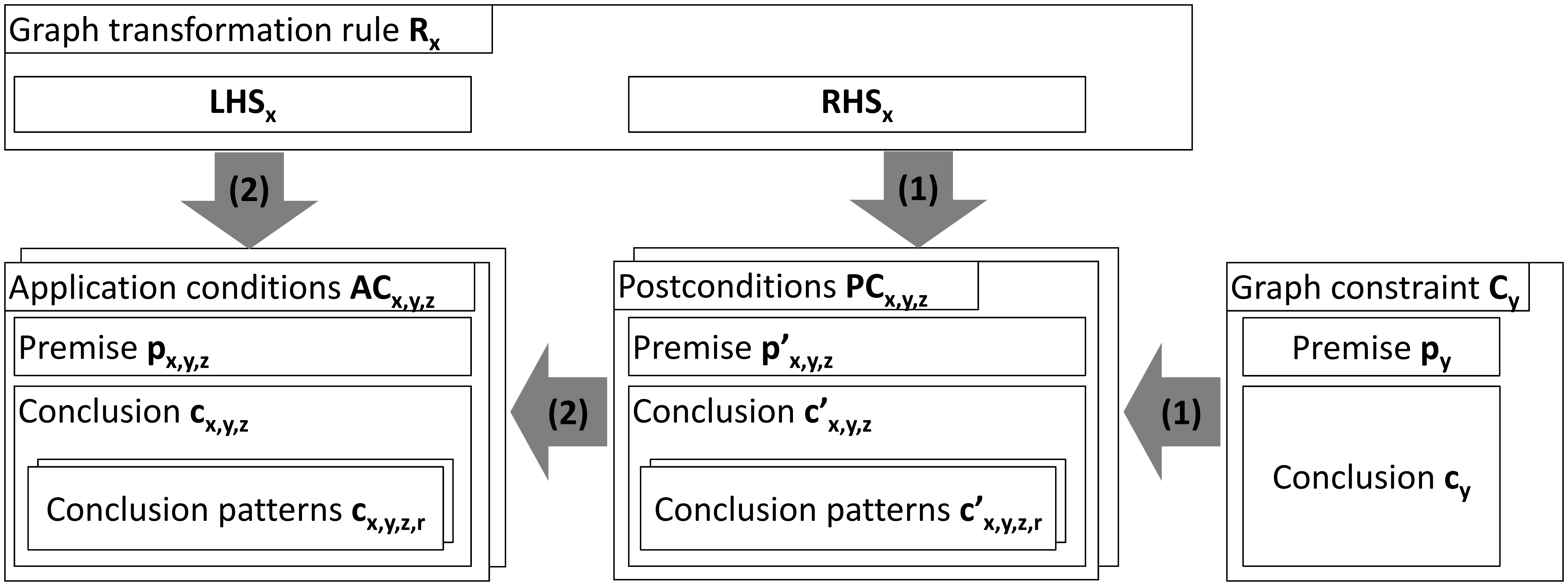}
        \caption{Refinement procedure for one \GT rule \GTrule{x} and one graph constraint \constraint{y}}
        \label{fig:overview-refinement-procedure}
    \end{center}
\end{figure}
\begin{enumerate}[(1)]
\item 
The (global) graph constraint \constraint{y} is combined with \RHS{x}, which results in a set of graph constraints that act as additional \emph{postconditions} of \GTrule{x}.
If the set of postconditions is empty, then \GTrule{x} already preserves \constraint{y}.
This is the case for all pairs of \GT rules and graph constraints that are labeled with a check mark (\OK) in \Cref{tab:constraint-preservation-analysis}.
Postconditions are similar to application conditions of a \GT rule \GTrule{x}, but they are checked \emph{after} applying \GTrule{x}.
If a postcondition is violated, the application of \GTrule{x} has to be rolled back by the \GT engine.
Avoiding this rollback is the purpose of the following transformation (2).
\item 
The postconditions are transformed into an equivalent set of application conditions by applying \GTrule{x} in reverse order to the  postconditions.
These application conditions block any application of \GTrule{x} that would violate its corresponding postcondition.
If applying \GTrule{x} to a particular postcondition in reverse order is impossible, then \GTrule{x} never violates this postcondition.
\end{enumerate}
The following explanations are deliberately simplified and shown for the following rule-constraint pairs: $(\activationRule, \activeLinkConstraintKTC), (\inactivationRule, \inactiveLinkConstraintKTC)$.
A detailed, formal description of the following steps can be found in \cite{Kluge2016,HW95}.
The presented steps are analogous for all other pairs of \GT rules and graph constraints.

First, we define a concept that is crucial for the refinement algorithm:
A \emph{gluing \gluing{\ell, r}{} of two patterns \leftPattern and \rightPattern} is a pattern that represents a possible way of overlapping \leftPattern and \rightPattern.
We label node variables in gluings with uppercase letters and the original node variable(s) in \leftPattern and \rightPattern.
For instance, the node variable $A[x,y]$ of a gluing originates from node variables $x$ of \leftPattern and $y$ of \rightPattern.
In a valid gluing,
\begin{inparaenum}
\item 
every node (link) variable has one or two original node (link) variables in \leftPattern and/or \rightPattern,
\item 
at least one node variable in \gluing{\ell, r}{} has original node variables in both \leftPattern and \rightPattern,
\item 
the link variable mappings are compatible, \idest, if \nodeVariablea and \nodeVariableb are the original node variables of \nodeVariableA and \nodeVariableB, then \linkVariableab is the original link variable of \linkVariableAB, and
\item 
the attribute constraints of \gluing{\ell, r}{} are the conjunction of the attribute constraints of \leftPattern and \rightPattern.
\end{inparaenum}

\subsubsection{Refinement of \ActivationRuleLong Based on the \ActiveLinkConstraintKTCLong}
We begin with the refinement of the \activationRuleLong based on the negative \activeLinkConstraintKTCLong.
The refinement step can be performed faster for negative graph constraints because its empty conclusion maps to empty conclusions of the postcondition and the application conditions.
We obtain the set of postconditions by calculating all twelve gluings $\gluingAA{z}{\prime}, 1 \leq z \leq 12$  of the \RHS{} of the \activationRuleLong{} and the premise \premiseActiveLinkConstraint of the \activeLinkConstraintKTCLong (\Cref{fig:refinement-actrule-actconstraint-gluings}).
\begin{figure}
    \begin{center}
        \includegraphics[width=\textwidth]{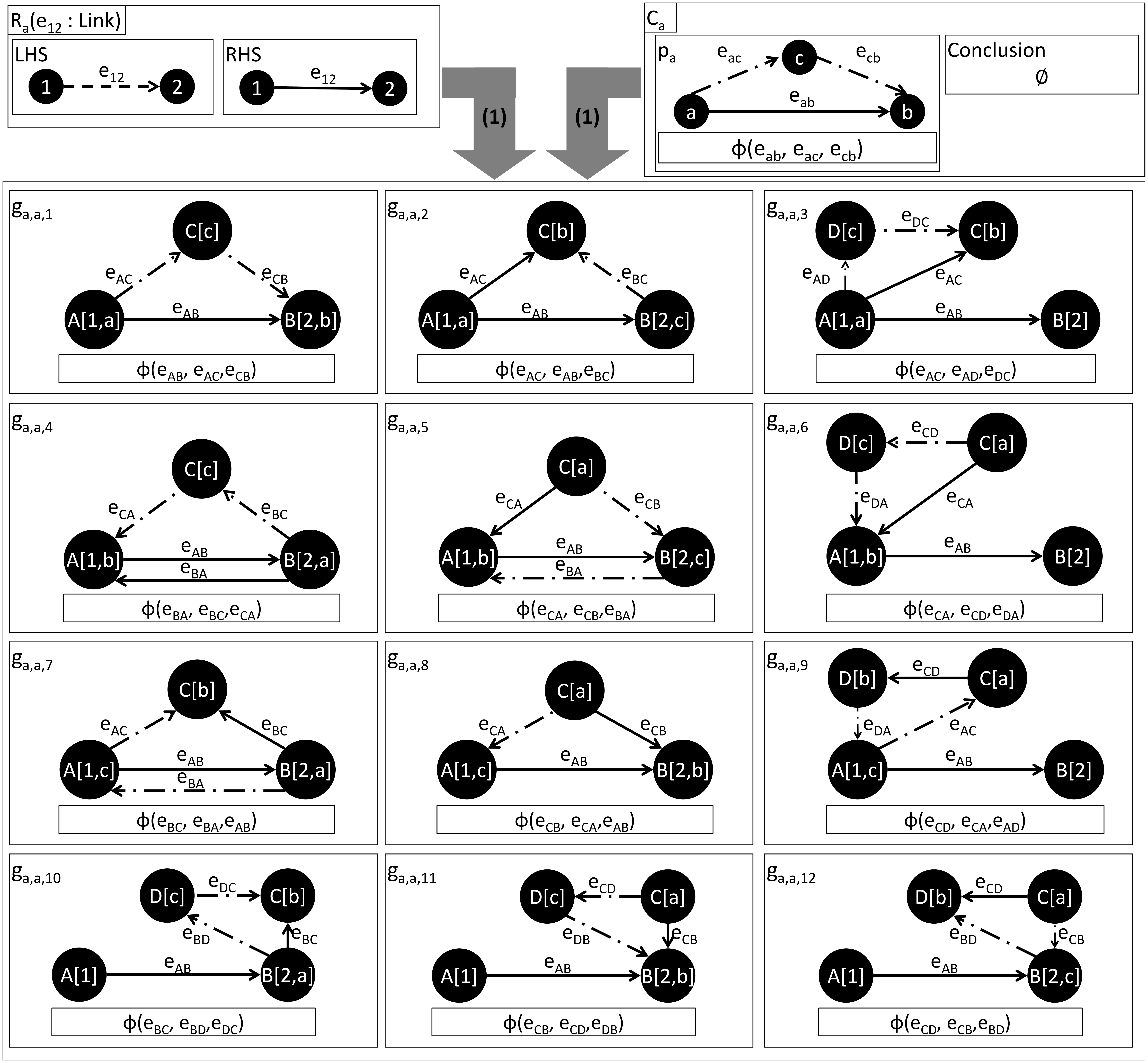}
        \caption{All twelve possible gluings of \RHSActivationRule and \premiseActiveLinkConstraint}
        \label{fig:refinement-actrule-actconstraint-gluings}
    \end{center}
\end{figure}
Each gluing represents a possible violation of \activeLinkConstraintKTC.
We observe that the nine gluings $\gluingAA{z}{\prime}$ for $z \in \{3, 4, 5, 6, 7, 9, 10, 11, 12\}$ represent constraint violations that are \emph{not} caused by activating \linkVariableAB.
This means that any constraint violation corresponding to these gluings already existed before applying the \activationRuleLong, which contradicts the assumption that weak consistency is fulfilled prior to invoking any \GT rule.
Therefore, we neglect the aforementioned gluings and only consider the remaining three gluings \gluingAA{1}{\prime}, \gluingAA{2}{\prime}, and \gluingAA{8}{\prime} for the transformation from postconditions to application conditions.
In step (2), we obtain three \NACs, by changing the link-state attribute condition of \linkVariableAB from $\state{\linkVariableAB} = \ACT$ to $\state{\linkVariableAB} = \UNCL$.
The resulting application conditions \NACaa{1}, \NACaa{2}, and \NACaa{8} are shown in \Cref{fig:refinement-actrule-actconstraint-nacs}.
\begin{figure}
    \begin{center}
        \includegraphics[width=.9\textwidth]{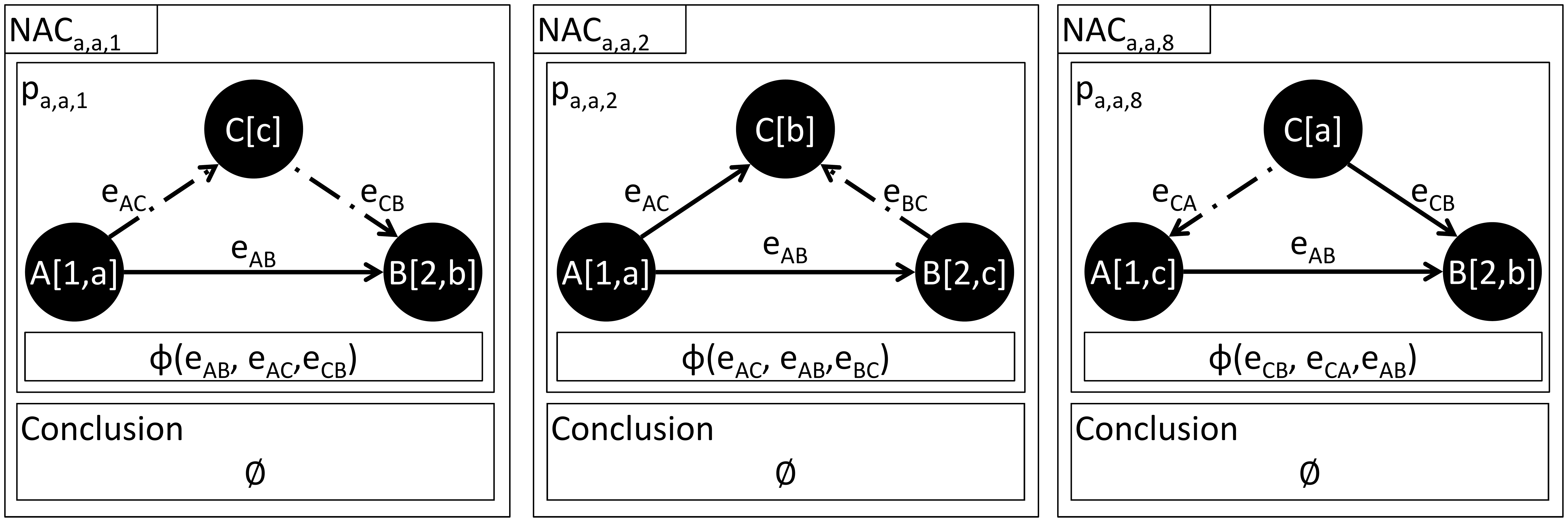}
        \caption{Application conditions resulting from the refinement of \activationRule and \activeLinkConstraintKTC}
        \label{fig:refinement-actrule-actconstraint-nacs}
    \end{center}
\end{figure}
Finally, we rename the node variables in the generated application conditions back to the original node variables of the \activationRuleLong, as shown in \Cref{fig:activation-rule-after-refinement}.

Note that \NACaa{1} prevents the constraint violation shown in \Cref{fig:constraint-violation-example-actrule-actconstraint1}, and \NACaa{2} prevents the constraint violation shown in \Cref{fig:constraint-violation-example-actrule-actconstraint2}.
Similarly, \NACaa{8} would prevent a constraint violation if \linkVariableThreeTwo (and not \linkVariableOneThree) were unclassified in \Cref{fig:constraint-violation-example-actrule-actconstraint2}.

\subsubsection{Refinement of \InactivationRuleLong Based on the \InactiveLinkConstraintKTCLong}
\label{sec:refinement-actrule-inactconstraint}
In this second example, we focus on the transformation of the constraint conclusion.
The basic idea is that the gluings, which result from combining \RHS{x} with the \premise{y}{}, serve as a basis for deriving the conclusion of the postcondition.
We obtain the conclusion pattern \conclusion{x,y,z,1}{\prime} of the postcondition \postcondition{x,y,z} by first adding images of all elements in the conclusion of \constraint{y} that are not covered by the gluing of \RHS{x}  and \premise{y}{}.
In the original presentation of the constructive approach~\cite{HW95} contains an additional step:
The conclusion pattern \conclusion{x,y,z,1}{\prime} serves as the basis to generate additional conclusion patterns $\conclusion{x,y,z,r}{\prime}, r > 1$ by merging the freshly added node variables with the existing node variables.
In our example, this step is not necessary because merging node variables results in loops or parallel links, which never occurs in the considered class of topologies.

As a concrete example, we refine the \inactivationRuleLong based on the positive \inactiveLinkConstraintKTCLong.
\Cref{fig:refinement-inactrule-inactconstraint-gluings} shows the three possible gluings $\gluingII{z}{\prime}, z \in \{1,2,3\}$ of \RHSInactivationRule the premise \premiseInactiveLinkConstraint of the \inactiveLinkConstraintKTCLong and the corresponding conclusions \conclusionII{1,1}{\prime}, \conclusionII{2,1}{\prime} and \conclusionII{3,1}{\prime}.
\begin{figure}
    \begin{center}
        \includegraphics[width=\textwidth]{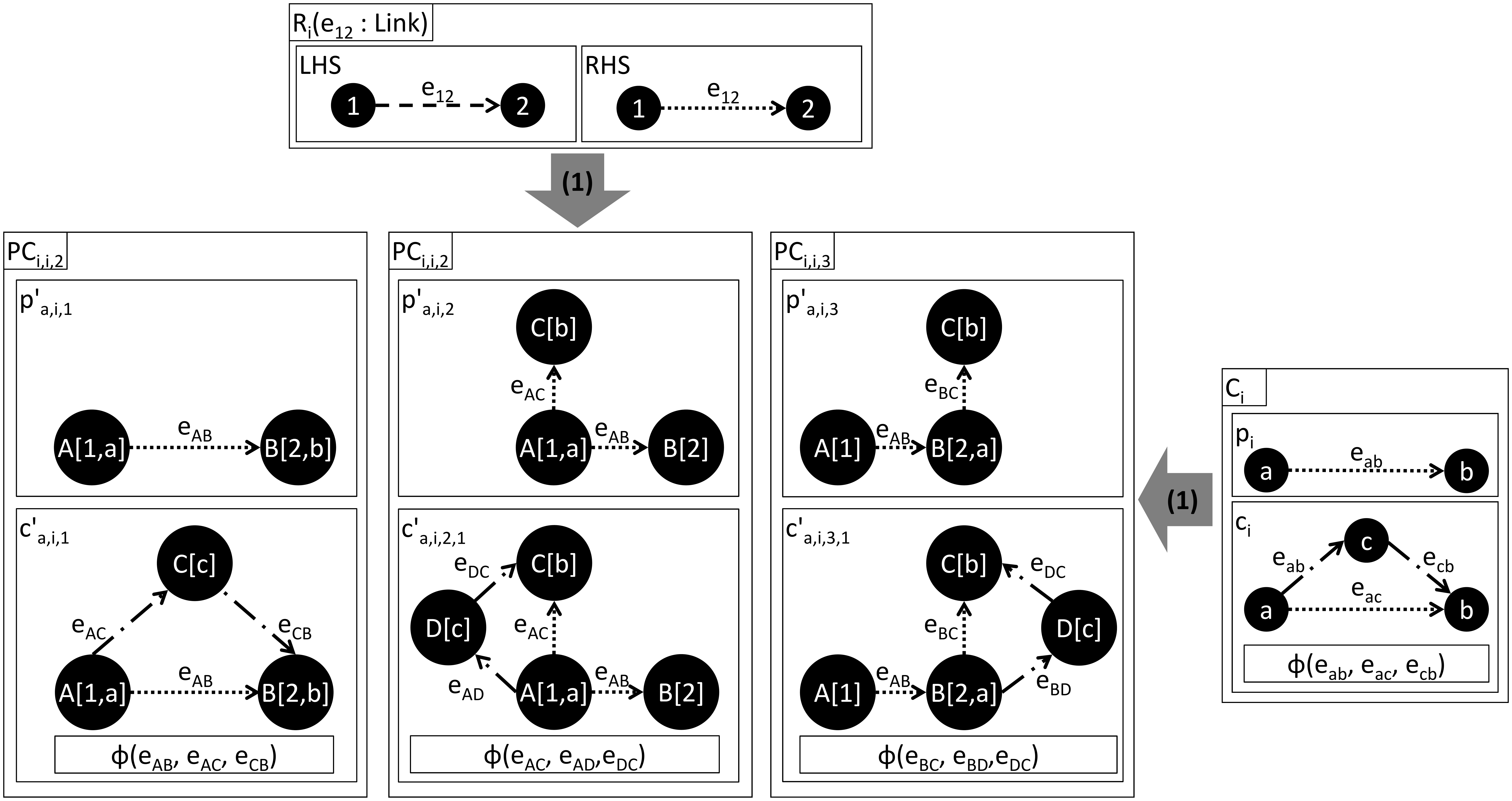}
        \caption{Derivation of the postconditions \postconditionII{1}, \postconditionII{2}, and \postconditionII{3} for \inactivationRule based on \inactiveLinkConstraintKTC}
        \label{fig:refinement-inactrule-inactconstraint-gluings}
    \end{center}
\end{figure}
For \postconditionII{2} and \postconditionII{3}, we may stop the refinement procedure here because the possible constraint violations represented by \postconditionII{2} and \postconditionII{3} are not caused by the inactivation of \linkVariableAB.
The remaining postcondition \postconditionII{1} results in a new positive application condition \PACii{1}.
\Cref{fig:inactivation-rule-after-refinement} shows the resulting application condition of \activationRule with appropriately renamed variables.

\subsubsection{Refinement of the Remaining Combinations of \GT Rules and Constraints}

We only sketch the required remaining rule refinements instead of describing them in detail.
Examples of additional rule refinement steps can be found in~\cite{Kluge2016}.
\begin{enumerate}
\item 
The refinement of $(\inactivationRule, \activeLinkConstraintKTC)$ results in two additional negative application conditions\linebreak\NACia{2} and \NACia{8} of the \inactivationRuleLong, as shown in \Cref{fig:inactivation-rule-after-refinement}, which are identical to the negative application conditions \NACaa{2} and \NACaa{8} of the \activationRuleLong.
The refined \inactivationRuleLong is no longer applicable to the left topology in \Cref{fig:constraint-violation-example-inactrule} due to \NACia{2}.

\item 
The refinement of $(\unclassificationRule, \inactiveLinkConstraintKTC)$ results in four additional positive application conditions\linebreak\PACui{2}, \PACui{4}, \PACui{5}, and \PACui{6} (\Cref{fig:unclassification-rule-after-refinement}).
These application conditions require that a link \linkVariableOneTwo may only be unclassified if it is not part of the last match of the conclusion \conclusionInactiveLinkConstraint for any incident link \linkVariableOneThree, \linkVariableThreeOne, \linkVariableTwoThree, or \linkVariableThreeTwo of its end nodes \nodeVariableOne and \nodeVariableTwo.
\begin{figure}
    \begin{center}
        \includegraphics[width=.7\textwidth]{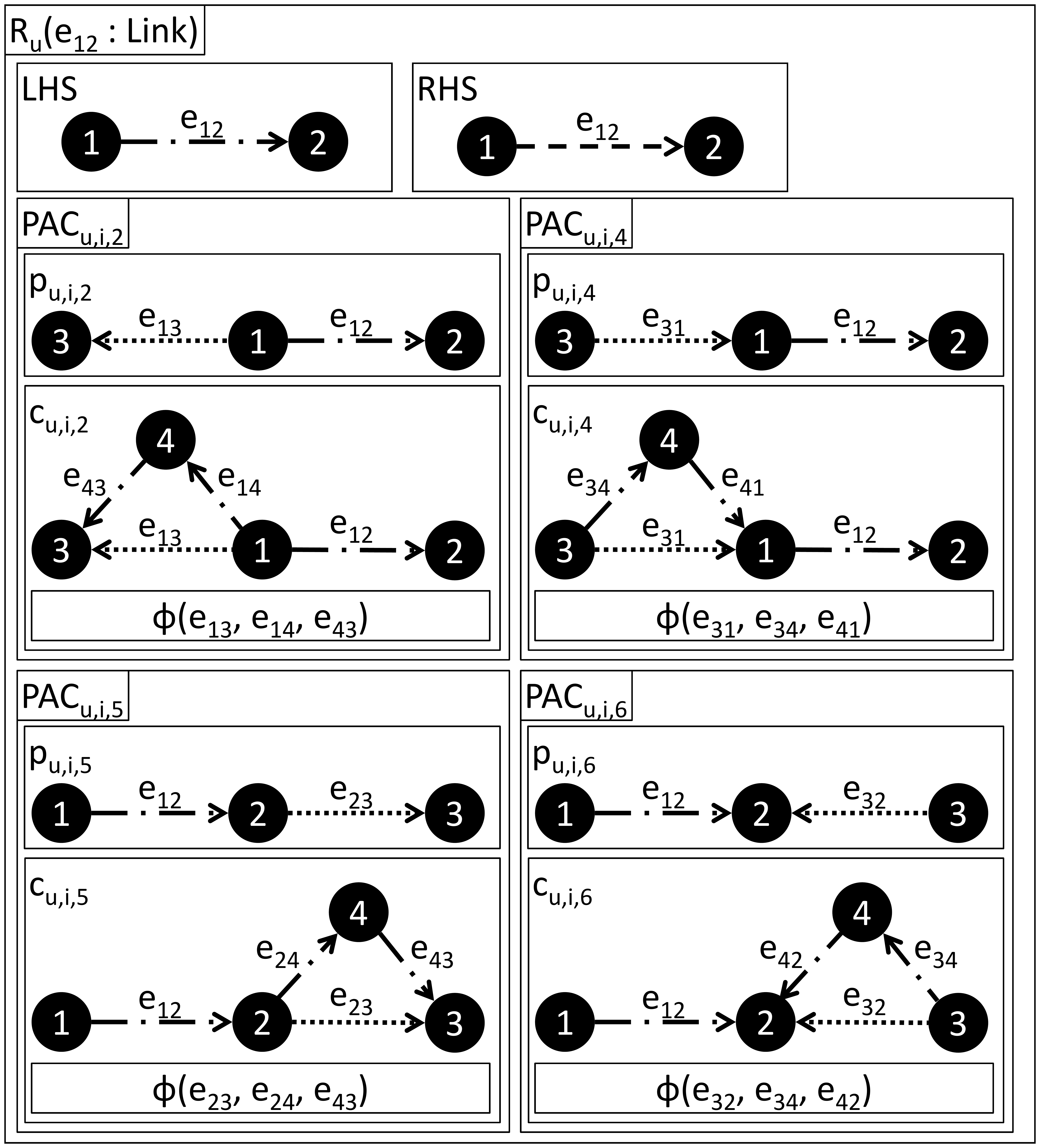}
        \caption{\xmakefirstuc{\unclassificationRuleLong} after refinement}
        \label{fig:unclassification-rule-after-refinement}
    \end{center}
\end{figure}
The refined \unclassificationRuleLong is now no longer applicable to the topology in \Cref{fig:constraint-violation-example-unclrule} due to \PACui{2}.

\item 
The refinements of $(\linkRemovalRule, \inactiveLinkConstraintKTC)$, $(\weightModificationRule, \inactiveLinkConstraintKTC)$, $(\latModificationRule, \inactiveLinkConstraintKTC)$,
$(\longModificationRule, \inactiveLinkConstraintKTC)$, $(\energyModificationRule, \inactiveLinkConstraintKTC)$, and $(\hopCountModificationRule, \inactiveLinkConstraintKTC)$ result in positive application conditions that are similar to the application conditions of the refined \unclassificationRuleLong and are not shown here for conciseness.
\end{enumerate}
The specification of the \TC algorithm after the rule refinement step is shown in \Cref{fig:algorithm-after-refinement}.
\begin{figure}
    \begin{center}
        \newcommand{\subfigWidth}{\textwidth}
        \begin{subfigure}[t]{\subfigWidth}
            \includegraphics[width=.99\textwidth]{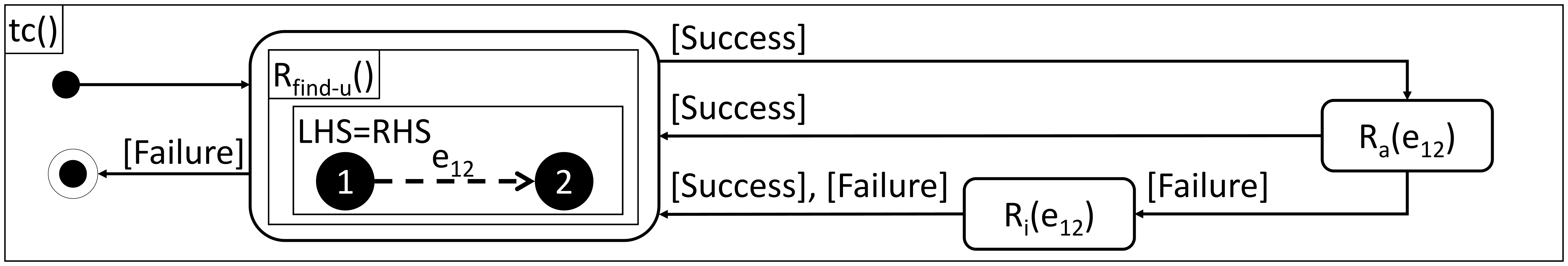}
            \caption{Control flow of \TC algorithm after refinement (Same as \Cref{fig:MaxpowerTCAlgorithm})}
            \label{fig:algorithm-after-refinement}
        \end{subfigure}
        \begin{subfigure}[t]{\subfigWidth}
             \includegraphics[width=.99\textwidth]{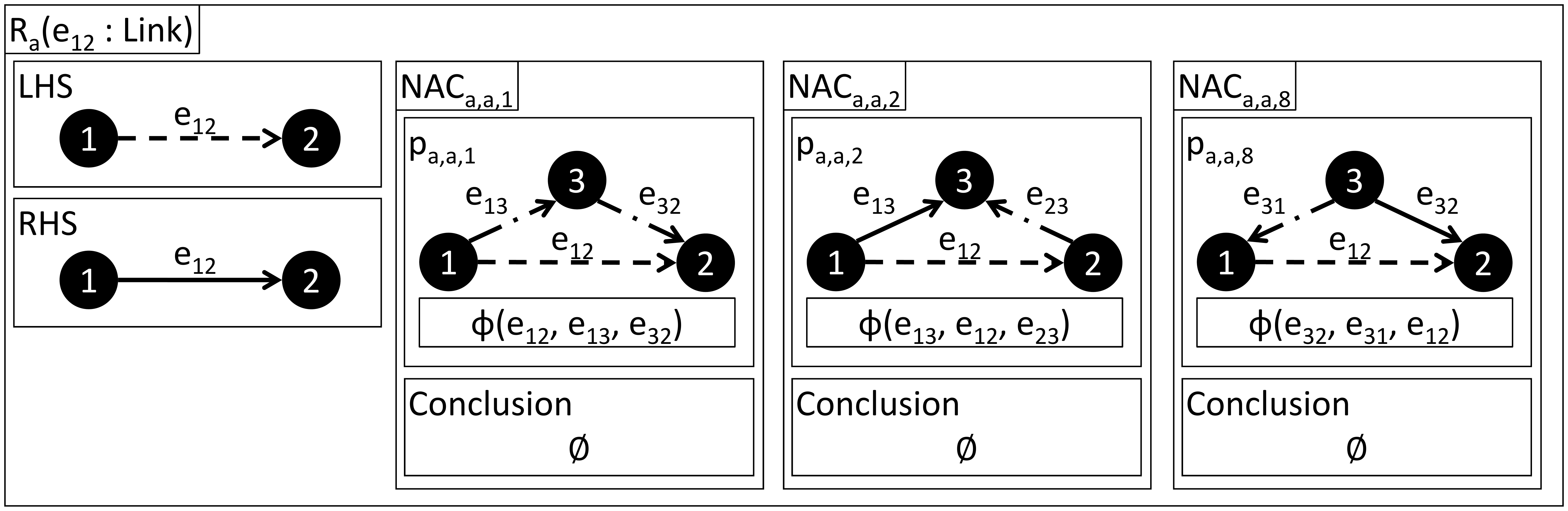}
             \caption{Refined \activationRuleLong}
             \label{fig:activation-rule-after-refinement}
        \end{subfigure}
        \begin{subfigure}[t]{\subfigWidth}
            \includegraphics[width=.99\textwidth]{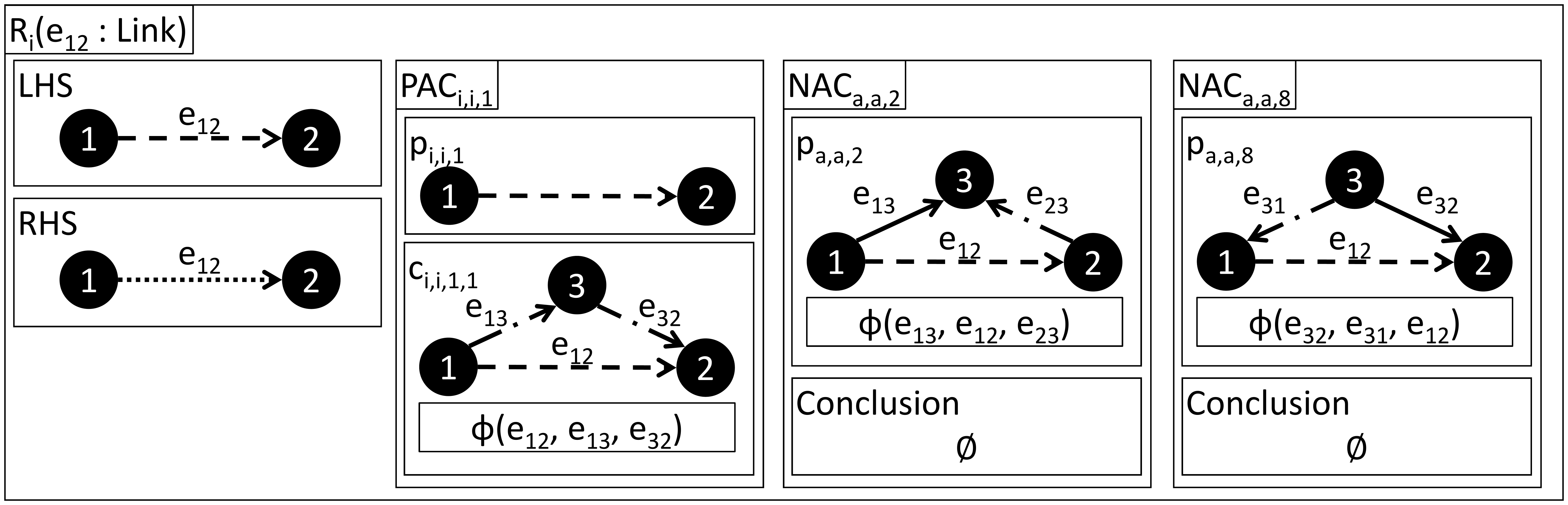}
            \caption{Refined \inactivationRuleLong}
            \label{fig:inactivation-rule-after-refinement}
        \end{subfigure}
        \caption{Overview of \TC algorithm after rule refinement}
    \end{center}
\end{figure}

\subsection{Deriving Context Event Handlers}
\label{sec:deriving-context-event-handlers}
The first refinement step resulted in additional application conditions for \TC and \CE rules.
The refined rules \unclassificationRule, \linkRemovalRule, \weightModificationRule, \latModificationRule, \longModificationRule, \energyModificationRule, and \hopCountModificationRule are now applicable in fewer situations than before.
For the \CE rules, this is problematic because \CE rules represent unrestrictable modifications of the topology caused by the environment.
For the \unclassificationRuleLong, this is problematic because its purpose is to deliberately unclassify links, which should always be possible.
Therefore, we have to restore the original applicability of these \GT rules \emph{without} sacrificing their constraint-preserving behavior.
We propose to transform each added application condition into a \emph{handler operation}, which repairs any constraint violation that results from applying the original \CE rule.
We will next describe the general idea of deriving handler operations and then illustrate the algorithm for the \unclassificationRuleLong.
The derivation of handler operations of the remaining \GT rules is analogous.

\Cref{fig:context-event-handler-sketch} shows the structure of a generic context event handler operation \handlerOperation{\GTrule{x}} for the \GTruleLong{x}.
\begin{figure}
    \begin{center}
        \includegraphics[width=\textwidth]{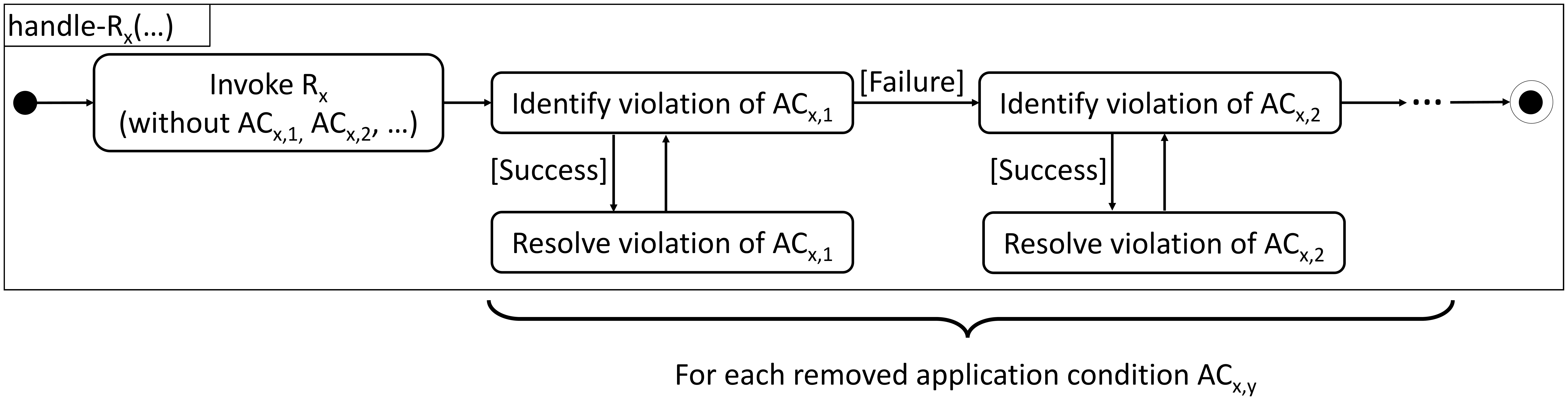}
        \caption{Structure of the handler operation \handlerOperation{\GTrule{x}} for \GT rule \GTrule{x}}
        \label{fig:context-event-handler-sketch}
    \end{center}
\end{figure}
The fundamental idea is to 
\begin{inparaenum}
\item 
first apply the original \GTruleLong{x} (\idest, without the additional application conditions \AC{x,1}, \AC{x,2}, \dots), and 
\item 
then to identify and resolve violations of \AC{x,1}, \AC{x,2}, \dots
\end{inparaenum}
The control flow of the handler operation ensures that it may only terminate if all violations have been resolved.
The most important requirement is that the violation resolution strategy shall not produce new constraint violations.
In our scenario, we propose to resolve any constraint violation by means of unclassifying links.
This approach is valid because a topology consisting exclusively of unclassified links fulfills the \inactiveLinkConstraintKTCLong and the \activeLinkConstraintKTCLong.
Therefore, a na\"{i}ve violation resolution strategy could simply unclassify all links in the topology.

We now derive the concrete handler operation \handleUnclassification for the \unclassificationRuleLong.
During the refinement step, four application conditions have been added to \unclassificationRule, which are now translated into four identify-and-resolve loops.
The invocation of the original, unrefined \unclassificationRuleLong and the violation resolution for \PACui{2} are shown in \Cref{fig:handler-operation-for-unclassification-rule}.
The violation identification rule \violationIdentificationRule{\text{PAC-u,i,2}} identifies any link \linkVariableOneThree that is \emph{not} part of a triangle together with classified links \linkVariableOneFour and \linkVariableFourThree so that the predicate \tcPredicate{}{\linkVariableOneThree, \linkVariableOneFour, \linkVariableFourThree} is fulfilled.
\begin{figure}
    \begin{center}
        \includegraphics[width=\textwidth]{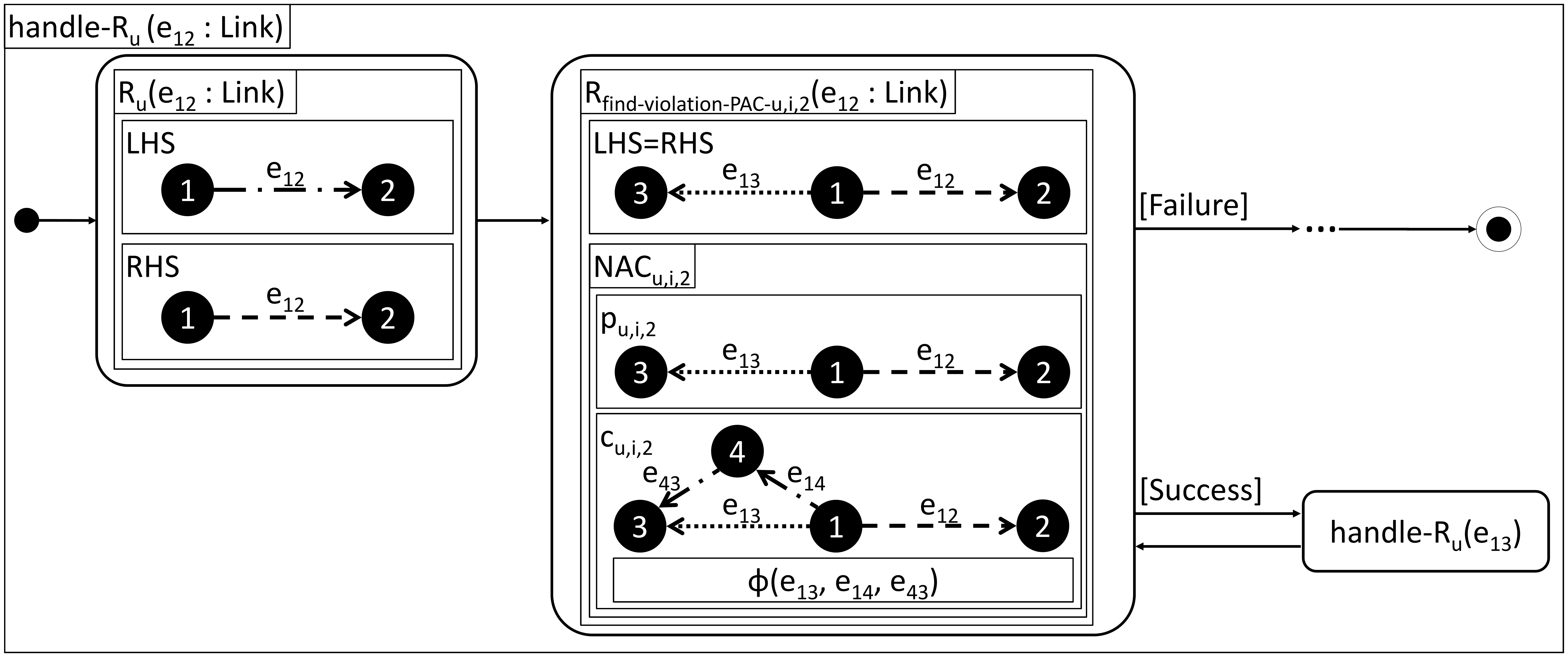}
        \caption{Handler operation \handleUnclassification for the \unclassificationRuleLong}
        \label{fig:handler-operation-for-unclassification-rule}
    \end{center}
\end{figure}
We unclassify any such link \linkVariableOneTwo by invoking \handleUnclassification recursively.
The violation-identifying and violation-resolving story nodes that correspond to \PACui{4}, \PACui{5}, and \PACui{6} are analogous and omitted here for conciseness.

\subsection{Enforcing Termination}
\label{sec:enforcing-termination}

The refined \activationRuleLong and \inactivationRuleLong share two pairs of identical \NACs.
This means that whenever \NACaa{2} prevents the application of the \activationRuleLong, \NACia{2} also prevents the application of the \inactivationRuleLong.
The same holds for \NACaa{8} and \NACia{8}.
In \Cref{fig:example-need-for-unclassification}, the top row of topology modifications shows a situation where a particular processing order of the unclassified links causes non-termination of the algorithm.
In this case, \ktc is executed with $k=2$, and \linkVariableThreeTwo and \linkVariableOneTwo are activated prior to processing \linkVariableOneThree.
Now, \linkVariableOneThree can be neither activated nor inactivated due to the negative application conditions \NACaa{2} and \NACia{2}, respectively.
This causes an infinite execution of the algorithm.

The situation can be solved by reverting link classifications.
In this example, a possible solution can be derived from the only possible strongly consistent output topology:
The link \linkVariableOneThree and \linkVariableThreeTwo are active, and \linkVariableOneTwo is inactive.
This means that \linkVariableOneTwo should be unclassified again and \linkVariableOneThree should be activated.
The two rule applications have to happen atomically, \idest,  the activation of \linkVariableOneThree must follow immediately after the unclassification of \linkVariableOneTwo.
This order can easily be implemented using programmed \GT.
Finally, \linkVariableOneTwo becomes inactive.
\begin{figure}
    \begin{center}
        \includegraphics[width=\textwidth]{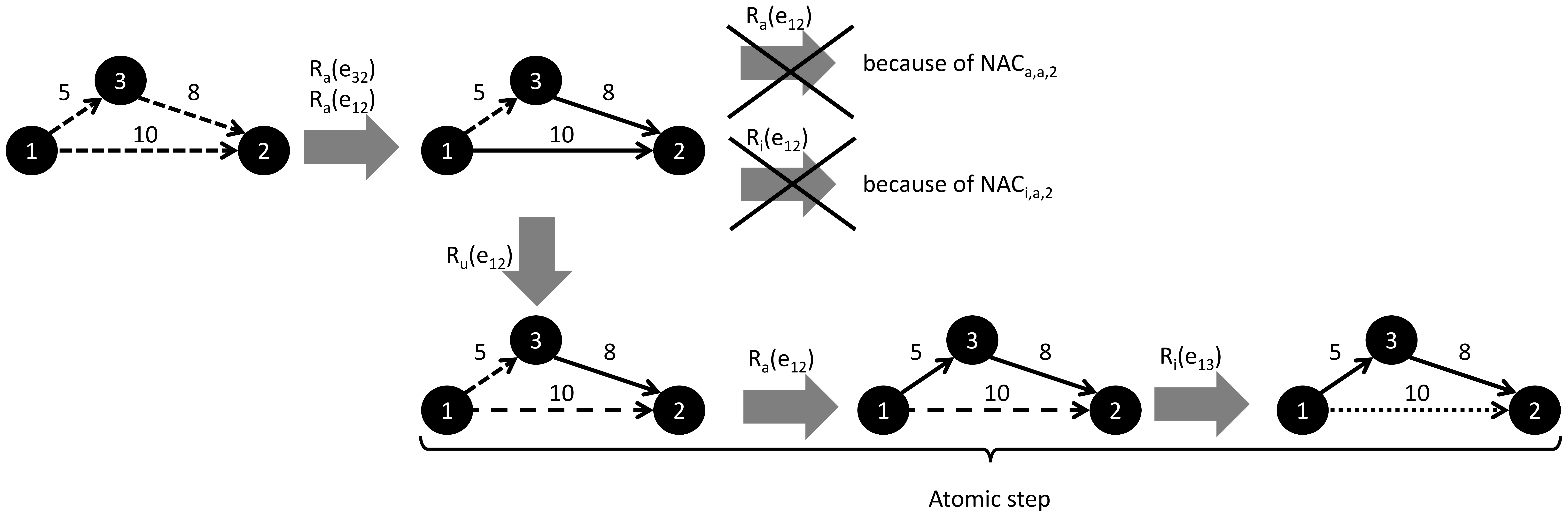}
        \caption{Example of non-terminating execution of \ktc with $k=2$ and possible solution}
        \label{fig:example-need-for-unclassification}
    \end{center}
\end{figure}

In fact, the solution for the example can be generalized, again using handler operations.
We systematically transform the shared negative application conditions of the \activationRuleLong and the \inactivationRuleLong into an appropriate \NAC{}-handling operation \handlerOperation{$\text{NAC}_{\text{aa2,ia2,aa8,ia8}}$}, which destroys all matches of the premises of the aforementioned four \NACs (\Cref{fig:algorithm-final}).
Contrary to the \CE handler operations, we decided to place the \NAC{}-handling operation in front of \activationRule and \inactivationRule because, otherwise, we would have to add invocations of the \NAC{}-handling operation in both \guardSuccess-branches.
\begin{figure}
    \begin{center}
        \newcommand{\subfigWidth}{\textwidth}
        \begin{subfigure}[t]{\subfigWidth}
            \includegraphics[width=.99\textwidth]{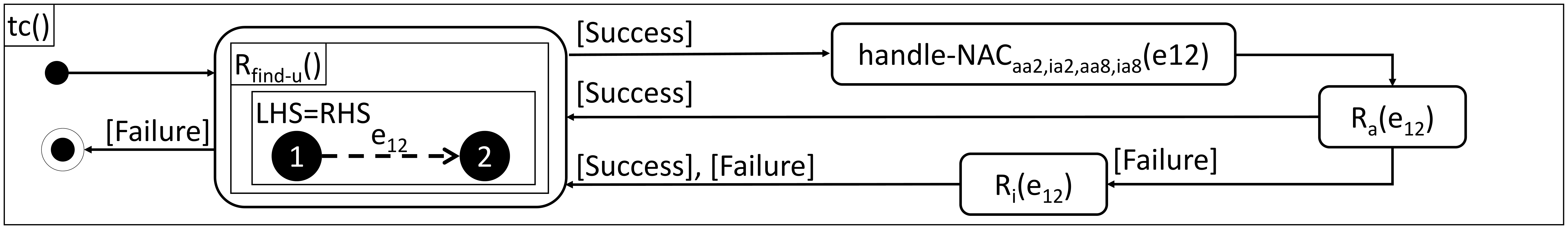}
            \caption{Control flow of \TC algorithm after refinement}
            \label{fig:TCAlgorithmFinal}
        \end{subfigure}
        
        \begin{subfigure}[t]{.45\subfigWidth}
            \includegraphics[width=.99\textwidth]{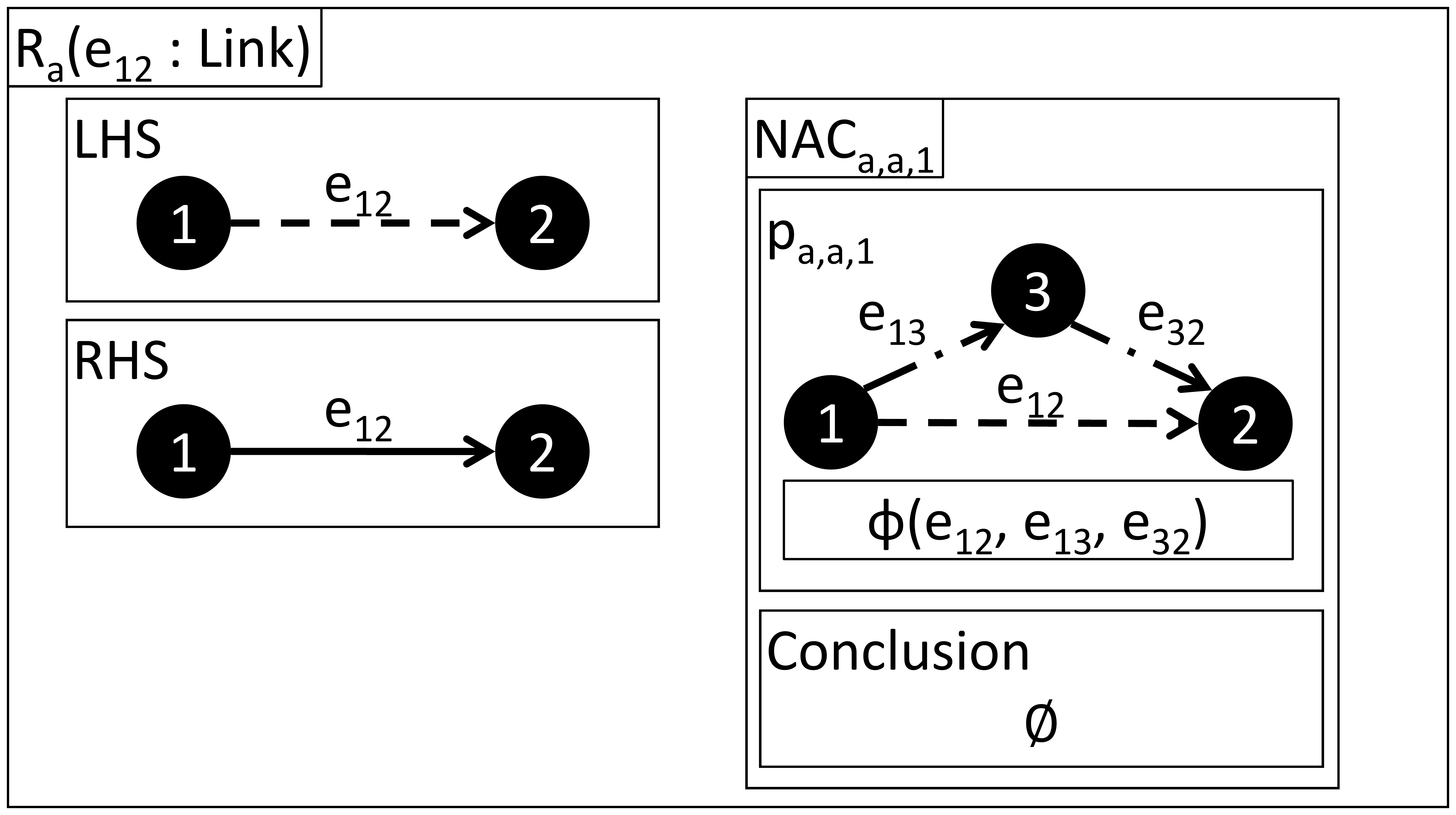}
            \caption{Final \activationRuleLong}
            \label{fig:ActivationRuleFinal}
        \end{subfigure}
        \hspace{3em}
        \begin{subfigure}[t]{.45\subfigWidth}
            \includegraphics[width=.99\textwidth]{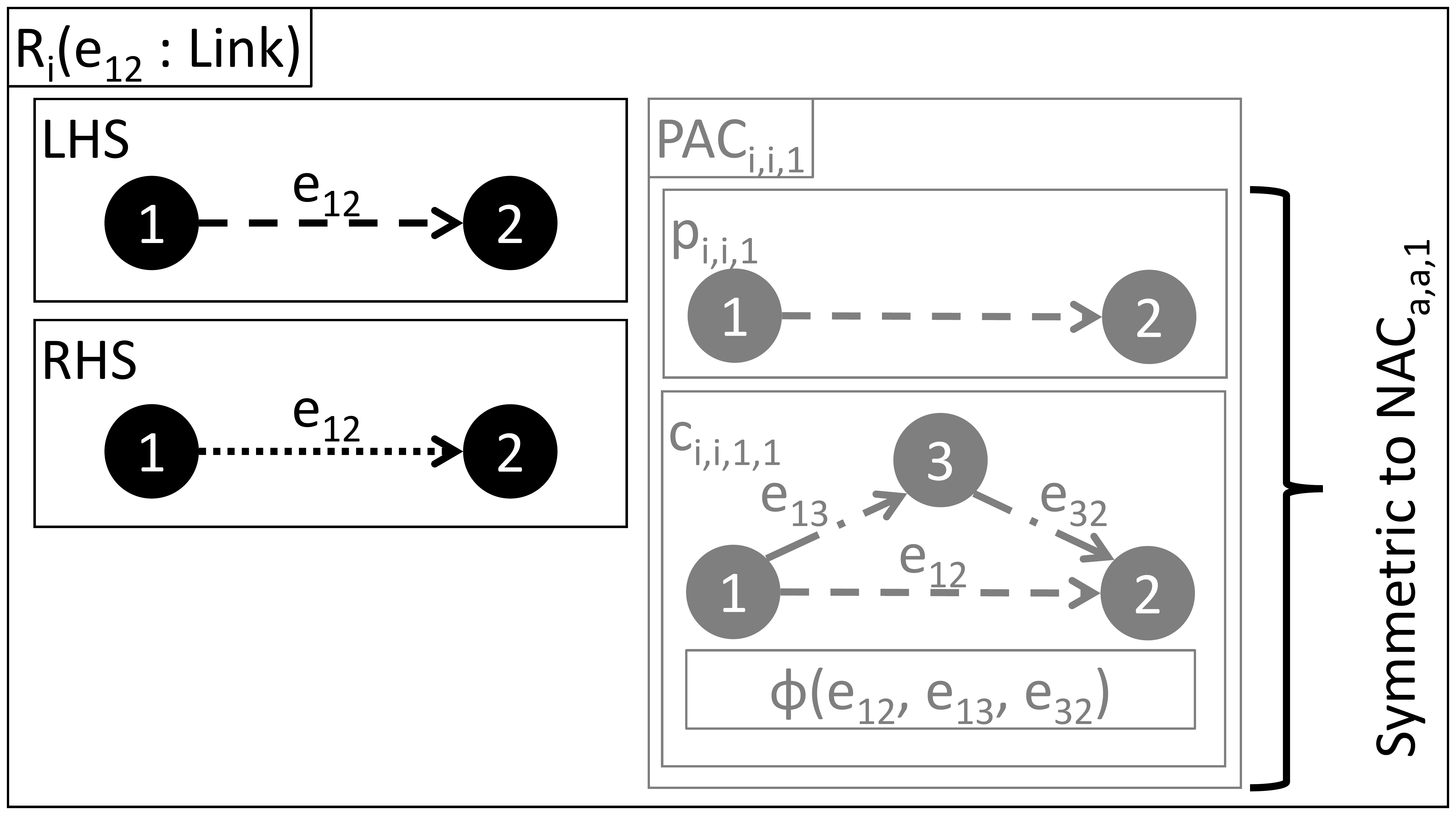}
            \caption{Final \inactivationRuleLong}
            \label{fig:InactivationRuleFinal}
        \end{subfigure}
        
        \begin{subfigure}[t]{\subfigWidth}
            \includegraphics[width=.99\textwidth]{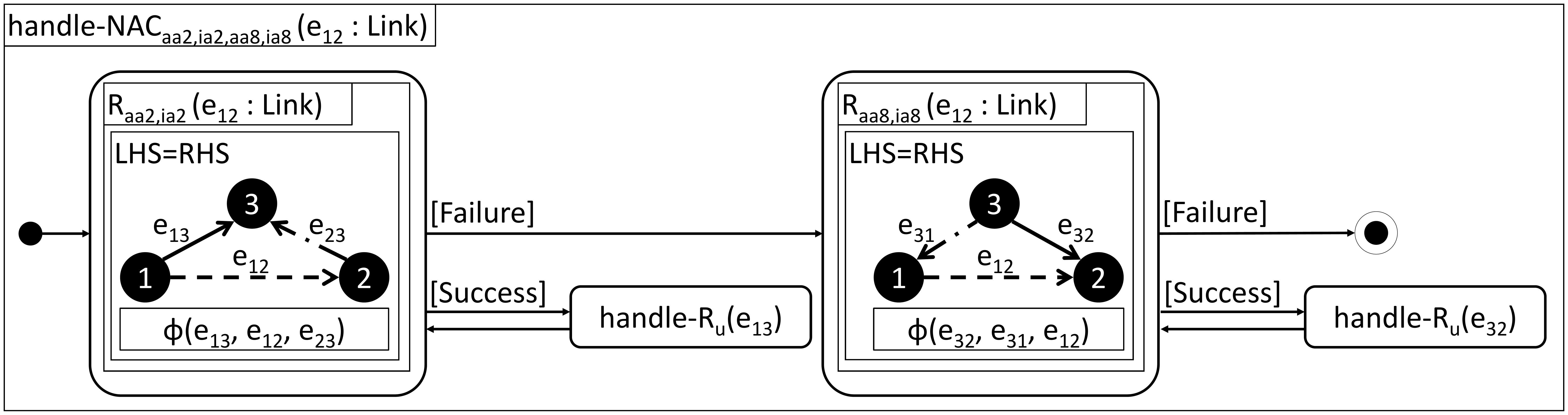}
            \caption{Final \inactivationRuleLong}
            \label{fig:NACHandlingOperationFinal}
        \end{subfigure}
        \caption{Final specification of the \TC algorithm}
        \label{fig:algorithm-final}
    \end{center}
\end{figure}

Inside the \NAC{}-handling operation, the first loop identifies all matches of the premise of \NACaa{2} and \NACia{2} and unclassifies the link \linkVariableOneThree, which is the largest link in the triangle \wrt \linkOrder{A}.
Similarly, the second loop identifies all matches of the premises of \NACaa{8} and \NACia{8} and unclassifies \linkVariableThreeTwo.
Finally, we drop \NACaa{2} and \NACaa{8} from \activationRule and \NACia{2} and \NACia{8} from \inactivationRule.

\paragraph{Proving termination}
Prior to introducing the \NAC{}-handling operation, the number of iterations of the main loop of the \TC algorithm was limited by the initial number of unclassified links.
Now, additional links may become unclassified in each iteration, which requires us to prove that the current \TC algorithm terminates.
\begin{theorem}
    The \TC algorithm $A$ with \NAC{}-handling operation (\Cref{fig:algorithm-final}) terminates for any input topology.
\end{theorem}
\begin{proof}
    \newcommand{\linkStateSequenceComp}{\ensuremath{\sqsubset}}
    Let $E$ be the link set of the processed topology.
    Let $A$ be the considered \TC algorithm, represented by the predicate \tcPredicate{A}{}.
    We consider the sequence of all link states $s_x(e_1),\hdots,s_x(e_m)$ with $m:=|E|$ after the $x$-th execution of the \findUnclassifiedLinkRuleLong, where the links are ordered according to \linkOrder{A}.
    We compare two sequences of link states, $s_{x}$ and $s_{x'}$, as follows:
    $s_{x} \linkStateSequenceComp s_{x'}$ if and only if
    \begin{inparaenum}[(i)]
        \item some link $e_{y}$ is unclassified in $s_{x}$ and classified in $s_{x'}$, and 
        \item the states of all links $e_{y'}$ with $e_{y'} \linkOrder{A} e_{y}$ are identical in $s_{x}$ and $s_{x'}$:
    \end{inparaenum}
    \begin{align*}
        s_{x} \linkStateSequenceComp s_{x'} \;:\Leftrightarrow\; 
        &\exists y \in \naturalNumbers, y \leq m: s_{x}(e_{y}) = \UNCL \wedge s_{x'}(e_{y}) \in \{\ACT,\INACT\}\\
        &\wedge\; \forall y' \in \naturalNumbers, y' < y : s_{x}(e_{y'}) = s_{x'}(e_{y'})
    \end{align*}
    Any sequence of active and inactive links is an upper bound for \linkStateSequenceComp.
    
    We now show that $s_{x-1} \linkStateSequenceComp s_{x}$ for $x > 1$.
    Let $e_{y}$ be the link that is bound by applying the \findUnclassifiedLinkRuleLong.
    The \NAC{}-handling operation unclassifies links $e_{y'}$ with $e_{y'} \linkOrder{A} e_{y}$ and thus $y' < y$.
    The \activationRuleLong or the \inactivationRuleLong activate or inactivate $e_{y}$, respectively.
    
    Therefore, $s_{x-1} \linkStateSequenceComp s_x$ because
    \begin{inparaenum}
    \item the first $y{-}1$ elements of $s_{x-1}$ and $s_x$ are identical, and
    \item $s_{x-1}(e_{y}) = \UNCL$ and $s_x(e_{y}) \in \{\ACT, \INACT\}$.
    \end{inparaenum}
    The termination follows because, for a finite topology, any ordered sequence $s_1 \linkStateSequenceComp s_2 \linkStateSequenceComp \dots$ has finite length.
    
\end{proof}

In this section, we dropped two of three application conditions of both the \activationRuleLong and \inactivationRuleLong.
The remaining application conditions are \NACaa{1} for \activationRule and \PACii{1} for \inactivationRule.
\NACaa{1} and \PACii{1} are complementary in the sense that \PACii{1} is fulfilled if and only if \NACaa{1} is not fulfilled for a link \linkVariableOneTwo.
This means that we may also drop \PACii{1} from \inactivationRule without risking to compromise weak consistency.
Removing \PACii{1} is sensible because the resulting \inactivationRuleLong is more efficient to evaluate and, arguably, easier to understand.

\section{Evaluation}\label{sec:evaluation}

In this section, we present a comparative simulation-based evaluation study that serves as a proof-of-concept for our integration of the \GT tool \eMoflon~\cite{LAS14} and the network simulator \Simonstrator~\cite{RSRS15}.
In \Cref{sec:eval-rqs}, we explain the research questions of this evaluation.
In \Cref{sec:eval-setup}, we describe the technical setup and the configuration parameters.
In \Cref{sec:eval-metrics}, we describe the cost and utility metrics to answer the research questions.
In \Cref{sec:eval-results}, we present the measurement results, discuss them, and answer the research questions.
In \Cref{sec:eval-threats}, we discuss threats to the validity of our findings.

\subsection{Research Questions}
\label{sec:eval-rqs}

Our goal is to investigate the benefits of the proposed approach for specifying families of \TC algorithms by composing their corresponding predicates.
For the sake of conciseness, we investigate the \TC algorithms \ktc and \ektc and the effect of combining each of the algorithms with the \minWeightPredicateLong{}.

\paragraph{\RQekTCLong}
The first research question addresses the performance of \ektc compared to \ktc.
As described in \Cref{sec:ektc}, the major objective of \ektc is to extend the lifetime of the topology, which is achieved by establishing a fairer distribution of the per-node energy consumption.
Our first research question is:
Does \ektc improve the network lifetime compared to \ktc?

\paragraph{\RQMinWeightPredicateLong}
The second research question addresses the performance of applying the \minWeightPredicateLong{}.
A motivation for introducing \minWeightPredicate{} in \Cref{sec:min-weight-predicate} was to reduce the cost of a \TC algorithm in terms of memory and runtime while preserving the quality of the output topology.
Our second research question is:
How does \weightThreshold influence the network lifetime and the resource consumption of \ktc and \ektc?

\subsection{Evaluation Setup}

\label{sec:eval-setup}
Following best practices for simulation studies \cite{Kurkowski2005,Hiranandani2013}, we rigorously document the evaluation setup to foster reproducibility of our results.
The technical platforms of this evaluation are the \GT tool \eMoflon \cite{LAS14} and the network evaluation platform \Simonstrator~\cite{RSRS15} with its contained network simulator \PFS~\cite{SGRNKS11}, as shown in \Cref{fig:evaluation-setup} and as described in detail in the following.
A SHARE virtual machine~\cite{GM11} is available for exploring our tool integration.%
\footnote{See \url{https://github.com/Echtzeitsysteme/CorrectByConstructionTCFamilies-SoSyM17}}
\begin{figure}
    \begin{center}
        \includegraphics[width=\textwidth]{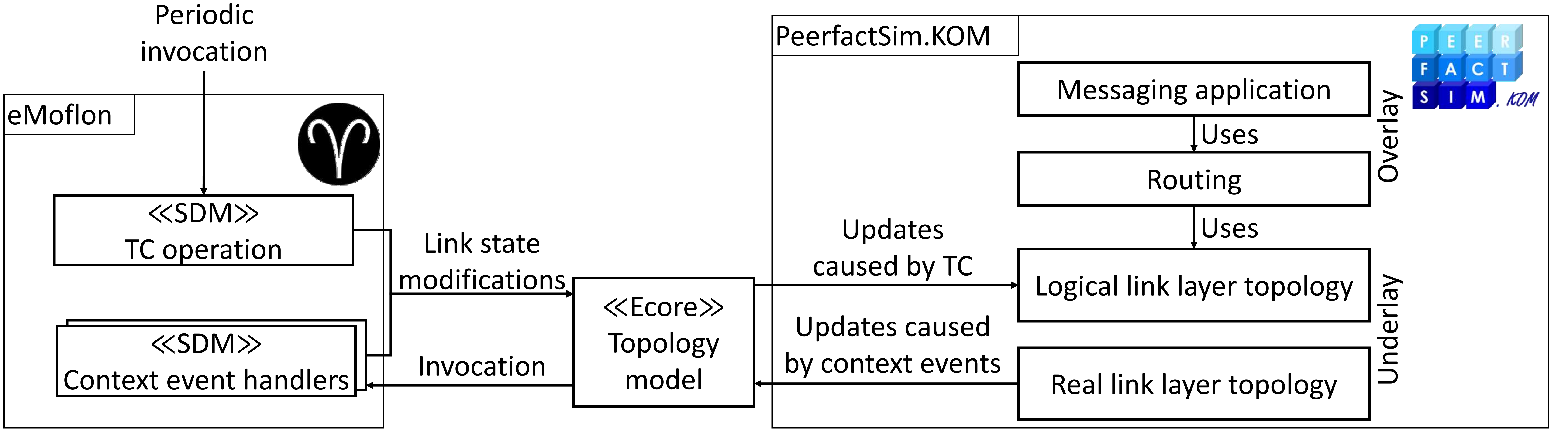}
        \caption{Overview of the evaluation setup}
        \label{fig:evaluation-setup}
    \end{center}
\end{figure}

\paragraph{\GT tool}
All \TC algorithms have been specified using a visual syntax of story diagrams in \eMoflon~\cite{LAS14}.
The specification is used to generate EMF-compliant Java code, which builds on an Ecore-based~\cite{SBMP08} topology metamodel.
Whenever the topology model is modified by a \CE in the simulator, the corresponding handler operation is triggered.
The \TC algorithm periodically updates the topology model via \LSMs.

\paragraph{Network simulator}
\PFS~\cite{SGRNKS11} is a time-discrete network simulator that allows to simulate protocols on the underlay (\idest, the physical and link layer~\cite{Zimmermann1980}) as well as on the overlay (\idest, network, transport and application layer~\cite{Zimmermann1980}).

\paragraph{Node configuration}
The nodes are initially placed uniformly at random onto the simulation area.
Each node moves according to the Gauss-Markov movement model~\cite{Camp2002}.
As link weight, we use the Euclidean distance of its end nodes.
A calibration curve estimates the required transmission power of each link based on its weight.
To create the calibration curve, we used a topology consisting of a sender and a receiver node and set up a simple data transmission application that transfers data from the sender to the receiver at a rate of \SI{1}{\mega\byte\per\second}.

\paragraph{Underlay}
We implemented a generic topology control framework inside the simulator.
This framework presents a \emph{logical link layer topology} to overlay applications, which consists of all active links in the \emph{real link layer topology}.
The real link layer topology is updated by the movement model and the energy model, which describes the battery depletion of each node.
In this evaluation, \TC is invoked periodically every \evalParTCInterval of simulated time.
The set of \CEs is collected during this period and handled before invoking \TC{}.
A \emph{\TC run} consists of the \CE handling and subsequent \TC invocation.
The parameter \ktcParameterK of \ktc and \ektc is set to \evalParK.
This is a typical value~\cite{Stein2016b,SKSVSM15,KSSVMS14,SWBM12} because it is the largest value for $k$ that allows to prove certain advanced properties (such as $\theta$-separation~\cite{SWBM12}).
The \expectedRemainingLifetimeLong{}{\linkVariableab} of each link \linkVariableab is calculated by dividing the energy level \energy{\nodeVariablea}{} of \nodeVariablea by the expected transmission power \expectedPower{}{\linkVariableab} of \linkVariableab.

\paragraph{Overlay}
On the overlay, we run a messaging application, which emulates sporadic message exchange between random pairs of nodes.
Every \evalParMsgInterval, each node sends a message of size \evalParMsgSize to another node that is selected randomly.
We apply a routing algorithm that builds routing paths based on global knowledge.

\paragraph{Summary of parameters}
\Cref{tab:evaluation-parameters-overview} summarizes the parameters of the simulation setting that are fixed for all configurations.
We run each simulation in \emph{terminating mode}, \idest, we abort the simulation after \evalParSimulationDuration, regardless of the remaining number of alive nodes.
This means that, in each simulation run, \TC is performed $\frac{\evalParSimulationDuration}{\evalParTCInterval} =  \evalParMeasurementsPerRun$ times.
All experiments were performed on a 64-bit machine with an Intel i7-2600~CPU (\SI{2}{cores}, \SI{3.7}{\giga\hertz}) and \SI{8}{\giga\byte} of RAM running Windows 7 Professional.
\begin{table}
\begin{minipage}{\textwidth}
\begin{center}        
    \caption{Parameters of the evaluation}
    \label{tab:evaluation-parameters-overview}
    \begin{tabular}{ll}
        \toprule
        \textbf{Parameter} & \textbf{Value}\\
        \midrule
        \multicolumn{2}{l}{\textbf{Simulator}}\\
        \midrule
        Network simulator & \Simonstrator/\PFS 2.4~\cite{RSRS15,SGRNKS11}\\
        Graph transformation tool & \eMoflon 2.16.0~\cite{LAS14}\\
        Runs per configuration & \evalParSeedCount\\
        Seeds & 1\dots\evalParSeedCount\\
        Simulation type & Terminating (after \evalParSimulationDuration of simulated time)\\
        Measurements per run & \evalParMeasurementsPerRun\\
        Code availability & As SHARE~\cite{GM11} virtual machine\footnote{See \url{https://github.com/Echtzeitsysteme/CorrectByConstructionTCFamilies-SoSyM17}}\\
        World size & \emph{Variable}  (see \Cref{tab:evaluation-configurations-overview})\\
        \midrule
        \multicolumn{2}{l}{\textbf{Node Configuration}}\\
        \midrule
        Node count& \evalParNodeCount \\
        Transmission radius & \evalParTransmissionRange\\
        Movement model & Gauss-Markov movement~\cite{Camp2002} ($\alpha = 0.2$, $v = \SI{0.005}{\meter\per\second}$)\\
        Placement model (initial) & Uniform at random\\       
        Battery level (initial) & Uniform at random betw. \SI{30}{\percent} and \SI{100}{\percent} of \evalParBatterySource\\
        \midrule
        \multicolumn{2}{l}{\textbf{Underlay}}\\
        \midrule
        Link count (initial) & \emph{Variable} (see \Cref{tab:evaluation-configurations-overview})\\
        Link layer & IEEE 802.11 Ad-Hoc\\
        \TC algorithm & \ktc, \ektc\\
        \weight{\linkVariableab} & Euclidean distance between \nodeVariablea and \nodeVariableb\\
        \expectedRemainingLifetime{1}{\linkVariableab} & Based on calibration curve\\
        \ktcParameterK parameter of \ktc & \evalParK\\
        Interval of \TC runs& \evalParTCInterval\\
        \TC runs per simulation run & \evalParMeasurementsPerRun\\
        \midrule
        \multicolumn{2}{l}{\textbf{Overlay}}\\
        \midrule
        Application & Messaging application\\
        Communication pattern & Many-to-many\\
        Message size & \evalParMsgSize\\
        Messaging interval & \evalParMsgInterval\\
        \bottomrule
    \end{tabular}
\end{center}
\end{minipage}
\end{table}

To obtain a set of representative \emph{(topology) configurations}, we vary the side length of the quadratic simulation area.
In total, we investigate \evalParNumberOfConfigurations different evaluation configurations, summarized in \Cref{tab:evaluation-configurations-overview}, which correspond to world sizes $\metricWorldSize \in \{\SI{750}{\meter}, \SI{500}{\meter}\}$.
Each configuration is simulated \evalParSeedCount times to obtain representative results.
In the table, the configurations are ordered from the sparsest (\cfgSmallSparse) to the densest topology (\cfgSmallMedium).
\begin{table}
    \begin{center}
        \caption{Varying parameters in the two topology configurations (\textbf{$n$}: initial node count,
            \textbf{$m$}: initial link count in real link layer, 
            \textbf{\metricWorldSize}: side length of quadratic simulation area,
            \textbf{\metricOutdegree}: average out-degree,
            values of $m$, $h$, and \metricOutdegree are averaged over \evalParSeedCount repetitions
            )}
        \label{tab:evaluation-configurations-overview}
        \begin{tabular}{crrrrr}
            \toprule
            \multicolumn{1}{c}{\textbf{ID}} & 
            \multicolumn{1}{c}{\textbf{$n$}} & 
            \multicolumn{1}{c}{\textbf{\metricWorldSize [m]}} & 
            \multicolumn{1}{c}{$m$} & 
            \multicolumn{1}{c}{\textbf{\metricOutdegree}}
            \\
            \midrule
            \cfgSmallSparse&100&750&\numprint{812}&7.7\\
            \cfgSmallMedium&100&500&\numprint{1651}&16.0\\
            \bottomrule
        \end{tabular}
    \end{center}
\end{table}
The node (out-)degree \metricOutdegree is an indicator of the density of the topology:
A high node degree indicates that a topology is dense, while a low node degree indicates a sparse topology.

\subsection{Metrics}
\label{sec:eval-metrics}

The major goal of this evaluation is to assess cost and utility metrics of the output topology of a \TC algorithm.
With \GwThresh, we denote the output topology of the currently active \TC algorithm, whose input topology is obtained by removing all links with a weight smaller than \weightThreshold from the physical topology.

As utility metric, we consider the network lifetime (see also \Cref{sec:ektc}).
We record for each simulation run the \remainingLifetimeLong{1}{}, the \remainingLifetimePctLong{50}{}, and \remainingLifetimePctLong{100}{}.
These values mark the starting point, the approximate middle, and the end of the degradation of the network.

As cost metric, we assess the required memory in terms of the size of the maintained input topology, and the required runtime in terms of the required CPU execution time as well as the number of \LSMs for performing the \TC algorithm.
The size \topologySize{\GwThresh} of the topology $\GwThresh$ is the sum of its node and edge counts.
The \emph{execution time  \executionTime{\GwThresh}} is the real time (in contrast to the simulated time) that it takes to perform a \TC run on \GwThresh.
The \emph{\LSM count \lsmCount{\GwThresh}} is determined analogously by counting the \LSMs during a \TC run on \GwThresh.

We compare results along two dimensions:
Either we consider a fixed configuration (\idest, \cfgSmallMedium or \cfgSmallSparse) and compare the performance of the \TC algorithms, or we fix the world size, node count,  \TC algorithm and \ktcParameterK-value and evaluate the influence of varying \weightThreshold.
For each dimension, we introduce relative metrics \relativeMetric{M}{} for each absolute metric $M$.
When comparing a \TC algorithm $A$ with the baseline \MaxpowerTC algorithm, \relativeMetric{M}{} is defined as follows:
\begin{align*}
&\relativeRemainingLifetime{1}{A} = \frac{\remainingLifetime{1}{A}}{\remainingLifetime{1}{\text{\MaxpowerTC}}},
&&\relativeRemainingLifetimePct{50}{A} = \frac{\remainingLifetimePct{50}{A}}{\remainingLifetimePct{50}{\text{\MaxpowerTC}}},\\
&\relativeRemainingLifetimePct{100}{A} = \frac{\remainingLifetimePct{100}{A}}{\remainingLifetimePct{100}{\text{\MaxpowerTC}}},
&&\relativeTopologySize{A} = \frac{\topologySize{A}}{\topologySize{\text{\MaxpowerTC}}},\\
&\relativeExecutionTime{A} = \frac{\executionTime{A}}{\executionTime{\text{\MaxpowerTC}}},
&&\relativeLsmCount{A} = \frac{\lsmCount{A}}{\lsmCount{\text{\MaxpowerTC}}}.\\
\end{align*}
When evaluating the influence of varying \weightThreshold, \relativeMetric{M}{} is defined as follows:
\begin{align*}
&\relativeRemainingLifetime{1}{\GwThresh} = \frac{\remainingLifetime{1}{\GwThresh}}{\remainingLifetime{1}{G_0}},
&&\relativeRemainingLifetimePct{50}{\GwThresh} = \frac{\remainingLifetimePct{50}{\GwThresh}}{\remainingLifetimePct{50}{G_0}},\\
&\relativeRemainingLifetimePct{100}{\GwThresh} = \frac{\remainingLifetimePct{100}{\GwThresh}}{\remainingLifetimePct{100}{G_0}},
&&\relativeTopologySize{\GwThresh} = \frac{\topologySize{\GwThresh}}{\topologySize{G_0}},\\
&\relativeExecutionTime{\GwThresh} = \frac{\executionTime{\GwThresh}}{\executionTime{G_0}},
&&\relativeLsmCount{\GwThresh} = \frac{\lsmCount{\GwThresh}}{\lsmCount{G_0}}.\\
\end{align*}
Setting \weightThreshold to 0 is equivalent to disabling the minimum-weight optimization.

\subsection{Results and Discussion}
\label{sec:eval-results}

The plots in \Cref{fig:eval-alive-nodes-ws0500-nc0099,fig:eval-alive-nodes-ws0750-nc0099} show the development of the number of alive nodes for the dense topology (\Cref{fig:eval-alive-nodes-ws0500-nc0099}) and the sparse topology (\Cref{fig:eval-alive-nodes-ws0750-nc0099}) for different minimum-weight thresholds $\weightThreshold \in \{0, 20, 40, 60, 80\}$.
In each plot, the simulation runs of \MaxpowerTC, \ktc, and \ektc are shown.
\foreach \worldSize/\nodeCount/\printWS/\printNodeCount in {0500/0099/500/100, 0750/0099/750/100}{
    \foreach \k/\printK in {1.410/1.41}{ 
        \begin{figure}
        \begin{center}
        \newcommand{\subfigWidth}{\textwidth}
        
        \foreach \x in {0.0, 20.0, 40.0, 60.0, 80.0} 
        {
            \begin{subfigure}[t]{\subfigWidth}
                \includegraphics[width=.99\textwidth]{\rootOfAliveNodePlots/\aliveNodesSuffix{\worldSize}{\nodeCount}{\k}{\x}.pdf}
                \caption{$\weightThreshold = \SI{\x}{\meter}$}
            \end{subfigure}
        }
        
        \caption{Alive nodes over simulation time (World size: \SI{\printWS}{\meter}, node count: \printNodeCount, \ktcParameterK: \printK)}
        \label{fig:eval-alive-nodes-ws\worldSize-nc\nodeCount}
        
        \end{center}
        \end{figure}
    } 
} 

While the plots provide a rather qualitative view of the performance of the \TC algorithms,  \Cref{tab:eval-chart-ws500-nc99-k1.41,tab:eval-chart-ws750-nc99-k1.41} 
allow for a fine-grained analysis of the cost and utility metrics, as introduced in \Cref{sec:eval-metrics}.
We provide mean values of all \evalParMeasurementsPerRun \TC runs per simulation for the topology size (\meanTopologySize{}), the execution time (\meanExecutionTime{}), and the \LSM count (\meanLsmCount{}).
\begin{table}[ht]
\centering
\caption{Algorithm performance for different weight thresholds \weightThreshold (World size: \SI{500}{\meter}, node count: 100, $k$: 1.41).
Lifetime values (\remainingLifetime{1}{}, \remainingLifetime{\SI{50}{\percent}}{}, \remainingLifetime{\SI{100}{\percent}}{}) are in simulated minutes. 
Execution time \meanExecutionTime{} is in CPU milliseconds.
\relativeMetric{M}{}: relative metrix $M$.
The best values for each combination of algorithm, \weightThreshold, and metric $M$ are highlighted in bold font.
} 
\label{tab:eval-chart-ws500-nc99-k1.41}
\begin{tabular}{cr|p{1.3cm}p{0.9cm}p{1.3cm}p{0.9cm}p{1.4cm}p{0.9cm}}
\toprule
\multicolumn{1}{c}{\textbf{Algo.}} & \multicolumn{1}{c|}{\textbf{\weightThreshold [m]}} & \multicolumn{1}{c}{\textbf{\remainingLifetime{1}{} [min]}} & \multicolumn{1}{c}{\textbf{\relativeRemainingLifetime{1}{}}} & \multicolumn{1}{c}{\textbf{\remainingLifetimePct{50}{} [min]}} & \multicolumn{1}{c}{\textbf{\relativeRemainingLifetimePct{50}{}}} & \multicolumn{1}{c}{\textbf{\remainingLifetimePct{100}{} [min]}} & \multicolumn{1}{c}{\textbf{\relativeRemainingLifetimePct{100}{}}}
\\
\midrule
Maxp. & 0 & 564.0 & 1.00 & 700.0 & 1.00 & 854.0 & 1.00 \\ 
   \midrule
  kTC & 0 & 544.0 & 1.00 & 664.0 & 1.00 & 862.0 & 1.00 \\ 
  kTC & 20 & 544.0 & 1.00 & 664.0 & 1.00 & 864.0 & 1.00 \\ 
  kTC & 40 & 570.0 & 1.05 & 684.0 & 1.03 & \bstVal{868.0} & \bstVal{1.01} \\ 
  kTC & 60 & \bstVal{580.0} & \bstVal{1.07} & 698.0 & 1.05 & 862.0 & 1.00 \\ 
  kTC & 80 & 564.0 & 1.04 & \bstVal{702.0} & \bstVal{1.06} & 860.0 & 1.00 \\ 
   \midrule
  e-kTC & 0 & \bstVal{624.0} & \bstVal{1.00} & \bstVal{776.0} & \bstVal{1.00} & \bstVal{906.0} & \bstVal{1.00} \\ 
  e-kTC & 20 & \bstVal{624.0} & \bstVal{1.00} & 770.0 & 0.99 & 902.0 & 1.00 \\ 
  e-kTC & 40 & 616.0 & 0.99 & 752.0 & 0.97 & \bstVal{906.0} & \bstVal{1.00} \\ 
  e-kTC & 60 & 600.0 & 0.96 & 726.0 & 0.94 & 882.0 & 0.97 \\ 
  e-kTC & 80 & 590.0 & 0.95 & 712.0 & 0.92 & 864.0 & 0.95 \\ 
   \bottomrule
\end{tabular}

\vspace{1ex}\nopagebreak
\begin{tabular}{cr|p{1.3cm}p{0.9cm}p{1.3cm}p{0.9cm}p{1.4cm}p{0.9cm}}
  \toprule
\multicolumn{1}{c}{\textbf{Algo.}} & \multicolumn{1}{c|}{\textbf{\weightThreshold [m]}} & \multicolumn{1}{c}{\textbf{\meanTopologySize{}}} & \multicolumn{1}{c}{\textbf{\relativeTopologySize{}}} & \multicolumn{1}{c}{\textbf{\meanExecutionTime{} [ms]}} & \multicolumn{1}{c}{\textbf{\relativeExecutionTime{}}} & \multicolumn{1}{c}{\textbf{\meanLsmCount{}}} & \multicolumn{1}{c}{\textbf{\relativeLsmCount{}}} \\ 
  \midrule
Maxp. & 0 & 822.3 & 1.00 & 19.8 & 1.00 & 18.1 & 1.00 \\ 
   \midrule
kTC & 0 & 774.7 & 1.00 & 253.4 & 1.00 & 1340.3 & 1.00 \\ 
  kTC & 20 & 775.0 & 1.00 & 227.6 & 0.90 & 1252.2 & 0.93 \\ 
  kTC & 40 & 768.7 & 0.99 & 147.8 & 0.58 & \hide{0}998.7 & 0.75 \\ 
  kTC & 60 & 712.7 & 0.92 & \hide{0}83.1 & 0.33 & \hide{0}728.0 & 0.54 \\ 
  kTC & 80 & \bstVal{603.2} & \bstVal{0.78} & \bstVal{\hide{0}49.9} & \bstVal{0.20} & \bstVal{\hide{0}504.7} & \bstVal{0.38} \\ 
   \midrule
e-kTC & 0 & 927.4 & 1.00 & 240.7 & 1.00 & 1303.1 & 1.00 \\ 
  e-kTC & 20 & 912.9 & 0.98 & 211.6 & 0.88 & 1218.0 & 0.93 \\ 
  e-kTC & 40 & 853.5 & 0.92 & 151.8 & 0.63 & 1031.2 & 0.79 \\ 
  e-kTC & 60 & 747.3 & 0.81 & 115.8 & 0.48 & \hide{0}791.2 & 0.61 \\ 
  e-kTC & 80 & \bstVal{622.2} & \bstVal{0.67} & \bstVal{\hide{0}64.0} & \bstVal{0.27} & \bstVal{\hide{0}555.2} & \bstVal{0.43} \\ 
   \bottomrule
\end{tabular}
\end{table}

\begin{table}[ht]
\centering
\caption{Algorithm performance for different weight thresholds \weightThreshold (world size: \SI{750}{\meter}, node count: 100, $k$: 1.41).
Lifetime values (\remainingLifetime{1}{}, \remainingLifetime{\SI{50}{\percent}}{}, \remainingLifetime{\SI{100}{\percent}}{}) are in simulated minutes. 
Execution time \meanExecutionTime{} is in CPU milliseconds.
\relativeMetric{M}{}: relative $M$ compared to same algorithm with $\weightThreshold = 0$.
The best values for each combination of algorithm, \weightThreshold, and metric $M$ are highlighted in bold font.
} 
\label{tab:eval-chart-ws750-nc99-k1.41}
\begin{tabular}{cr|p{1.3cm}p{0.9cm}p{1.3cm}p{0.9cm}p{1.4cm}p{0.9cm}}
\toprule
\multicolumn{1}{c}{\textbf{Algo.}} & \multicolumn{1}{c|}{\textbf{\weightThreshold [m]}} & \multicolumn{1}{c}{\textbf{\remainingLifetime{1}{} [min]}} & \multicolumn{1}{c}{\textbf{\relativeRemainingLifetime{1}{}}} & \multicolumn{1}{c}{\textbf{\remainingLifetimePct{50}{} [min]}} & \multicolumn{1}{c}{\textbf{\relativeRemainingLifetimePct{50}{}}} & \multicolumn{1}{c}{\textbf{\remainingLifetimePct{100}{} [min]}} & \multicolumn{1}{c}{\textbf{\relativeRemainingLifetimePct{100}{}}}
\\
\midrule
Maxp. & 0 & 552.0 & 1.00 & 668.0 & 1.00 & 884.0 & 1.00 \\ 
\midrule
kTC & 0 & 532.0 & 1.00 & 646.0 & 1.00 & 906.0 & 1.00 \\ 
  kTC & 20 & 532.0 & 1.00 & 650.0 & 1.01 & 942.0 & 1.04 \\ 
  kTC & 40 & 546.0 & 1.03 & 656.0 & 1.02 & 892.0 & 0.98 \\ 
  kTC & 60 & 554.0 & 1.04 & \bstVal{670.0} & \bstVal{1.04} & 894.0 & 0.99 \\ 
  kTC & 80 & \bstVal{556.0} & \bstVal{1.05} & \bstVal{670.0} & \bstVal{1.04} & \bstVal{956.0} & \bstVal{1.06} \\ 
\midrule
e-kTC & 0 & \bstVal{592.0} & \bstVal{1.00} & \bstVal{710.0} & \bstVal{1.00} & 958.0 & 1.00 \\ 
  e-kTC & 20 &\bstVal{592.0} & \bstVal{1.00} & 702.0 & 0.99 & \bstVal{966.0} & \bstVal{1.01} \\ 
  e-kTC & 40 & 590.0 & 1.00 & 696.0 & 0.98 & 902.0 & 0.94 \\ 
  e-kTC & 60 & 584.0 & 0.99 & 688.0 & 0.97 & 958.0 & 1.00 \\ 
  e-kTC & 80 & 574.0 & 0.97 & 684.0 & 0.96 & 944.0 & 0.99 \\ 
\bottomrule
\end{tabular}

\vspace{1ex}\nopagebreak
\begin{tabular}{cr|p{1.3cm}p{0.9cm}p{1.3cm}p{0.9cm}p{1.4cm}p{0.9cm}}
  \toprule
\multicolumn{1}{c}{\textbf{Algo.}} & \multicolumn{1}{c|}{\textbf{\weightThreshold [m]}} & \multicolumn{1}{c}{\textbf{\meanTopologySize{}}} & \multicolumn{1}{c}{\textbf{\relativeTopologySize{}}} & \multicolumn{1}{c}{\textbf{\meanExecutionTime{} [ms]}} & \multicolumn{1}{c}{\textbf{\relativeExecutionTime{}}} & \multicolumn{1}{c}{\textbf{\meanLsmCount{}}} & \multicolumn{1}{c}{\textbf{\relativeLsmCount{}}} \\ 
  \midrule
Maxp. & 0 & 412.9 & 1.00 & 8.7 & 1.00 & 9.0 & 1.00 \\ 
   \midrule
  kTC & 0 & 401.5 & 1.00 & 36.5 & 1.00 & 604.9 & 1.00 \\ 
  kTC & 20 & 399.9 & 1.00 & 34.3 & 0.94 & 567.2 & 0.94 \\ 
  kTC & 40 & 390.8 & 0.97 & 27.7 & 0.76 & 473.4 & 0.78 \\ 
  kTC & 60 & 366.4 & 0.91 & 22.0 & 0.60 & 361.3 & 0.60 \\ 
  kTC & 80 & \bstVal{319.7} & \bstVal{0.80} & \bstVal{17.6} & \bstVal{0.48} & \bstVal{245.3} & \bstVal{0.41} \\ 
   \midrule
  e-kTC & 0 & 443.6 & 1.00 & 32.8 & 1.00 & 563.0 & 1.00 \\ 
  e-kTC & 20 & 435.9 & 0.98 & 31.5 & 0.96 & 532.9 & 0.95 \\ 
  e-kTC & 40 & 418.7 & 0.94 & 30.7 & 0.93 & 470.1 & 0.83 \\ 
  e-kTC & 60 & 383.8 & 0.87 & 25.8 & 0.79 & 383.2 & 0.68 \\ 
  e-kTC & 80 & \bstVal{329.1} & \bstVal{0.74} & \bstVal{19.8} & \bstVal{0.60} & \bstVal{271.7} & \bstVal{0.48} \\ 
   \bottomrule
\end{tabular}
\end{table}

\paragraph{Network lifetime}
We observe that \ktc almost always performs worst and \ektc performs best of all three \TC algorithms \wrt the number of alive nodes.
More precisely, for a fixed value of \weightThreshold, the lifetime of $\GwThresh$ for \ktc is always shorter than for \ektc (see also \Cref{tab:eval-chart-ws500-nc99-k1.41,tab:eval-chart-ws750-nc99-k1.41}).
Also, the remaining lifetime of \ktc increases with increasing \weightThreshold.
For instance, all network lifetime values of the sparse topology are maximal when applying \ktc with $\weightThreshold = \SI{80}{\meter}$ (\remainingLifetime{1}{} = \SI{556}{\min} compared to \remainingLifetime{1}{} = \SI{532}{\min}, \SI{5}{\percent} improvement, \Cref{fig:eval-alive-nodes-ws0750-nc0099}).
A similar, yet not as strong, effect can be observed for \ektc.
For instance, in the sparse topology, \remainingLifetime{\SI{100}{\percent}}{} for $\weightThreshold = \SI{20}{\meter}$ is about \SI{1}{\percent} larger than \remainingLifetime{\SI{100}{\percent}}{} for $\weightThreshold = \SI{0}{\meter}$.
In the remaining cases, the network lifetime for \ektc is maximal for $\weightThreshold = \SI{0}{\meter}$.

\paragraph{Resource consumption}
We now focus on resulting topology size, real execution time, and required \LSM count.
These cost metrics evaluate the resource consumption of a \TC algorithm.
In contrast to the network lifetime, all of the cost metrics strictly decrease with \weightThreshold and reach their minimum for $\weightThreshold = \SI{80}{\meter}$.

For the dense topology (\cfgSmallMedium, \Cref{fig:eval-alive-nodes-ws0500-nc0099}), the size of the stored topology was reduced by up to \SI{22}{\percent} (for \ktc) and \SI{33}{\percent} (for \ektc).
The savings in execution time are even more drastic: \SI{80}{\percent} for \ktc and \SI{73}{\percent} for \ektc.
Finally, the savings \wrt \LSM count amount up to \SI{72} {\percent} for \ktc and \SI{57}{\percent} for \ektc.

For the sparser topology (\cfgSmallSparse, \Cref{fig:eval-alive-nodes-ws0750-nc0099}), the savings \wrt topology size are comparable (\SI{20}{\percent} for \ktc and \SI{26}{\percent} for \ektc).
The savings \wrt execution time (\SI{62}{\percent} for \ktc and \SI{40}{\percent} for \ektc) and \LSM count (\SI{59}{\percent} for \ktc and \SI{52}{\percent} for \ektc) are smaller.

\paragraph{Trade-off between cost and utility}
Surprisingly, the lifetime when applying \ktc was generally very good for $\weightThreshold = \SI{80}{\meter}$.
Therefore, we do not have to trade network lifetime against resource consumption in this situation.
Instead, we can optimize resource consumption and network lifetime at the same time by choosing high values of \weightThreshold.
This observation underpins the intuition behind the minimum-weight predicate that considering short links does probably not help increasing the performance of the network.
Our results even show that performing \ktc also on short links may be harmful to the network lifetime.
A possible explanation is that short links are probably helpful when used for direct message transfer (compared to a multi-hop forwarding of the same message).
Despite this positive result, we have to remark that \ktc is often not able to beat the baseline, \idest, \MaxpowerTC, \wrt network lifetime.

In contrast, \ektc shows a rather monotonic behavior:
In most cases, the network lifetime decreases with increasing \weightThreshold.
Here, we trade the reduced resource consumption against the shorter network lifetime.
In the most extreme case ($\weightThreshold = \SI{80}{\meter}$), the lifetime drops to values between \SI{92}{\percent} and \SI{95}{\percent} (for the dense topology) and between \SI{96}{\percent} and \SI{99}{\percent} (for the sparse topology) of its maximum value.
Choosing the best trade-off between network lifetime and resource consumption for \ektc is certainly a task that depends on the available resources of the selected (simulated) target platform.

\Cref{tab:network-lifetime-comparison} summarizes the best performance values of \MaxpowerTC, \ktc, and \ektc in \Cref{tab:eval-chart-ws500-nc99-k1.41,tab:eval-chart-ws750-nc99-k1.41}.
In this table, \relativeMetric{M}{} denotes the relative value of metric $M$ compared to the baseline \MaxpowerTC algorithm.
\ektc achieves network lifetime improvements between \pct{6.0} and \pct{10.9} for the dense topology (\cfgSmallMedium) and between \pct{6.3} and \pct{9.3} for the sparse topology (\cfgSmallSparse).
At the same time, \ktc only achieves improvements between \pct{0.3} and \pct{1.6} for the dense topology (\cfgSmallMedium) and between \pct{0.3} and \pct{8.1} for the sparse topology (\cfgSmallSparse).
The execution times of \ktc and \ektc grow by a factor of two to three compared to the execution time of \MaxpowerTC, and the number of required \LSMs even increases by a factor of 25 to 30.
This is not at all surprising because the logic behind \MaxpowerTC is simple compared to the rules of \ktc and \ektc.
Especially handling \CEs results in additional required \LSMs.
When comparing \ktc with \ektc \wrt cost metrics, \ktc outperforms \ektc.
Its best execution time is \pct{22.0} (for \cfgSmallMedium) and \pct{11.1} (for \cfgSmallSparse) lower than the best execution time of \ektc.
Similarly, \ktc outperforms \ektc by \pct{9.1} (for \cfgSmallMedium) and \pct{9.7} (for \cfgSmallSparse) \wrt the best achieved \LSM count.
However, the configuration that achieves the best execution time and \LSM count for \ektc perform worse than the best-performing configuration of \ktc \wrt network lifetime.
Therefore, an immediate comparison of \ktc and \ektc is not salient.
Instead, a trade-off between cost and utility is necessary also here.
\begin{table}
\begin{center}
\caption{
Performance of \ktc and \ektc compared to \MaxpowerTC
}
\label{tab:network-lifetime-comparison}
\begin{tabular}{cr|p{1.3cm}p{0.9cm}p{1.3cm}p{0.9cm}p{1.4cm}p{0.9cm}}
\toprule
\multicolumn{1}{c}{\textbf{Config.}} & 
\multicolumn{1}{c|}{\textbf{Algo.}} & \multicolumn{1}{c}{\textbf{\remainingLifetime{1}{} [min]}} & \multicolumn{1}{c}{\textbf{\relativeRemainingLifetime{1}{}}} & \multicolumn{1}{c}{\textbf{\remainingLifetimePct{50}{} [min]}} & \multicolumn{1}{c}{\textbf{\relativeRemainingLifetimePct{50}{}}} & \multicolumn{1}{c}{\textbf{\remainingLifetimePct{100}{} [min]}} & \multicolumn{1}{c}{\textbf{\relativeRemainingLifetimePct{100}{}}}
\\
\midrule
\cfgSmallMedium & Maxp. 	& 564.0 & -- &700.0 & -- &854.0 \\
\cfgSmallMedium & \ktc 			& 580.0 & +\SI{2.8}{\percent} & 702.0 & +\SI{0.3}{\percent} & 868.0 & +\SI{1.6}{\percent} \\
\cfgSmallMedium & \ektc 		& 624.0 & +\SI{10.6}{\percent} & 776.0 & +\SI{10.9}{\percent} & 906.0 & +\SI{6.0}{\percent} \\
\midrule
\cfgSmallSparse & Maxp. & 552.0 & -- & 668.0 & -- & 884.0 \\
\cfgSmallSparse & \ktc 			& 556.0 & +\SI{0.7}{\percent} & 670.0 & +\SI{0.3}{\percent} & 956.0 & +\SI{8.1}{\percent}\\
\cfgSmallSparse & \ektc 		& 592.0 & +\SI{7.2}{\percent} & 710.0 & +\SI{6.3}{\percent} & 966.0 & +\SI{9.3}{\percent}\\
\bottomrule
\end{tabular}

\nopagebreak\vspace{1ex}
\begin{tabular}{cr|p{1.3cm}p{0.9cm}p{1.3cm}p{0.9cm}p{1.4cm}p{0.9cm}}
\toprule
\multicolumn{1}{c}{\textbf{Config.}} & 
\multicolumn{1}{c|}{\textbf{Algo.}} & \multicolumn{1}{c}{\textbf{\meanTopologySize{}}} & \multicolumn{1}{c}{\textbf{\relativeTopologySize{}}} & \multicolumn{1}{c}{\textbf{\meanExecutionTime{} [ms]}} & \multicolumn{1}{c}{\textbf{\relativeExecutionTime{}}} & \multicolumn{1}{c}{\textbf{\meanLsmCount{}}} & \multicolumn{1}{c}{\textbf{\relativeLsmCount{}}}
\\
\midrule
\cfgSmallMedium & Maxp. & 822.3 & -- &19.8 & -- &\hide{0}18.1 \\
\cfgSmallMedium & \ktc 	& 603.2 & \pct{-26.6} & 49.9 & +\pct{152} & 504.7 & +\pct{2688}\\
\cfgSmallMedium & \ektc & 622.2 & \pct{-24.3} & 64.0 & +\pct{223} & 555.2 & +\pct{2967} \\
\midrule
\cfgSmallSparse & Maxp. & 412.9 & -- & \hide{0}8.7 & -- & \hide{00}9.0 \\
\cfgSmallSparse & \ktc 	& 319.7 & \pct{-22.6} & 17.6 & +\pct{+102} & 245.3 & +\pct{2626}\\
\cfgSmallSparse & \ektc & 329.1 & \pct{-20.3} & 19.8 & +\pct{128} & 271.7 & +\pct{2919} \\
\bottomrule
\end{tabular}
\end{center}
\end{table}

\paragraph{Answer to \RQekTC}

Our answer to the first research question is positive:
\ektc is able to beat \ktc with respect to network lifetime:
For the dense and the sparse topology, the best configurations of \ektc \wrt \remainingLifetime{1}{}, \remainingLifetime{\SI{50}{\percent}}{}, and \remainingLifetime{\SI{100}{\percent}}{} outperform \ktc in terms of execution time and \LSM count, but not in terms of topology size.
Additionally, \ktc is in general only able to beat the baseline \MaxpowerTC algorithm for large \weightThreshold between \SI{40}{\meter} and \SI{80}{\meter}, while \ektc generally performs better than \MaxpowerTC.

\paragraph{Answer to \RQMinWeightPredicate}

Our answer to the second research question is that varying \weightThreshold influences the lifetime to a certain extent (positively for \ktc, negatively for \ektc), while the savings \wrt resource consumption are considerable.
A key observation is that the network lifetime for \ktc tends to increase with increasing \weightThreshold, while the network lifetime for \ektc correlates negatively with \weightThreshold

\subsection{Threats to Validity}
\label{sec:eval-threats}

We first consider threats to \emph{external validity}, \idest, the ability to generalize our findings in this evaluation to \WSN topologies and \TC algorithms in general.
To mitigate this threat, we analyze different types of topologies: sparse and dense topologies.
Our results (\Cref{sec:eval-results}) are relatively homogeneous for all considered topologies (\eg, the relative performance of the algorithms).
This increases our confidence in the validity of these findings.
To strengthen our confidence, we plan to carry out additional experiments, \eg, with different movement models, placement models and node density to obtain additional representative data as future work.
Managing the fast-growing number of possible configurations is a major challenge.

Another degree of freedom is the choice of \ktcParameterK for \ktc and \ektc.
The selection of $\ktcParameterK = \evalParK$ may be the unfortunate reason for the suboptimal performance of \ktc.
To mitigate this concern, we performed additional simulation runs with $k \in \{1.0, 1.2, 1.3, 2.0\}$ and for a larger configuration with a world size of \SI{2000}{\meter} and \numprint{1000} sensor nodes, which are not discussed here for conciseness.
The results of these preliminary experiments are comparable to the results discussed earlier.
\ktc generally tends to perform worse than \ektc \wrt network lifetime, and the performance of \ktc correlates positively with \weightThreshold.

The limitation to only one overlay application, \idest, the random-pair communication application, is another threat to external validity.
While this application reflects only one particular communication pattern, \idest, the many-to-many pattern, this pattern is relatively widespread in the communication systems domain. 
We plan to conduct experiments that cover communication patterns such as many-to-one, which is typical in data collection scenarios, or one-to-many, which is applied for data dissemination.

The last considered threat to external validity is that we model topologies as simple graphs (see \Cref{sec:metamodeling-intro}).
Modeling other types of topologies apart from \WSN topologies will probably require to support loops and parallel links.
To support parallel links and loops, we would drop the structural constraints \noLoopsConstraint and \noParallelLinksConstraint.
This implies that the refinement step (see \Cref{sec:refinement}) may result in additional application conditions because less gluings may be dropped.
Therefore, our approach is not limited to simple graphs.

A major threat to \emph{internal validity} is the reproducibility of our results.
To address this threat, we ran all experiments corresponding to each configuration (\idest, a fixed combination of world size, node count, algorithm, \ktcParameterK-parameter, and minimum-weight threshold) at least \evalParSeedCount times with  different random seeds and used the average values of these runs.
We could further increase our confidence in the internal validity of our evaluation by increasing the number of random seeds and by performing an in-depth analysis of the variance within the data of each configuration.

\section{Related Work}\label{sec:relatedWork}

In this section, we survey related work of this paper with a focus on the development of correct algorithms, model-based software development and software product lines.

\subsection{Correctness of Algorithms}
\label{sec:relatedWorkCorrectness}
One motivation of our work is the pertinent missing traceability between specification and implementation in the communication systems domain (see \shortcomingGap in \Cref{sec:introduction}).
The authors of~\cite{Martins2010} propose a calculus for \WSN protocols that is the basis for a service-oriented middleware called \emph{MufFIN}~\cite{Martins2011}.
The idea behind their approach was to prove required properties based on a specification of a \WSN protocol in terms of the calculus.
Then, an equivalent byte code representation for the MufFin middleware should be generated.
Unfortunately, it appears that the approach has not yet been fully implemented to showcase its applicability.
In \cite{Mori2013}, the authors present a domain-specific language for \WSN protocols.
While their approach is not limited to \TC algorithms, this generality makes it hard to constructively ensure required properties, which is a key objective of our work.
In \cite{Dohler2007}, the \emph{ARESA} project is presented, an alliance of industrial and academic institutions that aims to tackle \WSN research challenges jointly.
One of the key objectives is formal analysis of \WSN protocols.
In contrast to this paper, the authors lay a focus on verifying properties based on a formal specification of \WSN protocols.
In contrast, our approach is to construct \TC algorithms, \idest, specialized \WSN protocols, that are correct by construction.
In \cite{Qadir2015}, the authors identify the automatic synthesis of implementations as a central research area because the verification of protocols with large state space is still intractable today.
Finally, a detailed anecdotal example of missing traceability between specification and implementation can be found in \cite[Sec.~1]{Kluge2016}.
There, we analyze the presentation of the CTCA algorithm~\cite{XH12} and highlight the existing gaps.
This does \emph{not} mean that the evaluated implementation is incorrect.
The key message of this example is rather that it is at least difficult for the reader of~\cite{XH12} to understand the relation between the game-theoretic specification and the extensive pseudo code implementation, let alone the unpublished source code of the simulation study.

One of the key properties of the specified \TC algorithms presented in this paper is that they are guaranteed to be correct by construction.
In general, at least three major approaches exist for asserting or checking the correctness of software:
First, \emph{correct-by-construction approaches} integrate the properties during the construction of the software, which is the fundamental idea of this paper.
Second, \emph{verification-based approaches} examine an existing piece of software with respect to the required properties.
Third, \emph{testing} exercises the algorithm by executing a set of representative test cases, each consisting of input data and the expected result and, for each test case, comparing the actual result with the expected result.
In \cite{Hall2002}, the authors highlight that combining these approaches is useful in realistic projects.
Correct-by-construction approaches have been extensively studied in the context of hardware design and software development processes (\eg,~\cite{BBBCJNS11,DB05,BFMSFSW05}).
However, to the best of our knowledge, no work (apart from ours~\cite{KVS15,Kluge2016}) exists on constructing \TC algorithms using the correct-by-construction paradigm.
One technique for verifying correctness properties of software is model checking~\cite{Rensink2004}.
Many model checking approaches are only suitable for verifying properties for models of a finite size.
Graph abstractions~\cite{Baldan2008} are a formalism that alleviates this problem by introducing symbolic representations of whole classes of models.
The benefit of the correct-by-construction approach, as applied in this paper, is that it ensures correctness for models of arbitrary size because each refined \GT rule is guaranteed to preserve the specified correctness properties.
In comparison to correct-by-construction and verification-based approaches, the goal of testing is to derive a finite number of representative test cases that cover a part of the space of possible input values of a piece of software that is as large as possible~\cite{Myers2010}.

In the following, we survey formal frameworks for specifying consistency properties in the context of graph-based models.
Graphical consistency constraints (for short graph constraints), as introduced in \cite{HW95} and as used throughout this paper, express the requirement that particular combinations of nodes and edges should be present in or absent from a graph.
This formalism has been generalized later to HLR categories in~\cite{EEPT06} and extended to cope with attributes in~\cite{DV14}.
While graph constraints provide the benefit that they can be used to constructively refine \GT rules, their expressiveness is relatively limited.
For instance, global constraints such as connectivity cannot be expressed using graph constraints.
For this reason, we proved the connectivity of weakly consistent topologies in this paper rather than integrating this property constructively during the development of the \TC algorithms.
In \cite{HabelRadke2010,Radke2010}, Habel and Radke present $HR^{*}$ constraints, a new type of graph constraints that allow to express path-related properties, which may in principle also serve as input for the constructive approach.
Future work should target the question in how far $HR^{*}$ constraints are applicable in our application scenario.

In~\cite{HHS02}, the authors distinguish four situations in which a model transformation considers consistency conditions, including the preservation and enforcement of consistency constraints.
The algorithm in this paper preserves the active-link constraint, and it enforces and preserves the inactive-link constraint.

\subsection{Model-based Development of Communication Systems}
Model-based techniques have shown to be suitable to describe \cite{TGM00} and construct~\cite{JRSST09} adaptive (communication) systems.
Formal analysis of supposed properties of complex topology adaptation algorithms has already revealed special cases in which the implemented algorithms violate crucial topology constraints~\cite{Zave2012}.
In~\cite{Katelmann2008}, model checking is applied to detect bugs in the \TC algorithm LMST, leading to an improved implementation thereof.
This paper, in contrast, applies a constructive methodology~\cite{HW95} for \GT to develop correct algorithms in the first place.
In~\cite{KSSVMS14}, variants of the \TC algorithm \ktc~\cite{SWBM12} are developed using \GT, integrating the \GT tool eMoflon\footnote{\url{www.emoflon.org}} with a network simulator.
While \cite{KSSVMS14} focuses on the rapid prototyping of \TC algorithms using programmed \GT, this paper aims at devising a generic methodology to develop \TC algorithms that fulfill a set of specified constraints by construction.

In the recent years, a number of model-based tools have been proposed for developing \TC algorithms.
The \emph{Agilla} project\footnote{\url{http://mobilab.wustl.edu/projects/agilla/}}~\cite{Fok2009} provides a middleware platform that allows to use the same implementation for evaluating \TC algorithms in a simulation and in a testbed environment.
Agilla builds upon fUML, a subset of UML with formally defined semantics~\cite{Mayerhofer2012}.
In~\cite{Berardinelli2015}, the authors present an extension of Agilla that is able to collect information \wrt the energy consumption of \WSN nodes.
The \emph{ScatterClipse} project\footnote{\url{http://www.mi.fu-berlin.de/inf/groups/ag-tech/projects/Z_Finished_Projects/ScatterClipse/index.html}}~\cite{AlSaad2008} follows a similar goal and supports the \TC algorithm developer with visualization and testing facilities.
In contrast to this paper, Agilla and Scatterclipse focus on easing the development workflow of \TC algorithms.
To the best of our knowledge, integrating consistency properties constructively into this design process has not been targeted, yet.

\subsection{Tackling Variability: Software Product Lines and Self-Adaptive Systems}
\label{sec:relatedWorkVariability}
One of the major contributions of this paper is to model commonalities and differences of \TC algorithms (see \shortcomingVariability in \Cref{sec:introduction}).
We begin with a short discussion of related work that identifies variability as a research challenge in the communication systems domain.
In \cite{Anguera2010}, the authors present an evaluation of five \TC algorithms in the \emph{WISELIB} algorithm library for \WSNs.
The authors focus on design decisions related to implementing reusable \TC algorithms in WISELIB.
In this paper, however, we focus on highlighting and formalizing the commonalities and variabilities of \TC algorithms already at specification time.

Managing the commonalities and variability of (software) systems is a central topic in the software product lines (SPL) community.
In the following, we provide a short introduction to SPLs.
An SPL describes the possible configuration options of a software system, \eg, using feature models~\cite{Kang1990}, whose syntax resembles the diagram shown in \Cref{fig:tc-algorithms-feature-model}.
While feature models are typically used to model the possible configuration options of a system (often called the \emph{problem space}), metamodeling is a technique for describing abstract representations of concrete systems, the \emph{solution space}~\cite{VanDerLinden2007,Pohl2005}.
In recent years, the expressiveness of feature models has been enhanced by introducing multiplicities, which further reduces the gap between problem space and solution space~\cite{Weckesser2016,Schnabel2016,Quinton2013}.
While traditional SPLs typically describe the configuration space of a software system at or before its deployment, a dynamic software product line (DSPL) models the possible reconfiguration at runtime~\cite{Schroeter2012}.
For this purpose, the concept of binding times has been introduced to distinguish between features of a system that are bound, \eg, statically at compile time or dynamically at runtime~\cite{Burdek2014}.

Until today, there are only few contributions that connect SPLs with \GT.
For instance, in \cite{Strueber2015}, the authors propose to model families of \GT rules by merging multiple related \GT rules into one \GT rule whose variables are annotated with presence conditions, which specify for each annotated variable in which variant of the \GT rule it is present.
Based on this approach, the authors present tool support for automatically deriving variability-based \GT rules from a set of traditional \GT rules~\cite{Strueber2016} and for editing the derived variability-aware rules~\cite{Strueber2016ICGT}.
In future work, this approach should be investigated \wrt its applicability to the \WSN domain.

Several works apply SPL concepts to model adaptive communication systems.
In \cite{Bencomo2008}, the authors specify the reconfiguration space of a flood warning \WSN using SPL.
In contrast to this paper, their focus lies on modeling the different communication interfaces (\eg, WiFi or Bluetooth) of a sensor node and the conditions under which the respective interfaces should be activated.
In \cite{Ortiz2012}, the authors specify a product family of devices that act as environment monitoring and guidance system in a museum at the same time.
In contrast to this paper, their focus lies not on \TC but on modeling variability \wrt the following four dimensions: communication scope (\idest, unicast \vs anycast communication), measured metrics (\eg, humidity or temperature), actuation (\eg, visual or acoustic), and localization technology (\eg, RFID-based localization).
In \cite{Delicato2009}, the possible configuration dimensions of wireless sensor-actor network (WSAN) nodes are outlined using feature models.
The authors propose a middleware that allows to instantiate the possible configurations in a memory- and energy-efficient way on WSAN nodes.
One of the considered components reflects the topology of the WSAN nodes.
In contrast to this paper, the authors of \cite{Delicato2009} focus on surveying the typical complexity of the WSN domain;
especially, they treat the two considered \TC algorithms (flat tree \vs hierarchical tree) as a black box.
In \cite{Fuentes2010} (an extension of \cite{Delicato2009}), an SPL engineering process is presented that allows to configure resource-efficient middleware systems for WSAN nodes with dedicated tasks based on a user-selected configuration (\eg, to build vehicular area networks~\cite{Delicato2009b} or intelligent living spaces~\cite{Fuentes2011}).
The mapping between configuration space and the larger low-level solution space is performed via model transformation and code generation engine (\eg, for Java2 ME).
Finally, \cite{Portocarrero2014} presents a variability-aware reference architecture that builds on \cite{Fuentes2010}.

In fact, DSPLs are also applied to model self-adaptive systems, \idest, systems that monitor their environment, analyze the monitored data, plan appropriate measures and execute them to adapt to changing contextual environments.
In \cite{Saller2012}, a framework for precalculating possible or probable configurations for resource-constraint devices is proposed.
Such techniques are also useful in the \WSN domain because the resources of \WSN nodes are typically highly limited.
In \cite{Saller2013}, an extension to traditional feature models is proposed that integrates a dedicated submodel for the context of the modeled system.
In the context of \WSN nodes, context feature models could be used to model the varying environmental conditions to which a sensor node may react by switching to a more appropriate \TC algorithm (\eg, \MaxpowerTC in context with increased node dynamics). 
While traditional \WSNs were typically configured at deployment time using a fixed \TC algorithm (if any) and a fixed parameter set, self-adaptive \WSNs can be suitable, \eg, for environmental monitoring to detect wildfires or floods.
For instance, in \cite{Anaya2014}, the authors present a framework that may be used, \eg, for reconfiguring parameters of \WSN nodes in a wildfire detection network.
In their case study, the framework is able to predict critical situations (\eg, imminent wildfires) and react appropriately by increasing the sampling rate.
Use cases such as flood or wildfire detection show that switching between \TC algorithms is a sensible use case and should be one of our future lines of research.

\section{Conclusion}
\label{sec:conclusion}

In this paper, we proposed a model-driven methodology for designing families of \TC algorithms using \GT and graph constraints.
The motivation of our research is that the state-of-the-art development of \TC algorithms for \WSNs exposes two major shortcomings:
\begin{inparaenum}
\item 
While it is common to prove formal properties of the designed algorithms and to evaluate them extensively using simulators (and less frequently using hardware testbeds), it is often hard or even impossible to verify that the formal specification and the simulator (or testbed) implementation indeed represent the same algorithm (\shortcomingGap).
\item 
Little effort is made to constructively reuse common substructures (\eg, structural constraints or tie breakers) that appear in many \TC algorithms (\shortcomingVariability).
\end{inparaenum}

In~\cite{KVS15}, we focused on \shortcomingGap by proposing a systematic approach for developing individual correct-by-construction \TC algorithms.
In this approach, valid topologies are characterized using graph constraints~\cite{HW95}, \TC algorithms are specified using programmed \GT~\cite{FNTZ98}, and an existing constructive approach~\cite{HW95} is applied to enrich \GT rules with application conditions derived from the graph constraints.
We illustrated the applicability of the proposed approach by re-engineering the \TC algorithm \ktc~\cite{SWBM12}.

In this paper, we complement our work on \shortcomingGap and tackle \shortcomingVariability by generalizing the results in~\cite{KVS15} as follows:
\begin{inparaenum}
\item 
We separate the (graph) constraints into common and algorithm-specific parts.
We show the applicability of this step by specifying six existing \TC algorithms as well as \ektc, a novel variant of \ktc.
\item 
We adjust and extend the steps of the constructive approach proposed in~\cite{KVS15} to be also applicable to families of \TC algorithms.
\end{inparaenum}
Finally, we present tool support that allows to immediately evaluate the \TC algorithms, specified using the \GT tool \eMoflon, in the \Simonstrator network simulation environment.

Thanks to the proposed approach, it is now possible to rapidly specify and evaluate new \TC algorithms.
First, after specifying an appropriate algorithm-specific predicate, all proves of formal properties carry over to the new \TC algorithm.
Second, the tool integration between \eMoflon and \Simonstrator mirrors the hierarchical, compositional structure of common and algorithm-specific predicates.
This means that only the algorithm-specific predicate of a new \TC algorithm needs to be implemented in \eMoflon---all other components may be reused immediately.

\paragraph{Outlook}
As future work, we aim at extending the proposed systematic approach to support the entire typical development workflow of \TC algorithms, consisting of specification, simulation and testbed evaluation (\Cref{fig:OverviewOfEntireApproach}).
\def\implMarker{\tikz[baseline=(char.base)]{\node[diamond,draw,inner sep=.1em,text width={width("S")},align=center,text=white,fill=black] (char){\textsf{I}};}}
\def\specMarker{\tikz[baseline=(char.base)]{\node[diamond,draw,inner sep=.1em,text width={width("S")},align=center,text=white,fill=black] (char){\textsf{S}};}}
\begin{figure}
\begin{center}
\includegraphics[width=.9\textwidth]{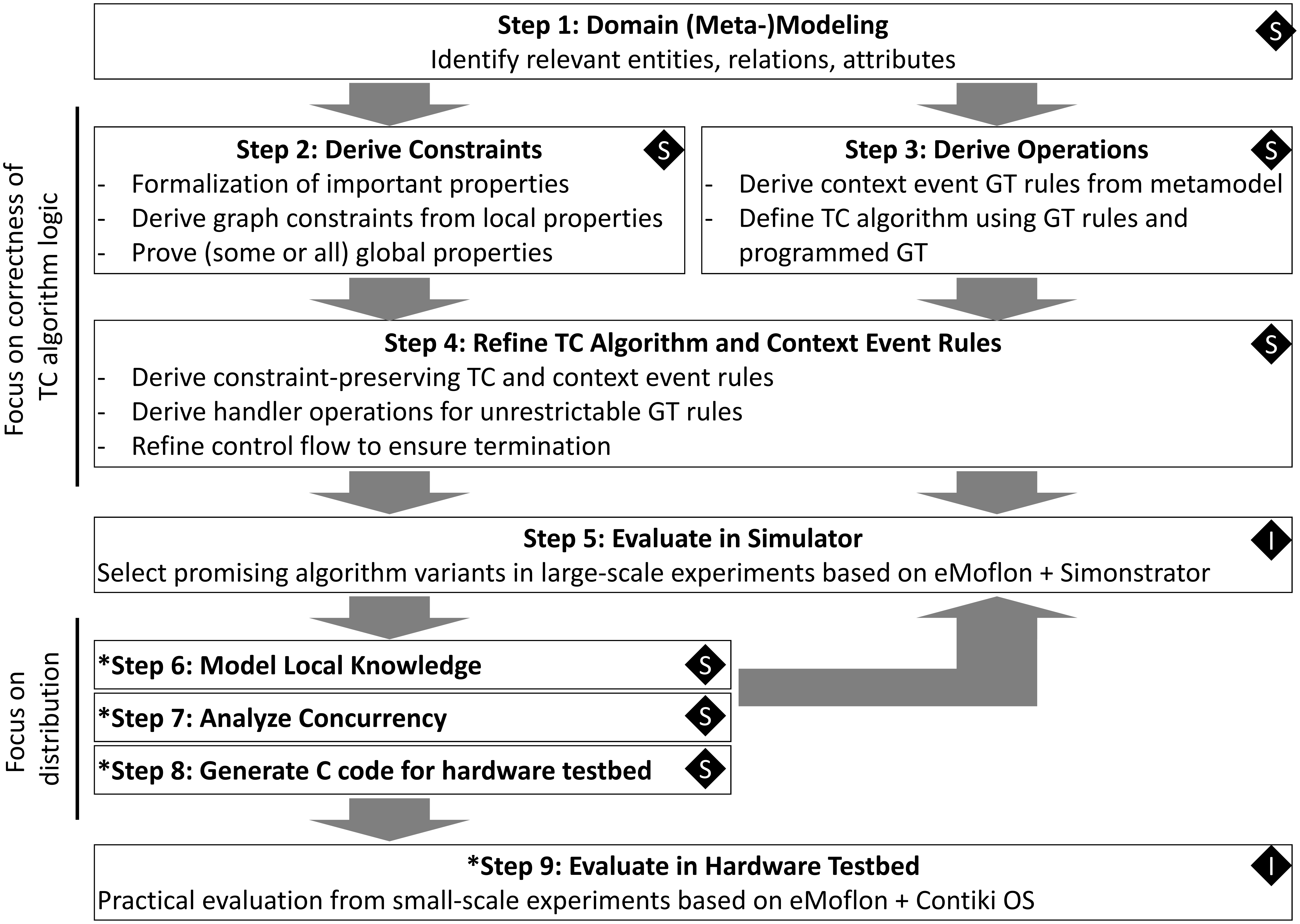}
\caption[Overview of our systematic approach]{Overview of systematic approach (\textbf{\textsf{*}} marks future work; \implMarker/\specMarker: Step in implementation/specification phase)}
\label{fig:OverviewOfEntireApproach}
\end{center}
\end{figure}
In \cite{KVS15,Kluge2016}, we have focused on the step-by-step specification and simulation-based evaluation of \TC algorithms, which is shortly recapitulated here to lead over to the planned future work.
For simplicity, we focus on the development of a single \TC algorithm, even though the steps carry over to families of \TC algorithms as well.
In Step 1, we model relevant entities, relations, and attributes of the considered class of topologies (\eg, node energy level and link weight).
In Step 2, we formalize the important properties of the considered \TC algorithms in terms of first-order logic predicates.
Properties that can be checked based on local knowledge are additionally transformed into graph constraints.
For instance, the requirement that the topology is completely classified upon termination of the \TC algorithm is translated into the \unclassifiedLinkConstraintLong.
Properties that can only be checked globally are proved manually (\eg, A-connectivity of the output topology).
In Step 3, we derive possible \TC actions and context events from the metamodel and formalize them using \TC and \CE \GT rules, respectively.
For instance, link-weight modifications are represented by the \weightModificationRuleLong.
Additionally, we use programmed \GT---in this case Story-Driven Modeling~\cite{FNTZ98}---to specify the control flow of the \TC algorithm.
In Step 4, we combine the graph constraints from the second step and the \GT rules from the third step to obtain refined \GT rules that preserve all graph constraints.
For all unrestrictable \GT rules, we transform the generated application conditions into handler operations.
At the end of the forth step, we analyze whether the generated application conditions of the \TC rules may lead to a non-termination of the \TC algorithm and refine its control flow to ensure termination.
In Step 5, the \TC algorithm is evaluated in a network simulator.
For this paper, we decided to use an integration-based approach because the source code of the simulator is available.
We could have generated code for the simulator as well, but this is typically a larger effort compared to a tool integration.
To support the full development process of \TC algorithms, we will address the following research questions in our future work.
\begin{itemize}
\item
Step 6: How to model information about local knowledge?
Currently, several of the patterns require knowledge about the 3-hop neighborhood of a node, which may not be available on real sensor nodes (due to memory limitations).
\item 
Step 7: How can we analyze concurrency issues due to the parallel execution of the \TC rules on real hardware?
Inside the simulator, \TC is executed sequentially, while, on real hardware nodes, the \TC algorithm is executed concurrently on each node.
We will have to analyze the \GT-based specification for race conditions and other concurrency issues and resolve them, \eg, by implementing conflict resolution strategies.
\item 
Step 8: How can we generate efficient code for the hardware testbed?
We plan to add a second set of code generation templates to the \GT tool \eMoflon to be able to generate embedded C code.
Since testbed devices are typically highly resource-constraint, we will focus on optimizing the runtime and the memory footprint (of the generated code and at runtime) in this step.
\end{itemize}
We plan to use the \Simonstrator~\cite{RSRS15} platform to evaluate the localized, parallel \GT-based specification, which results from answering the first two questions, and we plan to use the Contiki operating system~\cite{Dunkels2004} as target platform for generating embedded C code (Step 9).

\begin{acknowledgements}
The authors would like to thank Lukas Neumann for his contributions to the evaluation study.
\end{acknowledgements}

\bibliographystyle{spmpsci}
\bibliography{rkluge}

\end{document}